%% file: photonic_sapt.tex
\documentclass[a4paper,10pt,numbers=noenddot,twoside=on,headinclude=true,footinclude=false,headsepline=true]{scrartcl}

\usepackage[utf8]{inputenc}
\usepackage[english]{babel}
\usepackage{color}
\usepackage{amssymb}
\usepackage{amsmath}
\usepackage{amsfonts}
\usepackage{amsopn}
\usepackage{amsbsy}
\usepackage{enumerate}
\usepackage[plainpages=false,pdfpagelabels,breaklinks=true,pdfborderstyle={/S/U/W 1.2}]{hyperref}
\usepackage{graphics}
\numberwithin{equation}{section} 
\numberwithin{table}{section} 
\numberwithin{figure}{section} 
\usepackage[amsmath,thmmarks,thref,hyperref]{ntheorem}

\theoremstyle{plain}
\newtheorem{theorem}{Theorem}[section]
\newtheorem{definition}[theorem]{Definition}
\newtheorem{lemma}[theorem]{Lemma}
\newtheorem{corollary}[theorem]{Corollary}
\newtheorem{proposition}[theorem]{Proposition}
\newtheorem{assumption}[theorem]{Assumption}

\theorembodyfont{\upshape}
\newtheorem{remark}[theorem]{Remark}

\theoremstyle{nonumberplain}

\theoremsymbol{\ensuremath{_\Box}}
\newtheorem{proof}{Proof}
\usepackage[style=alphabetic,
	citestyle=alphabetic,
	firstinits=true,
	backend=biber,
	hyperref=auto,
	sorting=nyt,
	maxcitenames=3,
	maxbibnames=99,
	minnames=3,
	arxiv=pdf,
	url=false,
	isbn=false,
	dateabbrev=true,
	datezeros=true,
	bibencoding=utf8]{biblatex}
\newbibmacro{string+doi}[1]{%
	\iffieldundef{doi}{#1}{\href{http://dx.doi.org/\thefield{doi}}{#1}}}
\DeclareFieldFormat
	{title}{\usebibmacro{string+doi}{\mkbibemph{#1}}\isdot}
\DeclareFieldFormat
	[article,inbook,incollection,inproceedings,patent,thesis,unpublished]
	{title}{\usebibmacro{string+doi}{\mkbibemph{#1}}\isdot}
\DeclareFieldFormat
	[article,inbook,incollection,inproceedings,patent,thesis,unpublished]
	{pages}{#1}
\DeclareFieldFormat
	[article]
	{journaltitle}{#1\isdot}
\DeclareFieldFormat
	[article]
	{volume}{\textbf{#1}}
\DefineBibliographyStrings{english}{pages={}}
\DeclareBibliographyDriver{article}{%
	\usebibmacro{bibindex}%
	\usebibmacro{begentry}%
	\usebibmacro{author/translator+others}%
	\setunit{\labelnamepunct}\newblock
	\usebibmacro{title}%
	\usebibmacro{journal+issuetitle}%
	\setunit*{\addcomma\space}
	\printfield{pages}
	\setunit*{\addcomma\space}
	\printfield{year}
	\usebibmacro{finentry}%
}

\newbibmacro*{journal+issuetitle}{%
	\usebibmacro{journal}%
	\setunit*{\addspace}%
	\iffieldundef{series}
	{}
	{\newunit
		\printfield{series}%
		\setunit{\addspace}}%
	\iffieldundef{volume}
		{}
		{\printfield{volume}%
		\setunit{\addspace}}%
	\setunit*{\addcolon\space}%
	\usebibmacro{issue}%
	\newunit}
\usepackage[scaled]{beramono}
\usepackage{helvet}
\usepackage[charter]{mathdesign} % math font has ugly \mathcals
\SetMathAlphabet{\mathcal}{normal}{OMS}{cmsy}{m}{n} % fixes ugly \mathcals
\SetMathAlphabet{\mathcal}{bold}{OMS}{cmsy}{m}{n} % fixes ugly \mathcals
\providecommand{\ie}{i.~e.~}
\providecommand{\eg}{e.~g.~}
\providecommand{\cf}{cf.~}

\providecommand{\R}{\mathbb{R}}

\providecommand{\C}{\mathbb{C}}
\renewcommand{\C}{\mathbb{C}}

\providecommand{\T}{\mathbb{T}}
\renewcommand{\T}{\mathbb{T}}
\providecommand{\N}{\mathbb{N}}
\providecommand{\Z}{\mathbb{Z}}
\providecommand{\ii}{\mathrm{i}}
\providecommand{\e}{\mathrm{e}}
\renewcommand{\Re}{\mathrm{Re} \,}
\renewcommand{\Im}{\mathrm{Im} \,}
\providecommand{\Hil}{\mathcal{H}}
\providecommand{\eps}{\varepsilon}
\providecommand{\Cont}{\mathcal{C}}

\providecommand{\ran}{\mathrm{ran} \, }

\providecommand{\supp}{\mathrm{supp} \,}

\providecommand{\trace}{\mathrm{Tr} \,}

\providecommand{\dd}{\mathrm{d}}
\providecommand{\id}{\mathrm{id}}

\providecommand{\order}{\mathcal{O}}
\providecommand{\Fourier}{\mathcal{F}}

\providecommand{\trace}{\mathrm{Tr}}

\providecommand{\abs}[1]{\left \lvert #1 \right \rvert}
\providecommand{\sabs}[1]{\lvert #1 \vert}

\providecommand{\norm}[1]{\left \lVert #1 \right \rVert}
\providecommand{\snorm}[1]{\lVert #1 \rVert}
\providecommand{\bnorm}[1]{\bigl \lVert #1 \bigr \rVert}
\providecommand{\Bnorm}[1]{\Bigl \lVert #1 \Bigr \rVert}
\providecommand{\scpro}[2]{\left \langle #1 , #2 \right \rangle}
\providecommand{\sscpro}[2]{\langle #1 , #2 \rangle}
\providecommand{\bscpro}[2]{\bigl \langle #1 , #2 \bigr \rangle}

\providecommand{\sopro}[2]{\vert #1 \rangle \langle #2 \vert}

\providecommand{\ncint}{\mathrel{{\ooalign{$\int$\cr\kern+.07em\raise.15ex\hbox{$\pmb{\scriptstyle-}$}\cr}}}}           \providecommand{\ncpartial}{\mathrel{{\ooalign{$\partial$\cr\kern+.29em\raise.79ex\hbox{$\pmb{\scriptstyle-}$}\cr}}}}
\bibliography{./bibliography}
\renewcommand{\Hil}{\mathfrak{H}}

\providecommand{\Zak}{\mathcal{Z}}

\providecommand{\Msymb}{\mathcal{M}}
\providecommand{\Op}{\mathfrak{Op}}
\providecommand{\Weyl}{\sharp}

\providecommand{\Maxwell}{\mathbf{M}}

\providecommand{\Hper}{\Hil_0}
\providecommand{\Jper}{J_0}

\providecommand{\Mper}{\Maxwell_0}

\providecommand{\HperT}{\mathfrak{h}_0}

\providecommand{\BZ}{\mathbb{B}}
\providecommand{\WS}{\mathbb{W}}

\providecommand{\Jphys}{\mathbf{J}}
\providecommand{\Gphys}{\mathbf{G}}
\providecommand{\Mphys}{\mathbf{M}}
\providecommand{\Pphys}{\mathbf{P}}
\providecommand{\Qphys}{\mathbf{Q}}

\providecommand{\Hoer}[1]{S^{#1}_{\rho}}
\providecommand{\Hoerm}[2]{S^{#1}_{#2}}
\providecommand{\Hoereq}[1]{S^{#1}_{\rho,\mathrm{eq}}}
\providecommand{\Hoermeq}[2]{S^{#1}_{#2,\mathrm{eq}}}
\providecommand{\Hoermper}[2]{S^{#1}_{#2,\mathrm{per}}}

\providecommand{\SemiHoerm}[2]{A S^{#1}_{#2}}
\providecommand{\SemiHoereq}[1]{A S^{#1}_{\rho,\mathrm{eq}}}
\providecommand{\SemiHoermeq}[2]{A S^{#1}_{#2,\mathrm{eq}}}
\providecommand{\SemiHoerper}[1]{A S^{#1}_{\rho,\mathrm{per}}}
\providecommand{\SemiHoermper}[2]{A S^{#1}_{#2,\mathrm{per}}}

\providecommand{\eff}{\mathrm{eff}}
\providecommand{\piref}{\pi_{\mathrm{ref}}}

\providecommand{\blochf}{\varphi}
\providecommand{\blochb}{\mathcal{E}_{\BZ}}

\providecommand{\eff}{\mathrm{eff}}

\providecommand{\specrel}{\sigma_{\mathrm{rel}}}

\providecommand{\Rot}{\mathbf{Rot}}

\providecommand{\Index}{\mathcal{I}}

\providecommand{\Weyl}{\sharp}

\providecommand{\domain}{\mathfrak{D}}
\providecommand{\domainT}{\mathfrak{d}}
\providecommand{\Piref}{\Pi_{\mathrm{ref}}}

\providecommand{\Poynt}{\mathcal{S}}
\providecommand{\freqband}{\omega}
\bibliography{bibliography}
\usepackage{units}
\usepackage{diagxy}
\title{Effective Light Dynamics in \\ Perturbed Photonic Crystals}
\author{Giuseppe De Nittis${}^{\ast}$ \& Max Lein${}^{\star}$}

\begin{document}

\maketitle
\vspace{-9mm}
\begin{center}
	$^{\ast}$ Department Mathematik, Universität Erlangen-Nürnberg \linebreak
	Cauerstrasse 11, D-91058 Erlangen, Germany \linebreak
	{\footnotesize \href{mailto:denittis@math.fau.de}{\texttt{denittis@math.fau.de}}}
	\medskip
	\\
	$^{\star}$ University of Toronto \& Fields Institute, Department of Mathematics \linebreak
	Bahen Centre, 40 St{.} George Street, Toronto, ON M5S 2E4, Canada \linebreak
	{\footnotesize \href{mailto:max.lein@me.com}{\texttt{max.lein@me.com}}}
\end{center}
\begin{abstract}
	In this work, we rigorously derive \emph{effective dynamics} for light from within a limited frequency range propagating in a photonic crystal that is modulated on the macroscopic level; the perturbation parameter $\lambda \ll 1$ quantifies the separation of spatial scales. We do that by rewriting the dynamical Maxwell equations as a Schrödinger-type equation and adapting space-adiabatic perturbation theory. Just like in the case of the Bloch electron, we obtain a simpler, effective Maxwell operator for states from within a relevant almost invariant subspace. A correct physical interpretation for the effective dynamics requires to establish two additional facts about the almost invariant subspace: (1) The source-free condition has to be verified and (2) it has to support real states. The second point also forces one to consider a \emph{multi}band problem even in the simplest possible setting; This turns out to be a major difficulty for the extension of semiclassical methods to the domain of photonic crystals.
\end{abstract}

\newpage
\tableofcontents

%%% begin content %%% (fold)
\input{section_01}
\input{section_02}
\input{section_03}
\input{section_04}
\input{section_05}
\input{section_06}
\input{section_07}
\input{section_08}
% \input{appendix}
%%% end content %%% (end)

\printbibliography

\end{document}

%% file: section_01.tex
%!TEX root = /Users/max/Dropbox/research/photonic crystals/sapt for isotropic photonic crystals/arxiv/v2/photonic sapt.tex
\section{Introduction} % (fold)
\label{intro}
Photonic crystals are to the transport of light (electromagnetic waves) what crystalline solids are to the transport of electrons \cite{Joannopoulos_Johnson_Winn_Meade:photonic_crystals:2008}. This analogy extends to the mathematical formulation, because the source-free Maxwell equations 
\begin{align}
	&\partial_t \mathbf{E}(t) = + \eps^{-1} \nabla_x \times \mathbf{H}(t) 
	, 
	&&
	\partial_t \mathbf{H}(t) = - \mu^{-1} \nabla_x \times \mathbf{E}(t)
	, 
	\label{intro:eqn:traditional_Maxwell_equations_dynamics} 
	\\
	&\nabla_x \cdot \eps \mathbf{E}(t) = 0 
	, 
	&&
	\nabla_x \cdot \mu \mathbf{H}(t) = 0 
	, 
	\label{intro:eqn:traditional_Maxwell_equations_no_source}
\end{align}
can alternatively be written as a Schrödinger-type equation 
\begin{align}
	\ii \frac{\partial}{\partial t} \left (
	\begin{matrix}
		\mathbf{E}(t) \\
		\mathbf{H}(t) \\
	\end{matrix}
	\right ) &= \Mphys_{(\eps,\mu)} \left (
	\begin{matrix}
		\mathbf{E}(t) \\
		\mathbf{H}(t) \\
	\end{matrix}
	\right )
	, 
	&&
	\left (
		\begin{matrix}
			\mathbf{E}(0) \\
			\mathbf{H}(0) \\
		\end{matrix}
		\right ) = \left (
	\begin{matrix}
		\mathbf{E}_0 \\
		\mathbf{H}_0 \\
	\end{matrix}
	\right ) 
	\in \Hil_{(\eps,\mu)}
	,
	\label{intro:eqn:Maxwell_Schroedinger_type}
\end{align}
where the electric and magnetic component of the initial electromagnetic field configuration $(\mathbf{E}_0,\mathbf{H}_0)$ need to take values in $\R^6$ and satisfy the source-free condition \eqref{intro:eqn:traditional_Maxwell_equations_no_source}. Here, the tensorial quantities $\eps$ and $\mu$ are the material weights \emph{electric permittivity} and \emph{magnetic permeability} which enter into the definition of the Maxwell operator 
\begin{align}
	\Mphys_{(\eps,\mu)} := \left (
	\begin{matrix}
		0 & + \ii \, \eps^{-1} \, \nabla_x^{\times} \\
		- \ii \, \mu^{-1} \, \nabla_x^{\times} & 0 \\
	\end{matrix}
	\right ) 
	\label{intro:eqn:Maxwell_operator}
\end{align}
($\nabla_x^{\times}$ is the curl of vector fields on $\R^3$) and the Hilbert space 
\begin{align*}
	\Hil_{(\eps,\mu)} := L^2_{\eps}(\R^3,\C^3) \oplus L^2_{\mu}(\R^3,\C^3)
\end{align*}
consisting of weighted vector-valued $L^2$-functions $\Psi = (\psi^E,\psi^H)$ with scalar product 
\begin{align}
	\bscpro{\Psi}{\Phi}_{\Hil_{(\eps,\mu)}} := \int_{\R^3} \dd x \, \psi^E(x) \cdot \eps(x) \, \phi^E(x) + \int_{\R^3} \dd x \, \psi^H(x) \cdot \mu(x) \, \phi^H(x) 
	\label{intro:eqn:weighted_scalar_product}
\end{align}
where $a \cdot b := \sum_{j = 1}^3 \bar{a}_j \, b_j$ is the usual scalar product on $\C^3$. Even though $\Hil_{(\eps,\mu)}$ coincides with $L^2(\R^3,\C^6)$ as Banach spaces due to our assumptions on $\eps$ and $\mu$ (Assumption~\ref{intro:assumption:eps_mu}~(i)), the norm and scalar product depend on the weights. The square of the weighted $L^2$-norm is twice the energy of the electromagnetic field, 
\begin{align*}
	\mathcal{E} \bigl ( \mathbf{E}(t) , \mathbf{H}(t) \bigr ) &= \tfrac{1}{2} \bnorm{\bigl ( \mathbf{E}(t) , \mathbf{H}(t) \bigr )}_{\Hil_{(\eps,\mu)}}^2 
	% \\
	% &
	% = \frac{1}{2} \int_{\R^3} \dd x \, \eps(x) \, \babs{\mathbf{E}(t)}^2 + \frac{1}{2} \int_{\R^3} \dd x \, \mu(x) \, \babs{\mathbf{H}(t)}^2 
	\, , 
\end{align*}
and thus, the selfadjointness of $\Mphys_{(\eps,\mu)}$ with respect to the scalar product $\scpro{\cdot \,}{\cdot}_{\Hil_{(\eps,\mu)}}$ translates to conservation of field energy. 

We are interested in the case where the perturbed material weights 
\begin{align}
	\eps_{\lambda}(x) := \frac{\eps(x)}{\tau_{\eps}(\lambda x)^2} 
	, 
	&&
	\mu_{\lambda}(x) := \frac{\mu(x)}{\tau_{\mu}(\lambda x)^2} 
	\label{intro:eqn:slow_modulation_material_constants}
\end{align}
are slow modulations of the photonic crystal's $\Gamma$-periodic weights $\eps$ and $\mu$ where $\Gamma \cong \Z^3$ is the crystal lattice. The small dimensionless parameter $\lambda \ll 1$ relates the scale of the lattice to the scale of the \emph{macroscopic} modulation. 
The scalar \emph{modulation functions} $\tau_{\eps}$ and $\tau_{\mu}$ describe an isotropic perturbation of the matrix-valued material weights $\eps$ and $\mu$ which may be induced by uneven thermal \cite{Duendar_et_al:optothermal_tuning_photonic_crystals:2011,van_Driel_et_al:tunable_2d_photonic_crystals:2004} or strain tuning \cite{Wong_et_al:strain_tunable_photonic_crystals:2004}. Perturbations of this type have also been considered in several theoretical works \cite{Onoda_Murakami_Nagaosa:geometrics_optical_wave-packets:2006,Raghu_Haldane:quantum_Hall_effect_photonic_crystals:2008,Esposito_Gerace:photonic_crystals_broken_TR_symmetry:2013}. 

For this special case, let us denote the perturbed Maxwell operator associated to the weights $(\eps_{\lambda} , \mu_{\lambda})$ with 
\begin{align*}
	\Mphys_{\lambda} := \Mphys_{(\eps_{\lambda},\mu_{\lambda})}
\end{align*}
and the corresponding weighted $L^2$-space with $\Hil_{\lambda} := \Hil_{(\eps_{\lambda},\mu_{\lambda})}$. For more details, we refer to \textbf{Section~\ref{setup}}. 
\medskip

\noindent
Our goal is to \emph{derive effective electrodynamics $\e^{- \ii t \Mphys_{\eff}}$ which approximate full electrodynamics $\e^{- \ii t \Mphys_{\lambda}}$ for initial states supported in a narrow range of frequencies.} The main tool is space-adiabatic perturbation theory which has been used successfuly to solve the analogous problem for perturbed periodic Schrödinger operators \cite{PST:effective_dynamics_Bloch:2003,DeNittis_Lein:Bloch_electron:2009}. 
\medskip

\noindent
The strategy to tackle this problem is best explained with the help of a diagram: 
\begin{align}
	\bfig
		\node cHlambda(-800,450)[\mbox{physical}]
		\node Hlambda(-800,0)[\Hil_{\lambda}]
		\node piHlambda(-800,-600)[\pmb{\Pi}_{\lambda} \Hil_{\lambda}]
		\node cZakHlambda(0,450)[\mbox{BFZ}]
		\node ZakHlambda(0,0)[\Zak \Hil_{\lambda}]
		\node piZakHlambda(0,-600)[\pmb{\Pi}_{\lambda}^{\Zak} \Zak \Hil_{\lambda}]
		\node cZakH0(800,450)[\mbox{auxiliary}]
		\node ZakH0(800,0)[\Zak \Hper]
		\node piH0(800,-600)[\Pi_{\lambda} \Zak \Hper]
		\node cHref(1600,450)[\mbox{reference}]
		\node Href(1600,0)[\Zak \Hper]
		\node piHref(1600,-600)[\Piref \Zak \Hper]
		\arrow[Hlambda`ZakHlambda;\Zak]
		\arrow[ZakHlambda`ZakH0;S(\ii \lambda \nabla_k)]
		\arrow[ZakH0`Href;U_{\lambda}]
		\arrow[Hlambda`piHlambda;\pmb{\Pi}_{\lambda}]
		\arrow[ZakHlambda`piZakHlambda;\pmb{\Pi}_{\lambda}^{\Zak}]
		\arrow[ZakH0`piH0;\Pi_{\lambda}]
		\arrow[Href`piHref;\Piref]
		\arrow/-->/[piHlambda`piZakHlambda;]
		\arrow/-->/[piZakHlambda`piH0;]
		\arrow/-->/[piH0`piHref;]
		\arrow//[cHlambda`Hlambda;]
		\arrow//[cZakHlambda`ZakHlambda;]
		\arrow//[cZakH0`ZakH0;]
		\arrow//[cHref`Href;]
		\Loop(-800,0){\Hil_{\lambda}}(ur,ul)_{\e^{-\ii t \Mphys_{\lambda}}} 
		\Loop(0,0){\Zak \Hil_{\lambda}}(ur,ul)_{\e^{-\ii t \Mphys_{\lambda}^{\Zak}}} 
		\Loop(800,0){\Zak \Hper}(ur,ul)_{\e^{-\ii t M_{\lambda}}} 
		\Loop(1600,-600){\Piref \Zak \Hper}(dr,dl)^{\e^{-\ii t \Mphys_{\eff}}} 
	\efig
	\label{intro:eqn:spaces_diagram}
\end{align}
Apart from the \emph{physical representation}, there is the \emph{Bloch-Floquet-Zak (or BFZ) representation}, the \emph{auxiliary representation} and the \emph{reference representation}. The loops indicate unitary time evolution groups. All operators mapping from left to right are unitaries between Hilbert spaces. Any operator mapping an upper to a lower Hilbert space is an orthogonal projection; The ranges of these projections are made up of states from a narrow band of frequencies (we will be more specific in \textbf{Section~\ref{setup:periodic:pi0}}). Operators in different columns are related by conjugating with the corresponding unitary (at least up to $\order(\lambda^{\infty})$ errors in norm), \eg 
\begin{align*}
	\pmb{\Pi}_{\lambda}^{\Zak} = S^{-1}(\ii \lambda \nabla_k) \, \Pi_{\lambda} \, S(\ii \lambda \nabla_k)
	. 
\end{align*}
\begin{remark}
	To indicate a clear relationship between operators in different representations (\eg the projection or the Maxwell operator), we use bold-face letters for operators in the physical representation, we add the superscript $\Zak$ in BFZ representation, and in the auxiliary representation, we use normal letters without superscript. 
\end{remark}
Now to the strategy of the proof of our first main result, Theorem~\ref{real_eff:thm:eff_edyn}: The first change of representation via $\Zak$ exploits the $\Gamma$-periodicity of the unperturbed Maxwell operator (\cf \textbf{Section~\ref{setup:Zak}}). In this BFZ representation, the Maxwell operator $\Mphys_{\lambda}^{\Zak} = \Op_{\lambda}(\pmb{\Msymb}_{\lambda})$ is a pseudodifferential operator (\cf \textbf{Section~\ref{setup:PsiDOs}} and \cite{DeNittis_Lein:adiabatic_periodic_Maxwell_PsiDO:2013}), a necessity to implement space-adiabatic perturbation theory \cite{PST:effective_dynamics_Bloch:2003}. 

However, for technical reasons we need to change representation once more before we adapt \cite{PST:effective_dynamics_Bloch:2003}: in the physical representation, the Hilbert structure of $\Hil_{\lambda}$ depends on $\lambda$ so that comparing operators for different values of $\lambda$ is not straightforward. Thus, we use $S(\ii \lambda \nabla_k)$ to represent everything on the common, $\lambda$-\emph{in}dependent Hilbert space $\Zak \Hper$ of the periodic Maxwell operator (\cf \textbf{Section~\ref{setup:aux_rep}}). 

The $\lambda$-independence of the norms allows us to adapt space-adiabatic perturbation theory in \textbf{Section~\ref{sapt}} to the perturbed Maxwell operator. We construct the superadiabatic projection $\Pi_{\lambda}$, the intertwining projection $U_{\lambda}$ and the effective Maxwellian $\Op_{\lambda}(\Msymb_{\eff})$ order-by-order in $\lambda$. 
In the simplest case where $\Omega_{\mathrm{rel}} = \bigcup_{n \in \Index} \{ \omega_n \}$ consists of a family of non-intersecting bands and each associated line bundle is trivial, we obtain explicit formulas for the symbol $\Msymb_{\eff} = \Msymb_{\eff,0} + \lambda \, \Msymb_{\eff,1} + \order(\lambda^2)$: The perturbation enters the principal symbol 
\begin{align}
	\Msymb_{\eff,0}(r,k) = \tau(r) \, \sum_{n \in \Index} \omega_n(k) \, \sopro{\chi_n}{\chi_n} 
	\label{intro:eqn:Meff_0_simple}
\end{align}
through the \emph{deformation function} $\tau(r) := \tau_{\eps}(r) \, \tau_{\mu}(r)$. 
The subprincipal symbol in this simplified situation can be expressed in terms of the \emph{Berry connection} $\mathcal{A} := \bigl ( \mathcal{A}_{nj} \bigr )$ with real-valued entries $\mathcal{A}_{nj}(k) := \ii \, \bscpro{\varphi_n(k)}{\nabla_k \varphi_j(k)}_{\HperT}$ and the \emph{Poynting tensor} $\Poynt = \bigl ( \Poynt_{nj} \bigr )$ with entries 
\begin{align}
	\Poynt_{nj}(k) := \int_{\T^3} \dd y \; \overline{\varphi_n^H(k,y)} \times \varphi_j^E(k,y) 
	\label{intro:eqn:Poynting_vector}
\end{align}
for all $n,j \in \Index$, namely 
\begin{align}
	\Msymb_{\eff,1}(r,k) &= \frac{1}{2} \sum_{n,j \in \Index} \Bigl ( \ii \, \tau(r) \, \bigl ( \nabla_r \ln \tfrac{\tau_{\eps}}{\tau_{\mu}} \bigr )(r) \cdot \bigl ( \Poynt_{nj}(k) - \overline{\Poynt_{jn}(k)} \bigr ) 
	\Bigr . + \notag \\
	&\qquad \qquad \quad \Bigl . 
	+ \bigl ( \omega_n(k) + \omega_j(k) \bigr ) \, \nabla_r \tau \cdot \mathcal{A}_{nj} 
	\Bigr ) \, \sopro{\chi_n}{\chi_j} 
	. 
	\label{intro:eqn:Meff_1_simple}
\end{align}
In any case, $\e^{- \ii t \, \Op_{\lambda}(\Msymb_{\eff})}$ approximates the full dynamics in the reference representation for states from the narrow band of frequencies supported in $\Omega_{\mathrm{rel}}$: the composition of the three unitaries $V_{\lambda} := U_{\lambda} \, S(\ii \lambda \nabla_k) \, \Zak$ (which change from the physical to the reference representation) intertwines the two evolution groups up to $\order(\lambda^{\infty})$ in operator norm, 
\begin{align}
	\left ( \e^{- \ii t \Mphys_{\lambda}} - V_{\lambda}^* \; \e^{- \ii t \Op_{\lambda}(\Msymb_{\eff})} \; V_{\lambda} \right ) \, \pmb{\Pi}_{\lambda} &= \order_{\norm{\cdot}}(\lambda^{\infty}) 
	. 
	\label{intro:eqn:approximate_dynamics}
\end{align}
Compared to the quantum problem, we also need to prove that the superadiabatic subspace $\ran \pmb{\Pi}_{\lambda}$ \emph{supports physical states} in the following precise sense: up to errors of order $\order(\lambda^{\infty})$ in norm, its elements need to satisfy the source-free conditions \eqref{intro:eqn:traditional_Maxwell_equations_no_source} \emph{and} it needs to support \emph{real}-valued electromagnetic fields (\cf Definition~\ref{physical:defn:support_real}). The first condition is always satisfied (Proposition~\ref{physical:prop:divergence_cond_Pi_lambda}) while the second requires 
two additional assumptions: $\eps$ and $\mu$ need to be real and one needs to choose bands symmetrically (\cf Corollary~\ref{physical:cor:spectral_condition_reality} and Proposition~\ref{physical:prop:real_solutions_Pi_lambda}). 
This spectral condition has a simple explanation: the Bloch waves which enter $\pmb{\Pi}_{\lambda}$ are complex-valued functions, and thus we need to look at linear combinations of counter-propagating Bloch waves to piece together real-valued fields $(\mathbf{E}_0,\mathbf{H}_0)$. 

Thus, we can derive effective dynamics for physically relevant states (\textbf{Section~\ref{real_eff}}): the spectral condition in our first main result, Theorem~\ref{real_eff:thm:eff_edyn}, also implies that effective single-band dynamics (\eg \emph{one}-band ray optics) are \emph{not sufficient} to describe the dynamics of physical sates. This is because the superadiabatic subspaces associated to single, isolated bands cannot support real states (Corollary~\ref{real_eff:cor:no_eff_single_band_dynamics}). Moreover, we will establish a link between $\Op_{\lambda}(\Msymb_{\eff})$ and the \emph{Maxwell-Harper operator}~\eqref{twin_band:eqn:Harper_Maxwell_operator}; this operator can be seen as a multiband generalization of the usual Harper operator with a symmetry. 
% CHANGED smoothen sentence
Hence, also the Maxwell-Harper operator is affiliated to a non-commutative torus (of dimension $6$). 

Although the setting suggests that a semiclassical limit is within reach, we will argue in \textbf{Section~\ref{twin_bands}} why obtaining ray optics equations for \emph{real} states is not as easy as modifying existing results (\eg \cite{PST:effective_dynamics_Bloch:2003,DeNittis_Lein:Bloch_electron:2009,Stiepan_Teufel:semiclassics_op_valued_symbols:2012}). 
Even in the simplest case, the \emph{twin band case} for non-gyrotropic photonic crystals, one needs to consider an isolated, non-degenerate band $\omega_+(k)$ together with its symmetric twin $\omega_-(k) = - \omega_+(-k)$. 
% CHANGED smoothen this sentence 
Consequently, $\Msymb_{\eff,0}$ \emph{cannot} be scalar, a case which is not covered by existing semiclassical techniques. Hence, we postpone a rigorous treatment to a future publication \cite{DeNittis_Lein:ray_optics_photonic_crystals:2013}. 

Finally, we will put our main results in perspective with existing results in \textbf{Section~\ref{discussion}} and outline avenues for future research. The proofs of some of the statements have been relegated to \textbf{Section~\ref{technical}}.

\subsection*{Acknowledgements} % (fold)
\label{intro:acknowledgements}
The foundation of this article was laid during the trimester program “Mathematical challenges of materials science and condensed matter physics”, and the authors thank the Hausdorff Research Institute for Mathematics for providing a stimulating research environment. Moreover, G.~D{.} gratefully acknowledges support by the Alexander von Humboldt Foundation. 
M.~L{.} is supported by Deutscher Akademischer Austauschdienst. 
The authors also appreciate the useful comments and references provided by C.~Sparber. Moreover, M.~L{.} thanks M. Sigal for the stimulating discussions surrounding ray optics and meaning of observables in electrodynamics. 
% subsection Acknowledgements (end)
% section Introduction (end)

%% file: section_02.tex
%!TEX root = /Users/max/Dropbox/research/photonic crystals/sapt for isotropic photonic crystals/arxiv/v2/photonic sapt.tex
\section{Setup} % (fold)
\label{setup}
This section serves to introduce the main objects of study and give some of their properties. For more details and proofs, we point to \cite{DeNittis_Lein:adiabatic_periodic_Maxwell_PsiDO:2013} and references therein.

\subsection{The slowly-modulated Maxwell operator} % (fold)
\label{setup:Maxwell}
Throughout this paper, we shall always make the following 
\begin{assumption}[Material weights]\label{intro:assumption:eps_mu}
	Let the dielectric permittivity $\eps_{\lambda}$ and magnetic permeability $\mu_{\lambda}$ be hermitian-matrix-valued functions of the form \eqref{intro:eqn:slow_modulation_material_constants} with the following properties: 
	\begin{enumerate}[(i)]
		\item $\eps , \mu \in L^{\infty} \bigl ( \R^3 , \mathrm{Mat}_{\C}(3) \bigr )$ are positive, $\Gamma$-periodic and bounded away from $0$ and $+ \infty$, \ie there exist $c , C > 0$ such that $0 < c \, \id_{\C^3} \leq \eps , \mu \leq C \, \id_{\C^3} < + \infty$. 
		\item $\tau_{\eps} , \tau_{\mu} \in \Cont^{\infty}_{\mathrm{b}}(\R^3)$ are positive, $\tau_{\eps}(0) = 1 = \tau_{\mu}(0)$ and bounded away from $0$ and $+\infty$. 
	\end{enumerate}
\end{assumption}
\begin{remark}
	These assumptions ensure that $\Hil_{\lambda}$ and $L^2(\R^3,\C^6)$ coincide as Banach spaces, a fact that will be crucial in many technical arguments. 
	Physically, we include all isotropically perturbed photonic cyrstals which can be treated as lossless, linear media. 
\end{remark}
In some of our theorems, we need to assume that the material weights are \emph{real} to ensure we can construct real solutions; this means that these results do not apply to gyrotropic photonic crystals where $\eps$ and $\mu$ have non-zero imaginary offdiagonal entries. 
\begin{assumption}[Real weights]\label{setup:assumption:reality_weights}
	Suppose the material weights satisfy Assumption~\ref{intro:assumption:eps_mu}. We call them \emph{real} iff the entries of $\eps$ and $\mu$ are all \emph{real}-valued functions. 
\end{assumption}
The slowly modulated Maxwell operator 
\begin{align}
	\Mphys_{\lambda} = S(\lambda \hat{x})^{-2} \, \Mper 
	\label{setup:eqn:modulated_Maxwell}
\end{align}
is naturally defined on the weighted $L^2$-space $\Hil_{\lambda}$ and most conveniently written in terms of the modulation 
\begin{align}
	S(\lambda \hat{x}) := \left (
	\begin{matrix}
		\tau_{\eps}^{-1}(\lambda \hat{x}) & 0 \\
		0 & \tau_{\mu}^{-1}(\lambda \hat{x}) \\
	\end{matrix}
	\right )
	\label{setup:eqn:modulation_S}
\end{align}
and the \emph{periodic Maxwell operator} 
\begin{align*}
	\Mper := W \, \Rot 
	:= \left (
	\begin{matrix}
		\eps^{-1}(\hat{x}) & 0 \\
		0 & \mu^{-1}(\hat{x}) \\
	\end{matrix}
	\right ) \, \left (
	\begin{matrix}
		0 & + \ii \nabla_x^{\times} \\
		- \ii \nabla_x^{\times} & 0 \\
	\end{matrix}
	\right ) 
	. 
\end{align*}
We will consistently use $v^{\times} \psi := v \times \psi$ to associate a matrix to any vectorial quantity $v = (v_1 , v_2 , v_3)$ for suitable vectors $\psi = (\psi_1 , \psi_2 , \psi_3)$ through 
\begin{align*}
	v^{\times} = \left (
	\begin{matrix}
		0 & - v_3 & + v_2 \\
		+ v_3 & 0 & - v_1 \\
		- v_2 & + v_1 & 0 \\
	\end{matrix}
	\right )
	. 
\end{align*}
Equipped with the weight-independent domain $\domain := \domain(\Rot)$, the Maxwell operator $\Mphys_{\lambda}$ defines a selfadjoint operator on $\Hil_{\lambda}$ \cite[Theorem~2.2]{DeNittis_Lein:adiabatic_periodic_Maxwell_PsiDO:2013}. 

The splitting \eqref{setup:eqn:modulated_Maxwell} also allows us to decompose 
\begin{align}
	\Hil_{\lambda} = \Jphys_{\lambda} \oplus \Gphys 
\end{align}
into the subspace of unphysical zero mode fields 
\begin{align*}
	\Gphys := \Bigl \{ \Psi = \bigl ( \nabla_x \varphi_1 , \nabla_x \varphi_2 \bigr ) \in L^2(\R^3,\C^6) \; \; \big \vert \; \; \varphi_1,\varphi_2 \in L^2_{\mathrm{loc}}(\R^3) \Bigr \}
\end{align*}
and its $\scpro{\cdot \,}{\cdot}_{\Hil_{\lambda}}$-orthogonal complement
\begin{align*}
	\Jphys_{\lambda} := \Gphys^{\perp_{\Hil_{\lambda}}} 
\end{align*}
comprised of the \emph{physical states} (\cf equations~(2.2)--(2.4) in \cite{DeNittis_Lein:adiabatic_periodic_Maxwell_PsiDO:2013}) which satisfy the source-free conditions~\eqref{intro:eqn:traditional_Maxwell_equations_no_source} in the distributional sense. We will denote the orthogonal projections onto $\Jphys_{\lambda}$ and $\Gphys$ with $\Pphys_{\lambda}$ and $\Qphys_{\lambda}$, respectively. The Maxwell operator 
\begin{align*}
	\Mphys_{\lambda} = \Mphys_{\lambda} \vert_{\Jphys_{\lambda}} \oplus 0 \vert_{\Gphys}
\end{align*}
is block-diagonal with respect to this decomposition, and thus, dynamics $\e^{- \ii t \Mphys_{\lambda}}$ leaves $\Jphys_{\lambda}$ and $\Gphys$ invariant \cite[Theorem~2.2]{DeNittis_Lein:adiabatic_periodic_Maxwell_PsiDO:2013}. 
% subsection The slowly-modulated Maxwell operator (end)

\subsection{BFZ representation} % (fold)
\label{setup:Zak}
In order to exploit the $\Gamma$-periodicity of the unperturbed, periodic Maxwell operator $\Mper$, we employ a variant of the Bloch-Floquet transform called Zak transform \cite[Section~3]{DeNittis_Lein:adiabatic_periodic_Maxwell_PsiDO:2013}. Here, the crystal lattice $\Gamma = \mathrm{span}_{\Z} \{ e_1 , e_2 , e_3 \}$ is spanned by three (non-unique) basis vectors and allows to decompose vectors $x = y + \gamma$ in real space $\R^3 \cong \WS \times \Gamma$ into a component $y$ from the so-called \emph{Wigner-Seitz} or fundamental cell $\WS$ and a lattice vector $\gamma \in \Gamma$. We will identify $\WS$ with the $3$-dimensional torus $\T^3$ for convenience. 

This decomposition of real space induces a similar splitting of momentum space $\R^3 \cong \BZ \times \Gamma^*$ where $\BZ$ is the fundamental cell associated to the dual lattice $\Gamma^* := \mathrm{span}_{\Z} \{ e_1^* , e_2^* , e_3^* \}$ generated by the family of vectors defined through $e_j \cdot e_n^* = 2 \pi \, \delta_{jn}$, $j , n = 1 , 2 , 3$. The most common choice of fundamental cell 
\begin{align*}
	\BZ := \Bigl \{ \mbox{$\sum_{j = 1}^3$} k_j \, e_j^* \in \R^3 \; \big \vert \; k_1 , k_2 , k_3 \in [-\nicefrac{1}{2},+\nicefrac{1}{2}) \Bigr \} 
\end{align*}
is dubbed (first) \emph{Brillouin zone}, and its elements $k \in \BZ$ are known as crystal momentum. 

The Zak transform is now initially defined on $\mathcal{S}(\R^3,\C^6)$ as 
\begin{align*}
	( \Zak \Psi )(k,y) := \sum_{\gamma \in \Gamma} \e^{- \ii k \cdot (y + \gamma)} \, \Psi(y + \gamma) 
	, 
\end{align*}
and extends to the Banach space $\Hil_{\lambda} \cong L^2(\R^3,\C^6)$ by density. 
From the definition, we can read off the following periodicity properties of 
Zak transformed $L^2(\R^3,\C^6)$-functions: 
\begin{align*}
	( \Zak \Psi )(k , y - \gamma) &= ( \Zak \Psi )(k , y) 
	&& 
	\forall \gamma \in \Gamma 
	\\
	( \Zak \Psi )(k - \gamma^* , y) &= \e^{+ \ii \gamma^* \cdot y} ( \Zak \Psi )(k , y) 
	&& 
	\forall \gamma^* \in \Gamma^* 
\end{align*}
For $\lambda = 0$, it is well known that $\Zak$ extends to a unitary map 
\begin{align*}
	\Zak : \Hper \longrightarrow L^2_{\mathrm{eq}}(\R^3,\HperT) 
	\cong L^2(\BZ) \otimes \HperT
\end{align*}
between $\Hper$ and the $L^2$-space of equivariant functions in $k$ with values in 
\begin{align*}
	\HperT := L^2_{\eps}(\T^3,\C^3) \oplus L^2_{\mu}(\T^3,\C^3) 
	\, , 
\end{align*}
namely, 
\begin{align}
	L^2_{\mathrm{eq}} (\R^3,\HperT) := \Bigl \{ 
	\Psi \in L^2_{\mathrm{loc}} (\R^3,\HperT) \; \; \big \vert \; \; 
	\Psi(k - \gamma^*) = \e^{+ \ii \gamma^* \cdot \hat{y}} \Psi(k) \; \mbox{ a.~e. $\forall \gamma^* \in \Gamma^*$}
	\Bigr \} 
	\, . 
\end{align}
Here $\e^{+ \ii \gamma^* \cdot \hat{y}}$ denotes the multiplication operator by the function $y\mapsto \e^{+ \ii \gamma^* \cdot {y}}$ on each fiber space $\HperT$. These spaces are equipped with the weighted scalar products 
\begin{align*}
	\bscpro{\Psi(k)}{\Phi(k)}_{\HperT} := \int_{\T^3} \dd y \, \psi^E(k,y) \cdot \eps(y) \, \phi^E(k,y) + \int_{\T^3} \dd y \, \psi^H(k,y) \cdot \mu(y) \, \phi^H(k,y) 
\end{align*}
and 
\begin{align*}
	\scpro{\Psi}{\Phi}_{\mathrm{eq}} &:= \int_{\BZ} \dd k \, \bscpro{\Psi(k)}{\Phi(k)}_{\HperT} 
	. 
\end{align*}
For $\lambda \neq 0$ the scalar product $\scpro{\cdot \,}{\cdot}_{\Hil_{\lambda}}$ does \emph{not} necessarily decompose fiberwise in a similar fashion. However, if we endow the space $\Zak \Hil_{\lambda} \cong L^2_{\mathrm{eq}} ( \R^3 , \HperT )$ with the induced scalar product 
\begin{align*}
	\sscpro{\cdot \,}{\cdot}_{\Zak \Hil_{\lambda}} := \bscpro{\Zak^{-1} \, \cdot \,}{\Zak^{-1} \, \cdot \,}_{\Hil_{\lambda}}
	= \bscpro{S(\ii\lambda\nabla_k) \; \cdot \,}{S(\ii\lambda\nabla_k) \; \cdot \,}_{\mathrm{eq}} 
	, 
\end{align*}
the map $\Zak : \Hil_{\lambda} \longrightarrow \Zak \Hil_{\lambda}$ is again unitary. 
\medskip

\noindent
In BFZ representation, the Maxwell operator reads 
\begin{align*}
	\Mphys_{\lambda}^{\Zak} := \Zak \, \Mphys_{\lambda} \, \Zak^{-1} 
	= S(\ii \lambda \nabla_k)^{-2} \, \Mper^{\Zak} 
\end{align*}
where the periodic Maxwell operator 
\begin{align*}
	\Mper^{\Zak} :& \negmedspace= \Zak \, \Mper \, \Zak^{-1} 
	= \int_{\BZ}^{\oplus} \dd k \, \Mper(k)
\end{align*}
fibers in $k$, 
\begin{align*}
	\Mper(k) &= \left (
	\begin{matrix}
		0 & - \eps^{-1} \, (- \ii \nabla_y + k)^{\times} \\
		+ \mu^{-1} \, (- \ii \nabla_y + k)^{\times} & 0 \\
	\end{matrix}
	\right ) 
	, 
\end{align*}
but $\Mphys_{\lambda}^{\Zak}$ does not. The operator-valued function 
\begin{align*}
	\R^3 \ni k \mapsto \Mper(k) \in \mathcal{B}(\domainT,\HperT)
\end{align*}
is analytic \cite[Proposition~3.3]{DeNittis_Lein:adiabatic_periodic_Maxwell_PsiDO:2013}, and the domain $\domainT = \domain \bigl ( \Mper(k) \bigr )$ of each fiber operator of $\Mper(k)$ is independent of $k$. 
% subsection Zak transform (end)

\subsection{The periodic Maxwell operator} % (fold)
\label{setup:periodic}
Since some of the properties of periodic Maxwell operators differ from those of periodic Schrödinger operators, we will give a small overview. Apart from the unphysical zero modes which contribute essential spectrum at $0$, the spectrum of $\Mper(k)$ due to the physical states in $\Jper(k) \subset \HperT$ is purely discrete \cite[Theorem~3.4]{DeNittis_Lein:adiabatic_periodic_Maxwell_PsiDO:2013}. Here $\Jper$ is the subspace of states satisfying \eqref{intro:eqn:traditional_Maxwell_equations_no_source} in the distributional sense and the Zak transform fibers it, 
\begin{figure}[t]
	\centering
		\resizebox{100mm}{!}{\includegraphics{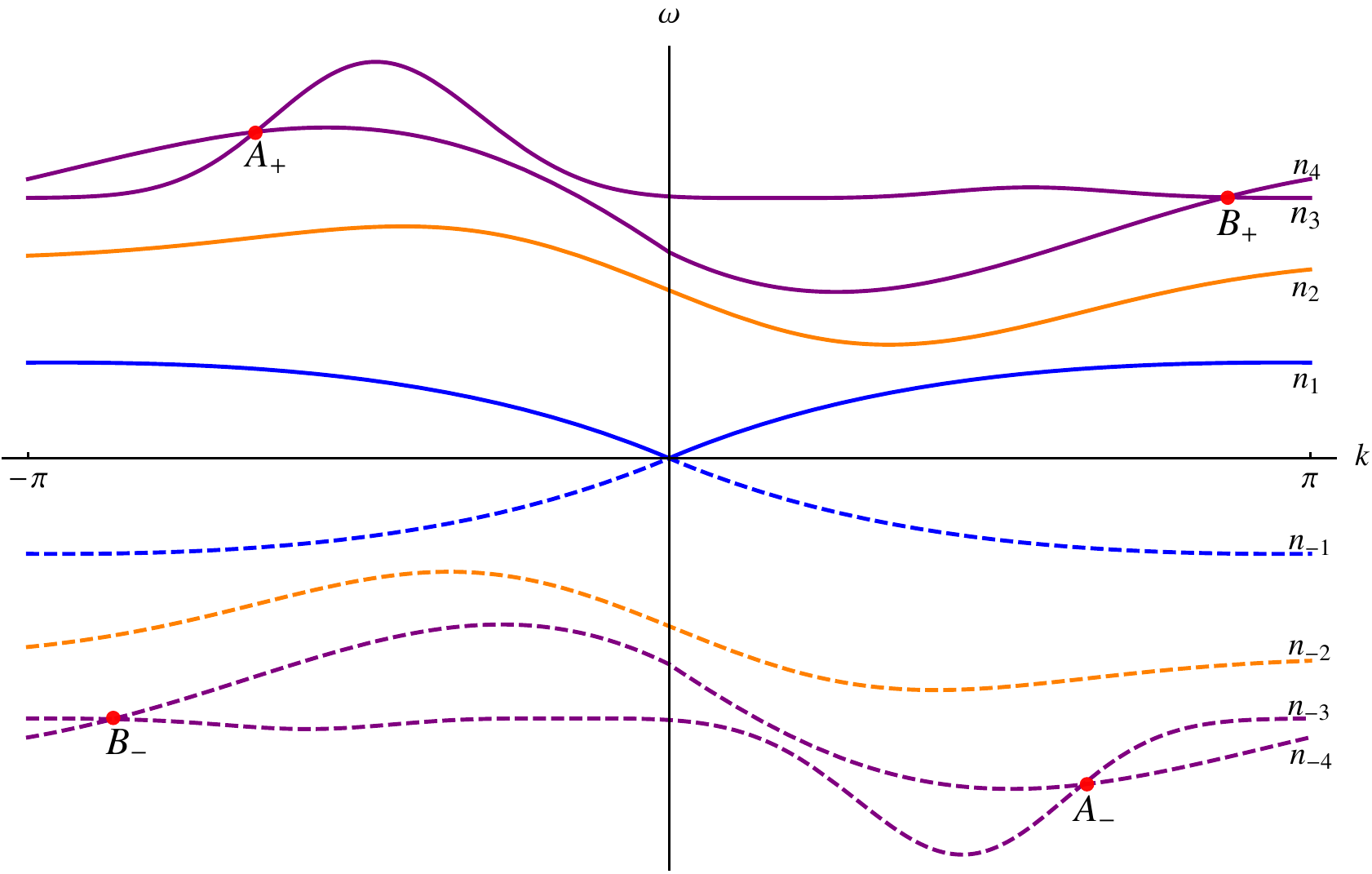}} 
	\caption{A sketch of a typical band spectrum of $\Mper(k) \vert_{\Jper(k)}$ for a non-gyrotropic photonic crystal (\ie $\eps$ and $\mu$ are real). The $2+2$ ground state bands with linear dispersion around $k = 0$ are blue. Positive frequency bands are drawn using solid lines while the lines for the symmetrically-related negative frequency bands are in the same color, but dashed. }
	\label{setup:fig:band_picture}
\end{figure}
\begin{align*}
	\Zak \Jper = \int_{\BZ}^{\oplus} \dd k \, \Jper(k) 
	. 
\end{align*}
If we number the eigenvalues appropriately, we obtain a band picture similar to that for periodic Schrödinger operators \cite[Theorem~1.4]{DeNittis_Lein:adiabatic_periodic_Maxwell_PsiDO:2013}. A schematic representation is given in Figure~\ref{setup:fig:band_picture}. First of all, the analyticity of $k \mapsto \Mper(k)$ and the discrete nature of the spectrum away from $0$ means that all band functions $k \mapsto \omega_n(k)$ are analytic away from band crossings; due to the unitary equivalence of $\Mper(k)$ and $\Mper(k - \gamma^*)$, the band functions are necessarily $\Gamma^*$-periodic. Also the corresponding eigenfunctions $\varphi_n(k)$ (\emph{Bloch functions}) which satisfy 
\begin{align*}
	\Mper(k) \varphi_n(k) = \omega_n(k) \, \varphi_n(k) 
	, 
	&&
	\varphi_n(k) \in \domainT 
	, 
\end{align*}
can be chosen to be locally analytic. 

The first difference to periodic Schrödinger operators we notice is that $\Mper(k)$ is not bounded from below. For the particular case we are interested later on where $\eps$ and $\mu$ are real, one can see that \emph{frequency bands come in pairs}: if we denote complex conjugation on $\HperT$ with $C$, then $C$ commutes with $W$ and we obtain 
\begin{align}
	C \, \Mper(k) \, C = - \Mper(-k) 
	. 
	\label{setup:eqn:complex_symmetry_fiber_Maxwell}
\end{align}
Hence, any frequency band $k \mapsto \omega_n(k)$ with Bloch function $k \mapsto \varphi_n(k)$ has a “symmetric twin” $k \mapsto - \omega_n(-k)$ with Bloch function $k \mapsto \overline{\varphi_n(-k)}$, 
\begin{align}
	\Mper(k) \, \overline{\varphi_n(-k)} = - \omega_n(-k) \, \overline{\varphi_n(-k)}
	. 
	\label{setup:eqn:symmetric_twin_Bloch_function}
\end{align}
On the level of spectra, we can write this succinctly as $\sigma \bigl ( \Mper(-k) \bigr ) = - \sigma \bigl ( \Mper(k) \bigr )$. This symmetry will be discussed in more detail in Section~\ref{physical}. 
\medskip

\noindent
A second feature that is unique to periodic Maxwell operators is the appearance of so-called \emph{ground state bands} (blue bands in Figure~\ref{setup:fig:band_picture}); this is independent of whether $\eps$ and $\mu$ are real. By definition these are the bands $\omega_n$ which approach $0$ as $k \rightarrow 0$. We have proven that there are $2 + 2$ of them, \ie $2$ approach $\omega = 0$ from above, the other $2$ from below. In a vicinity of $k = 0$, these bands have approximately linear dispersion (with a conical intersection), and the slope depends on the direction of $k$ \cite[Theorem~1.4~(iii)]{DeNittis_Lein:adiabatic_periodic_Maxwell_PsiDO:2013}. Given material weights, the slope can be computed from the Fourier coefficients of $\eps$ and $\mu$ as well as elementary calculations involving $2 \times 2$ and $3 \times 3$ matrices. What makes these bands special is that these are the only ones which enter when studying the transport of light whose in-vacuo wave length is long compared to the lattice spacing of the photonic crystal (usually referred to as homogenization limit). Deriving effective light dynamics for ground state bands presents us with an additional technical challenge: the dimensionality of the ground state eigenspace at $k \neq 0$ and $k = 0$ jumps from $4$ to $6$ \cite[Lemma~3.2~(iii)]{DeNittis_Lein:adiabatic_periodic_Maxwell_PsiDO:2013}, and thus, the associated ground state projection is necessarily discontinuous at $k = 0$. Even though we believe we could overcome this technical obstacle by means of a suitable regularization procedure (\cf discussion in \cite[Section~3.2]{DeNittis_Lein:adiabatic_periodic_Maxwell_PsiDO:2013}), the \emph{physical} situation in the low frequency regime is different: The dispersion of the ground state bands near $k = 0$ is approximately linear which means that the wave length $\lambda \propto \nicefrac{1}{\abs{k}}$ for $k \approx 0$ is \emph{never} small compared to the scale on which $\tau_{\eps}$ and $\tau_{\mu}$ vary -- the multiscale ansatz breaks down. This is the reason why we decide to \emph{exclude} the ground state bands from our considerations at this point.

\subsubsection{States associated to a narrow range of frequencies} % (fold)
\label{setup:periodic:pi0}
Typically, photonic crystals only have local as opposed to global spectral gaps. For instance, the yellow ($n_{\pm 2}$) or violet bands ($n_{\pm 3}$ and $n_{\pm 4}$) in Figure~\ref{setup:fig:band_picture} are only \emph{locally} separated by a gap. Mathematically, this translates to the following: 
\begin{assumption}[Gap Condition]\label{intro:assumption:gap_condition}
	A family of frequency bands $\Omega_{\mathrm{rel}} := \{ \omega_n \}_{n \in \Index}$ of $\Mper^{\Zak}$ associated to the index set $\Index \subset \Z$ satisfies the \emph{Gap Condition} iff 
	\begin{enumerate}[(i)]
		\item $\Omega_{\mathrm{rel}}$ splits into finitely many \emph{contiguous} families of bands, \ie 
		\begin{align*}
			 \Omega_{\mathrm{rel}} = \bigcup_{\alpha = 1}^N \Omega_{\alpha}
		\end{align*}
		where each $\Omega_{\alpha} := \{ \omega_n \}_{n \in \Index_{\alpha}} $ is defined in terms of \emph{contiguous} index sets $\Index_{\alpha} = [ n_{\min \, \alpha} , n_{\max \, \alpha} ] \cap \Z$, 
		\item all the $\Omega_{\alpha}$ are locally separated by a gap from all the other bands, 
		\begin{align*}
			\inf_{k \in \BZ} \mathrm{dist} \biggl ( \bigcup_{n \in \Index_{\alpha}} \bigl \{ \omega_n(k) \bigr \} , \bigcup_{j \not\in \Index_{\alpha}} \bigl \{ \omega_j(k) \bigr \} \biggr ) 
			=: c_{\mathrm{g}}(\alpha)
			> 0 
			, 
		\end{align*}
		\item and none of these bands are ground state bands in the sense of \cite[Definition~3.6]{DeNittis_Lein:adiabatic_periodic_Maxwell_PsiDO:2013}, \ie $\inf_{n \in \Index} \, \sabs{\omega_n(0)} > 0$. 
	\end{enumerate}
	The infimum $c_{\mathrm{g}} := \inf_{\alpha = 1 , \ldots , N} c_{\mathrm{g}}(\alpha)$ is called \emph{gap constant}. 
\end{assumption}
In contrast to \cite[Definition~1]{PST:effective_dynamics_Bloch:2003}, our definition of Gap Condition admits several contiguous families of bands; 
this modification is necessary for photonic crystals, because in case of real $(\eps,\mu)$ only symmetrically chosen families of bands (corresponding to solid \emph{and} dashed lines of the same color) support real states (\cf Section~\ref{physical:real}). 

One of the most basic results is that a local gap suffices to make the “local spectral projection” $\pi_0(k) := 1_{\{ \omega_n(k)\}_{n \in \Index}} \bigl ( \Mper(k) \bigr )$ associated to isolated families of bands $\Omega_{\mathrm{rel}}$ analytic in $k$. It can also be usefully expressed in terms of the Bloch functions,
\begin{align}
        \pi_0(k) = \sum_{n \in \Index} \sopro{\varphi_n(k)}{\varphi_n(k)}
        \, .
        % \hat{\pi}_0 := \int^\oplus_\BZ \dd k \; \left ( \sum_{n \in \Index} \sopro{\varphi_n(k)}{\varphi_n(k)} \right )
        % \hat{\pi}_0 := \int^\oplus_\BZ \dd k \; 1_{\{ \omega_n(k)\}_{n \in \Index}} \bigl ( \Mper(k) \bigr )
        % \pi_0(k) = 1_{\{ \omega_n(k)\}_{n \in \Index}} \bigl ( \Mper(k) \bigr )
        % \hat{\pi}_0 = \int^\oplus_\BZ \dd k \, \pi_0(k)
        % := \int^\oplus_\BZ \dd k \; 1_{\{ \omega_n(k)\}_{n \in \Index}} \bigl ( \Mper(k) \bigr )
		\label{physical:eqn:pi_0}
\end{align}
%
% Later on, we will see that $\hat{\pi}_0 = \int_{\BZ}^{\oplus} \dd k \, \pi_0(k)$ can also be seen as a pseudodifferential operator (\cf Lemma~\ref{setup:lem:pi_0_PsiDO}).
%
\begin{lemma}\label{setup:lem:existence_pi_0}
	Suppose Assumption~\ref{intro:assumption:eps_mu} and the Gap Condition~\ref{intro:assumption:gap_condition} are satisfied. Then the map $\BZ \ni k \mapsto \pi_0(k)$ is analytic and takes values in the orthogonal projections on $\HperT$.
\end{lemma}
We omit the proof since it is quite standard. It rests on locally writing the projection as a sum of Cauchy--Dunford integrals, and then using the analyticity of the resolvent in $k$  \cite[Proposition~3.3]{DeNittis_Lein:adiabatic_periodic_Maxwell_PsiDO:2013} to deduce the analyticity of $\pi_0$. 
% subsubsection States associated to a narrow range of frequencies (end)

\subsubsection{The Bloch bundle} % (fold)
\label{setup:periodic:Bloch_bundle}
The analyticity of the map $k \mapsto \pi_0(k)$ means that the projection $\hat{\pi}_0 = \int^{\oplus}_{\BZ} \dd k \, \pi_0(k)$ gives rise to an analytic vector bundle, the so-called \emph{Bloch bundle}. The analysis of this particular vector bundle has helped gain an understanding of topological and analytic aspects of periodic Schrödinger operators \cite{Nenciu:exponential_loc_Wannier:1983,Panati:triviality_Bloch_bundle:2006,Kuchment:exponential_decaying_wannier:2009,DeNittis_Lein:exponentially_loc_Wannier:2011}, and it will prove useful to understand periodic Maxwell operators as well. 

One starts introducing the vector bundle $p : \mathcal{E}_{\R^3}(\pi_0) \longrightarrow {\R^3}$ with total space 
\begin{align*}
	\mathcal{E}_{\R^3}(\pi_0) := \bigsqcup_{k \in \R^3} \ran \pi_0(k) 
\end{align*}
and bundle projection $p^{-1}(k)=\ran \pi_0(k)$. This is a hermitian vector bundle which is (topologically) trivial since the base space $\R^3$ is contractible \cite[Lemma 1.4.3]{Atiyah:K_theory:1994}. Due to the Oka principle \cite[Satz~I, p.~258]{Grauert:analytische_Faserungen:1958}, we need not distinguish between topological and analytic triviality. This allows immediately for the following result:
\begin{lemma}
	Assume $\Omega_{\mathrm{rel}} = \{ \omega_n \}_{n \in \Index}$ satisfies the Gap Condition~\ref{intro:assumption:gap_condition}. Then there exists an orthonormal family of functions $\bigl \{ \R^3 \ni k \mapsto \psi_j(k) \in \HperT \; \; \vert \; \; j = 1 , \ldots , \sabs{\Index} \bigr \} $ so that 
	\begin{align*}
		\mathrm{span}_{\C} \, \left \{ \psi_1(k) , \ldots , \psi_{\sabs{\Index}}(k) \right \} = \ran \pi_0(k)
	\end{align*}
	and all the $\psi_j$ are analytic. 
\end{lemma}
The Bloch bundle, which is the relevant vector bundle here, is obtained from the trivial vector bundle $p:\mathcal{E}_{\R^3}(\pi_0) \longrightarrow {\R^3}$ by a standard quotient procedure: the group $\Gamma^{\ast}$ acts freely on the base space $\R^3$ by translation and for all $\gamma^{\ast}\in\Gamma^{\ast}$ we have a vector space isomorphism $\ran \pi_0(k) \longrightarrow \ran \pi_0(k-\gamma^{\ast})$ induced by the equivariance condition $\pi_0(k-\gamma^{\ast}) = \e^{+ \ii \gamma^* \cdot \hat{y}} \, \pi_0(k) \; \e^{- \ii \gamma^* \cdot \hat{y}}$. This $\Gamma^{\ast}$-action on the fibers is compatible with the $\Gamma^{\ast}$-action on the base space $\R^3$ via the bundle projection $p$. All this endows $p : \mathcal{E}_{\R^3}(\pi_0) \longrightarrow \R^3$ with the structure of a $\Gamma^{\ast}$-vector bundle in the terminology of \cite[Section~1.6]{Atiyah:K_theory:1994}. Now, since the action of $\Gamma^{\ast}$ is free on the base, we have according to \cite[Proposition~1.6.1]{Atiyah:K_theory:1994} a unique \emph{quotient vector bundle} with total space $\mathcal{E}_{\R^3}(\pi_0)/\Gamma^{\ast}$ and base space $\R^3/\Gamma^{\ast}$. With a little abuse of notation, we can identify the Brillouin zone $\BZ$ with the torus $\R^3/\Gamma^{\ast}$, and thus obtain the \emph{Bloch bundle} $p : \blochb(\pi_0) \longrightarrow \BZ$ where the total space is the quotient

\begin{align*}
	\blochb(\pi_0) := \mathcal{E}_{\R^3}(\pi_0) / \Gamma^{\ast}= \biggl ( \bigsqcup_{k \in \R^3} \ran \pi_0(k) \biggr ) / \Gamma^{\ast} 
	. 
\end{align*}
The Bloch bundle is an analytic hermitian vector bundle \emph{which is trivial if and only if the first Chern class vanishes}. For more details we refer to \cite{Panati:triviality_Bloch_bundle:2006,DeNittis_Lein:exponentially_loc_Wannier:2011}. 
% subsubsection The Bloch bundle (end)
% subsection The periodic Maxwell operator (end)

\subsection{Pseudodifferential operators on weighted $L^2$-spaces} % (fold)
\label{setup:PsiDOs}
The main point of \cite{DeNittis_Lein:adiabatic_periodic_Maxwell_PsiDO:2013} was to explain how $\Mphys_{\lambda}^{\Zak} = \Op_{\lambda}(\pmb{\Msymb}_{\lambda})$ can be understood as the pseudodifferential operator associated to the semiclassical symbol 
\begin{align}
	\pmb{\mathcal{M}}_{\lambda}(r,k) &= \bigl ( S^{-2} \Weyl \Mper(\, \cdot \,) \bigr )(r,k)
	\label{setup:eqn:symbol_Mphys}
	\\
	&= S^{-2}(r) \; \Mper(k) + \lambda \, W \, \left (
	\begin{matrix}
		0 & - \ii \, \tau_{\eps}(r) \, \bigl ( \nabla_r \tau_{\eps} \bigr )^{\times}(r) \\
		+ \ii \, \tau_{\mu}(r) \, \bigl ( \nabla_r \tau_{\mu} \bigr )^{\times}(r) & 0 \\
	\end{matrix}
	\right )
	.
	\notag 
\end{align}
Formally, $\Op_{\lambda}(f)$ for a $\mathcal{B} \bigl ( L^2(\T^3,\C^6) \bigr )$-valued function is defined in the usual way as 
\begin{align*}
	\Op_{\lambda}(f) := \frac{1}{(2\pi)^3} \int_{\R^3} \dd r' \int_{\R^3} \dd k' \; (\Fourier_{\sigma} f)(r',k') \; \e^{- \ii (k' \cdot (\ii \lambda \nabla_k) - r' \cdot \hat{k})} 
\end{align*}
where the symplectic Fourier transform is given by 
\begin{align*}
	(\Fourier_{\sigma} f)(r',k') := \frac{1}{(2\pi)^3} \int_{\R^3} \dd r' \int_{\R^3} \dd k' \; \e^{+ \ii (k' \cdot r - r' \cdot k)} \, f(r,k) 
	. 
\end{align*}
For a rigorous discussion, we point the reader to \cite[Section~4]{DeNittis_Lein:adiabatic_periodic_Maxwell_PsiDO:2013} and references therein. Instead, we will content ourselves with giving only the definitions necessary for a self-contained presentation. 

Assume $\mathfrak{h}_1$ and $\mathfrak{h}_2$ are Banach or Hilbert spaces; in our applications, they stand for $L^2(\T^3,\C^6)$, $\HperT$ and $\domainT$. A function $f \in \Cont^{\infty} \bigl ( \R^6 , \mathcal{B}(\mathfrak{h}_1,\mathfrak{h}_2) \bigr )$ is called \emph{equivariant} iff 
\begin{align}
	f(r,k - \gamma^*) = \e^{+ \ii \gamma^* \cdot \hat{y}} \, f(r,k) \, \e^{- \ii \gamma^* \cdot \hat{y}} 
	\label{setup:eqn:equivariance}
\end{align}
holds for all $(r,k) \in \R^6$ and $\gamma^* \in \Gamma^*$; we will call it right- or left-covariant, respectively, iff 
\begin{align}
	f(r,k - \gamma^*) &= f(r,k) \, \e^{- \ii \gamma^* \cdot \hat{y}} 
	\label{setup:eqn:covariance}
	\\
	f(r,k - \gamma^*) &= \e^{+ \ii \gamma^* \cdot \hat{y}} \, f(r,k) 
	\notag 
\end{align}
holds instead. Operator-valued Hörmander symbols 
\begin{align*}
	\Hoer{m} \bigl ( \mathcal{B}(\mathfrak{h}_1,\mathfrak{h}_2) \bigr ) :& \negmedspace= \left \{ 
		f \in \Cont^{\infty} \bigl ( \R^6,\mathcal{B}(\mathfrak{h}_1,\mathfrak{h}_2) \bigr )
		\; \; \big \vert \; \; 
		\forall a , b \in \N_0^3 : \snorm{f}_{m , a , b} < \infty
		\right \}
\end{align*}
of order $m \in \R$ and type $\rho \in [0,1]$ are defined through the usual seminorms 
\begin{align*}
	\snorm{f}_{m , a , b} := \sup_{(r,k) \in \R^6} \left ( \sqrt{1 + k^2}^{\; -m + \sabs{\beta} \rho} \, \bnorm{\partial_r^{\alpha} \partial_k^{\beta} f(r,k)}_{\mathcal{B}(\mathfrak{h}_1,\mathfrak{h}_2)} \right )
	, 
	&&
	a , b \in \N_0^3 
	, 
\end{align*}
where $\N_0 := \N \cup \{ 0 \}$. The class of symbols $\Hoerm{m}{\rho} \bigl ( \mathcal{B}(\mathfrak{h}_1,\mathfrak{h}_2) \bigr )$ which satisfy the equivariance condition~\eqref{setup:eqn:equivariance} are denoted with $\Hoermeq{m}{\rho} \bigl ( \mathcal{B}(\mathfrak{h}_1,\mathfrak{h}_2) \bigr )$; similarly, $\Hoermper{m}{\rho} \bigl ( \mathcal{B}(\mathfrak{h}_1,\mathfrak{h}_2) \bigr )$ is the class of $\Gamma^*$-periodic symbols, $f(r,k - \gamma^*) = f(r,k)$. Lastly, we introduce the notion of 
\begin{definition}[Semiclassical symbols]\label{setup:defn:semiclassical_symbols}
	Assume $\mathfrak{h}_j$, $j = 1 , 2$, are Banach spaces as above.	A map $f : [0,\lambda_0) \longrightarrow \Hoereq{m} \bigl ( \mathcal{B}(\mathfrak{h}_1,\mathfrak{h}_2) \bigr )$, $\lambda \mapsto f_{\lambda}$, is called an \emph{equivariant semiclassical symbol} of order $m \in \R$ and weight $\rho \in [0,1]$, that is $f \in \SemiHoereq{m} \bigl ( \mathcal{B}(\mathfrak{h}_1,\mathfrak{h}_2) \bigr )$, iff there exists a sequence $\{ f_n \}_{n \in \N_0}$, $f_n \in \Hoereq{m - n \rho} \bigl ( \mathcal{B}(\mathfrak{h}_1,\mathfrak{h}_2) \bigr )$, such that for all $N \in \N_0$, one has 
	\begin{align*}
		\lambda^{-N} \left ( f_{\lambda} - \sum_{n = 0}^{N - 1} \lambda^n \, f_n \right ) \in \Hoereq{m - N \rho} \bigl ( \mathcal{B}(\mathfrak{h}_1,\mathfrak{h}_2) \bigr ) 
	\end{align*}
	uniformly in $\lambda$ in the sense that for any $a , b \in \N_0^3$, there exist constants $C_{a b} > 0$ so that 
	\begin{align*}
		\norm{f_{\lambda} - \sum_{n = 0}^{N - 1} \lambda^n \, f_n}_{m , a , b} \leq C_{a b} \, \lambda^N 
	\end{align*}
	holds for all $\lambda \in [0,\lambda_0)$. 
	
	The Fréchet space of \emph{periodic semiclassical symbols} $\SemiHoerper{m} \bigl ( \mathcal{B}(\mathfrak{h}_1,\mathfrak{h}_2) \bigr )$ is defined analogously. 
\end{definition}
Not only the modulated Maxwell operator $\Mphys_{\lambda}^{\Zak}$ can be seen as a $\Psi$DO \cite[Theorem~1.3]{DeNittis_Lein:adiabatic_periodic_Maxwell_PsiDO:2013}, also the projection $\hat{\pi}_0$ associated to separate families of bands coincides with the quantization of the symbol $\pi_0 \in S^0_{0,\mathrm{eq}} \bigl ( \mathcal{B}(\HperT) \bigr ) \cap S^1_{0,\mathrm{eq}} \bigl ( \mathcal{B}(\HperT,\domainT) \bigr )$ defined in Lemma~\ref{setup:lem:existence_pi_0}.

Another building block of pseudodifferential calculus is the Moyal product $\Weyl$ implicitly defined through $\Op_{\lambda}(f \Weyl g) := \Op_{\lambda}(f) \, \Op_{\lambda}(g)$. It defines a bilinear continuous map 
\begin{align*}
	\sharp : \Hoereq{m_1} \bigl ( \mathcal{B}(\mathfrak{h}_1,\mathfrak{h}_2) \bigr ) \times \Hoereq{m_2} \bigl ( \mathcal{B}(\mathfrak{h}_2,\mathfrak{h}_3) \bigr ) \longrightarrow \Hoereq{m_1 + m_2} \bigl ( \mathcal{B}(\mathfrak{h}_1,\mathfrak{h}_3) \bigr ) 
\end{align*}
which has an asymptotic expansion 
\begin{align}
	\sharp : &\Hoereq{m_1} \bigl ( \mathcal{B}(\mathfrak{h}_1,\mathfrak{h}_2) \bigr ) \times \Hoereq{m_2} \bigl ( \mathcal{B}(\mathfrak{h}_2,\mathfrak{h}_3) \bigr ) \longrightarrow \Hoereq{m_1 + m_2} \bigl ( \mathcal{B}(\mathfrak{h}_1,\mathfrak{h}_3) \bigr ) 
	\notag \\
	&f \Weyl g\, \asymp \sum_{n = 0}^{\infty} \lambda^n \, (f \Weyl g)_{(n)} 
	= f \, g - \lambda \, \tfrac{\ii}{2} \{ f , g \} + \order(\lambda^2) 
	\label{setup:eqn:Moyal_asymp_expansion}
\end{align}
where $\{ f , g \} := \sum_{j = 1}^3 \bigl ( \partial_{k_j} f \; \partial_{r_j} g - \partial_{r_j} f \; \partial_{k_j} g \bigr )$ is the usual Poisson bracket. Each term $(f \Weyl g)_{(n)}(r,k)$ is a sum of products of derivatives of $f$ and $g$ evaluated at $(r,k)$. 

For technical reasons, we need to distinguish between the oscillatory integral $f \Weyl g$ and the formal sum $\sum_{n = 0}^{\infty} \lambda^n \, (f \Weyl g)_{(n)}$ when constructing the local Moyal resolvent. To simplify notation though, we will denote the formal sum on the right-hand side with the same symbol $f \Weyl g$. Given that the terms $(f \Weyl g)_{(n)}$ are purely local, the formal sum $\sum_{n = 0}^{\infty} \lambda^n \, (f \Weyl g)_{(n)}$ also makes sense if $f$ and $g$ are defined only on some common, open subset of $\R^6$. 
% subsection Pseudodifferential operators on weighted $L^2$-spaces (end)

\subsection{The $\lambda$-independent auxiliary representation} % (fold)
\label{setup:aux_rep}
As we have explained in Section~\ref{setup:Zak}, the Hilbert space $\Zak \Hil_{\lambda}$, its scalar product and its norm depend on $\lambda$. Hence, we will represent the problem on the $\lambda$-independent Hilbert space $\Zak \Hper$ of the periodic Maxwell operator using the unitary 
\begin{align}
	S(\ii \lambda \nabla_k) : \Zak \Hil_{\lambda} \longrightarrow \Zak \Hper 
	. 
	\label{setup:eqn:S_unitary}
\end{align}
This \emph{auxiliary representation} has the advantage that instead of dealing with $\lambda$-dependent operators on $\lambda$-dependent spaces, we work in a setting where the $\lambda$-dependence lies with the operators alone. The unitarity of $S(\ii \lambda \nabla_k)$ implies all norm bounds derived in the auxiliary representation carry over to the physical representation. Moreover, according to the arguments in \cite[Section~4.3]{DeNittis_Lein:adiabatic_periodic_Maxwell_PsiDO:2013} this operator can also be seen as $\Psi$DO associated to the symbol $S \in \Hoermeq{0}{0} \bigl ( \mathcal{B}(\HperT) \bigr )$. 

In this representation, the slowly modulated Maxwell operator 
\begin{align}
	M_{\lambda} :& \negmedspace= S(\ii \lambda \nabla_k) \, \Mphys_{\lambda}^{\Zak} \, S(\ii \lambda \nabla_k)^{-1} 
	\label{sapt:eqn:Maxwell_lambda_indep_rep}
	\\
	&= \tau(\ii \lambda \nabla_k) \, \Mper^{\Zak} 
	+ \notag \\
	&\qquad 
	+ \lambda \, \tau(\ii \lambda \nabla_k) \, W \, \left (
	\begin{matrix}
		0 & + \ii \bigl ( \nabla_r \ln \tau_{\mu} \bigr )^{\times}(\ii \lambda \nabla_k) \\
		- \ii \bigl ( \nabla_r \ln \tau_{\eps} \bigr )^{\times}(\ii \lambda \nabla_k) & 0 \\
	\end{matrix}
	\right )
	\notag
\end{align}
splits into a sum of two terms, the leading-order term $M_0$ is a simple rescaling of the periodic Maxwell operator by the deformation function 
\begin{align*}
	\tau(r) := \tau_{\eps}(r) \, \tau_{\mu}(r)
\end{align*}
while the first-order term $M_1$ is independent of $- \ii \nabla_y + \hat{k}$. The $M_{\lambda} = \Op_{\lambda}(\Msymb_{\lambda})$ all share the same $\lambda$-independent domain \cite[Lemma~2.6]{DeNittis_Lein:adiabatic_periodic_Maxwell_PsiDO:2013} 
\begin{align*}
	\Zak \domain \cong L^2(\BZ) \otimes \domainT 
\end{align*}
and can be seen as $\Psi$DOs associated to the $\mathcal{B}(\domainT,\HperT)$-valued function 
\begin{align}
	\Msymb_{\lambda} &
	= \tau \, \Mper(\, \cdot \,) 
	- \lambda \, \tau \, W \, \left (
	\begin{matrix}
		0 & \tfrac{\ii}{2} \, \bigl ( \nabla_r \ln \nicefrac{\tau_{\eps}}{\tau_{\mu}} \bigr )^{\times} \\
		\tfrac{\ii}{2} \, \bigl ( \nabla_r \ln \nicefrac{\tau_{\eps}}{\tau_{\mu}} \bigr )^{\times} & 0 \\
	\end{matrix}
	\right )
	. 
	\label{setup:eqn:M_lambda_PsiDO}
\end{align}
The auxiliary representation also simplifies some technical arguments involving $\Psi$DOs. For instance, symbols of selfadjoint pseudodifferential operators $\Op_{\lambda}(f)^* = \Op_{\lambda}(f) \in \mathcal{B}(\Zak \Hper)$ are necessarily selfadjoint-operator-valued, 
\begin{align*}
	f(r,k)^* = f(r,k) \in \mathcal{B}(\HperT)
	. 
\end{align*}
Consequently, the unitarity of $S(\ii \lambda \nabla_k) = \Op_{\lambda}(S)$ implies 
\begin{align*}
	\bigl ( S(\ii \lambda \nabla_k)^{-1} \, \Op_{\lambda}(f) \, S(\ii \lambda \nabla_k) \bigr )^* &= S(\ii \lambda \nabla_k)^{-1} \, \Op_{\lambda}(f) \, S(\ii \lambda \nabla_k) 
	\\
	&
	= \Op_{\lambda} \bigl ( S^{-1} \Weyl f \Weyl S \bigr ) 
	\in \mathcal{B}(\Zak \Hil_{\lambda})
\end{align*}
is also a selfadjoint $\Psi$DO with symbol $S^{-1} \Weyl f \Weyl S$. For a detailed discussion of pseudo\-differential theory in this setting, we refer to \cite[Section~4]{DeNittis_Lein:adiabatic_periodic_Maxwell_PsiDO:2013}. 
%
% subsection The $\lambda$-independent auxiliary representation (end)
% section Setup (end)

%% file: section_03.tex
%!TEX root = /Users/max/Dropbox/research/photonic crystals/sapt for isotropic photonic crystals/arxiv/v2/photonic sapt.tex
\section{Space-adiabatic perturbation theory} % (fold)
\label{sapt}
The key technical tool, space-adiabatic perturbation theory \cite{PST:sapt:2002} exploits a common structure shared by a wide array of multiscale systems, \eg \cite{PST:effective_dynamics_Bloch:2003,Teufel:adiabatic_perturbation_theory:2003,PST:Born-Oppenheimer:2007,Tenuta_Teufel:adiabatic_QED:2008,DeNittis_Lein:Bloch_electron:2009,Fuerst_Lein:scaling_limits_Dirac:2008}: 
\begin{enumerate}[(i)]
	\item \emph{Slow and fast degrees of freedom:} the physical Hilbert space $\Hil_{\lambda}$ is decomposed using a sequence of unitaries into $L^2(\BZ) \otimes \HperT$ (\cf equation~\eqref{intro:eqn:spaces_diagram}) in which the unperturbed Maxwell operator is block-diagonal. The commutator of the fast variables $\bigl ( \hat{y} , - \ii \nabla_y \bigr )$ is $\order(1)$ while that of the slow variables $\bigl ( \ii \lambda \nabla_k , \hat{k} \bigr )$ is $\order(\lambda)$, and hence, the names. 
	\item \emph{A small parameter $\lambda \ll 1$:} the adiabatic parameter quantifies on which scale the photonic crystal is modulated, measured in units of the lattice constant. 
	\item \emph{The relevant part of the spectrum} $\Omega_{\mathrm{rel}}$ consisting of bands which do not intersect or merge with the remaining bands (\cf Assumption~\ref{intro:assumption:gap_condition}). 
\end{enumerate}
This technique relies on pseudodifferential calculus for systematic perturbation expansions in $\lambda$ and yields the last column and lower row in \eqref{intro:eqn:spaces_diagram} as well as equation \eqref{intro:eqn:approximate_dynamics}: 
\begin{theorem}[Effective dynamics: physical representation]\label{sapt:thm:eff_dyn}
	Suppose Assumption~\ref{intro:assumption:eps_mu} is satisfied, $\Omega_{\mathrm{rel}} = \{ \omega_n \}_{n \in \Index}$ satisfies the Gap Condition~\ref{intro:assumption:gap_condition} and the Bloch bundle $\blochb(\pi_0)$ is trivial. Pick an arbitrary orthogonal rank-$\sabs{\Index}$ projection $\piref \in \mathcal{B}(\HperT)$ and define $\Piref = \id_{L^2(\BZ)} \otimes \piref$. 
	
	Then there exist 
	\begin{enumerate}[(1)]
		\item an orthogonal projection $\pmb{\Pi}_{\lambda} \in \mathcal{B}(\Hil_{\lambda})$, 
		\item a unitary map $\mathbf{U}_{\lambda}$ which intertwines $\pmb{\Pi}_{\lambda}^{\Zak} = \Zak \, \pmb{\Pi}_{\lambda} \, \Zak^{-1}$ and $\Piref$, and 
		\item a selfadjoint operator $\Op_{\lambda}(\Msymb_{\eff}) \in \mathcal{B}(\Zak \Hper)$
	\end{enumerate}
	such that 
	\begin{align*}
		\bnorm{\bigl [ \Mphys_{\lambda} , \pmb{\Pi}_{\lambda} \bigr ]}_{\mathcal{B}(\Hil_{\lambda})} = \order(\lambda^{\infty})
	\end{align*}
	and 
	\begin{align}
		\norm{\Bigl ( \e^{- \ii t \Mphys_{\lambda}} - \Zak^{-1} \, \mathbf{U}_{\lambda}^* \, \e^{- \ii t \Op_{\lambda}(\Msymb_{\eff})} \, \mathbf{U}_{\lambda} \, \Zak \Bigr ) \, \pmb{\Pi}_{\lambda}}_{\mathcal{B}(\Hil_{\lambda})} &= \order \bigl ( \lambda^{\infty} ( 1 + \abs{t}) \bigr ) 
		. 
		\label{sapt:eqn:effective_dynamics_approximate_phys}
	\end{align}
	The effective Maxwell operator is the quantization of the $\Gamma^*$-periodic symbol 
	\begin{align}
		\Msymb_{\eff} :& \negmedspace= \piref \, \mathbf{u}_{\lambda} \Weyl \pmb{\Msymb}_{\lambda} \Weyl \mathbf{u}_{\lambda}^* \, \piref 
		\asymp \sum_{n = 0}^{\infty} \lambda^n \, \Msymb_{\eff \, n} 
		\in \SemiHoermper{0}{0} \bigl ( \mathcal{B}(\HperT) \bigr ) 
		\label{sapt:eqn:Meff_physical_rep}
	\end{align}
	which is defined in terms of equations \eqref{setup:eqn:symbol_Mphys} and \eqref{sapt:eqn:u_ubold}, and whose asymptotic expansion can be computed to any order in $\lambda$. 
\end{theorem}
The proof is a combination of a change of representation with a straightforward adaption of the arguments in \cite{PST:effective_dynamics_Bloch:2003}; To streamline the presentation, we have moved the proofs of these intermediate results to Section~\ref{technical}. 

Space-adiabatic perturbation theory uses pseudodifferential theory to construct 
\begin{align*}
	\pmb{\Pi}_{\lambda}^{\Zak} := \Zak \, \pmb{\Pi}_{\lambda} \, \Zak^{-1} 
	= \Op_{\lambda}(\pmb{\pi}_{\lambda}) + \order_{\norm{\cdot}}(\lambda^{\infty})
\end{align*}
systematically in BFZ representation as the quantization of an operator-valued symbol $\pmb{\pi}_{\lambda} \in \SemiHoermeq{0}{0} \bigl ( \mathcal{B}(\HperT) \bigr )$. However, due to the advantages outlined in Section~\ref{setup:aux_rep} it is more convenient to equivalently construct the \emph{superadiabatic projection}
\begin{align*}
	\Pi_{\lambda} &= \Op_{\lambda}(\pi_{\lambda}) + \order_{\norm{\cdot}}(\lambda^{\infty}) 
	\\
	&= S(\ii \lambda \nabla_k) \, \pmb{\Pi}_{\lambda}^{\Zak} \, S(\ii \lambda \nabla_k)^{-1} 
	\\
	&= S(\ii \lambda \nabla_k) \, \Op_{\lambda}(\pmb{\pi}_{\lambda}) \, S(\ii \lambda \nabla_k)^{-1} + \order_{\norm{\cdot}}(\lambda^{\infty}) 
	\in \mathcal{B}(\Zak \Hper) 
\end{align*}
and the \emph{intertwining unitary} 
\begin{align*}
	U_{\lambda} &= \Op_{\lambda}(u_{\lambda}) + \order_{\norm{\cdot}}(\lambda^{\infty})
	\\
	&
	= \mathbf{U}_{\lambda} \, S(\ii \lambda \nabla_k)^{-1}
	= \Op_{\lambda}(u_{\lambda}) + \order_{\norm{\cdot}}(\lambda^{\infty})
	\in \mathcal{B}(\Zak \Hper)
\end{align*}
in the auxiliary representation instead where bold and non-bold symbols are related by $S$, 
\begin{align}
	\pi_{\lambda} &= S \Weyl \pmb{\pi}_{\lambda} \Weyl S^{-1} 
	, 
	\qquad  
	u_{\lambda} = \mathbf{u}_{\lambda} \Weyl S^{-1}
	, 
	\label{sapt:eqn:u_ubold}
	\\
	\Msymb_{\eff} &= \piref \, u_{\lambda} \Weyl \Msymb_{\lambda} \Weyl u_{\lambda}^* \, \piref 
	. 
	\notag 
\end{align}
The error estimates we derive in the $\lambda$-independent representation carry over \emph{identically} to the representation on $\Zak \Hil_{\lambda}$ due to the unitarity of $S(\ii \lambda \nabla_k)$. 
\begin{proposition}[Superadiabatic projection]\label{sapt:prop:projection}
	Suppose Assumption~\ref{intro:assumption:eps_mu} holds and assume $\Omega_{\mathrm{rel}} = \{ \omega_n \bigr \}_{n \in \Index}$ satisfies the Gap Condition~\ref{intro:assumption:gap_condition}. 
	
	Then there exists an orthogonal projection 
	\begin{align*}
		\Pi_{\lambda} = \Op_{\lambda}(\pi_{\lambda}) + \order_{\norm{\cdot}}(\lambda^{\infty})
	\end{align*}
	that $\order(\lambda^{\infty})$-almost commutes in norm with the Maxwell operator, 
	\begin{align*}
		\bigl [ M_{\lambda} \, , \, \Pi_{\lambda} \bigr ] = \order_{\norm{\cdot}}(\lambda^{\infty}) 
		, 
	\end{align*}
	and is $\order(\lambda^{\infty})$-close in norm to the $\Psi$DO associated to the equivariant semiclassical symbol 
	\begin{align*}
		\pi_{\lambda} \asymp \sum_{n = 0}^{\infty} \lambda^n \, \pi_n 
		\in \SemiHoermeq{0}{0} \bigl ( \mathcal{B}(\HperT) \bigr ) \cap \SemiHoermeq{1}{0} \bigl ( \mathcal{B}(\HperT,\domainT) \bigr ) 
		. 
	\end{align*}
	The principal part $\pi_0$ coincides with the projection constructed in Lemma~\ref{setup:lem:existence_pi_0}. 
\end{proposition}
The triviality of the Bloch bundle is \emph{not} needed for the construction of $\Pi_{\lambda}$. But it enters crucially in the construction of the intertwining unitary $U_{\lambda} = \Op_{\lambda}(u_0) + \order_{\norm{\cdot}}(\lambda)$: The triviality of $\blochb(\pi_0)$ is necessary and sufficient for the existence of $u_0$ as a right-covariant symbol. However, unlike periodic Schrödinger operators where in the absence of strong magnetic fields the Bloch bundle is automatically trivial \cite{Panati:triviality_Bloch_bundle:2006,DeNittis_Lein:exponentially_loc_Wannier:2011}, the situation for periodic Maxwell operators is more delicate; we will explore this aspect further in Section~\ref{physical:real:perturbed}. 
\begin{proposition}[Intertwining unitary]\label{sapt:prop:unitary}
	Suppose we are in the setting of Theorem~\ref{sapt:thm:eff_dyn}. Then there exists a unitary map $U_{\lambda}: \Zak \Hper \longrightarrow \Zak \Hper$ given by
	\begin{align*}
		U_{\lambda} = \Op_{\lambda}(u_{\lambda}) + \order_{\norm{\cdot}}(\lambda^{\infty}) 
	\end{align*}
	which is $\order(\lambda^{\infty})$-close in norm to the quantization of a right-covariant symbol (\cf \eqref{setup:eqn:covariance})
	\begin{align*}
		u_{\lambda} &\in \SemiHoerm{0}{0} \bigl ( \mathcal{B} (\HperT) \bigr )
		, 
	\end{align*}
	and intertwines $\Pi_{\lambda}$ from Proposition~\ref{sapt:prop:projection} with the reference projection $\Piref = \id_{L^2(\BZ)} \otimes \piref$, 
	\begin{align}
		U_{\lambda} \, \Pi_{\lambda} \, U_{\lambda}^* = \Piref 
		. 
		\label{sapt:eqn:intertwining_property}
	\end{align}
\end{proposition}
Using the previously constructed symbols of the superadiabatic projection and the intertwining unitary, we can use a Duhamel argument to show the Weyl quantization of the effective Maxwellian approximates full electrodynamics to any order for states in $\ran \Pi_{\lambda}$. 
\begin{proposition}[Effective Maxwellian]\label{sapt:prop:effective_Maxwellian}
	Suppose we are in the setting of Theorem~\ref{sapt:thm:eff_dyn}. Then the pseudodifferential operator associated to a resummation of 
	\begin{align}
		\Msymb_{\eff} &= \piref \, u_{\lambda} \Weyl \Msymb_{\lambda} \Weyl u_{\lambda}^* \, \piref 
		\asymp \sum_{n = 0}^{\infty} \lambda^n \, \Msymb_{\eff \, n} 
		\label{sapt:eqn:Meff_lambda_indep_rep}
		, 
		\\ 
		\Msymb_{\eff \, n} &\in S^0_{0,\mathrm{per}} \bigl ( \mathcal{B}(\HperT) \bigr ) 
		, 
		\notag 
	\end{align}
	approximates the full electrodynamics for states in $\ran \Pi_{\lambda}$, 
	\begin{align}
		\Bigl ( \e^{- \ii t M_{\lambda}} - U_{\lambda}^{-1} \, \e^{- \ii t \Op_{\lambda}(\Msymb_{\eff})} \, U_{\lambda} \Bigr ) \, \Pi_{\lambda} &= \order_{\norm{\cdot}} \bigl ( \lambda^{\infty} ( 1 + \abs{t}) \bigr ) 
		. 
		\label{sapt:eqn:effective_dynamics_approximate}
	\end{align}
\end{proposition}
Now the proof of Theorem~\ref{sapt:thm:eff_dyn} simply follows by combining the previous results with a change of representation. 
\begin{proof}[Theorem~\ref{sapt:thm:eff_dyn}]
	The assumptions of the theorem include those of Propositions~\ref{sapt:prop:projection}, \ref{sapt:prop:unitary} and \ref{sapt:prop:effective_Maxwellian}, and hence we obtain a projection $\Pi_{\lambda} = \Op_{\lambda}(\pi_{\lambda}) + \order_{\norm{\cdot}}(\lambda^{\infty}) \in \mathcal{B}(\Zak \Hper)$, a unitary $U_{\lambda} = \Op_{\lambda}(u_{\lambda}) + \order_{\norm{\cdot}}(\lambda^{\infty}) \in \mathcal{B}(\Zak \Hper)$ and an effective Maxwell operator $\Op_{\lambda}(\Msymb_{\eff})$. 
	
	Inspecting Diagram~\eqref{intro:eqn:spaces_diagram} and the comments in Section~\ref{setup:aux_rep}, we conclude that 
	\begin{align*}
		\pmb{\Pi}_{\lambda} :& \negmedspace= \Zak^{-1} \, \pmb{\Pi}_{\lambda}^{\Zak} \, \Zak 
		= \Zak^{-1} \, S(\ii \lambda \nabla_k)^{-1} \, \Pi_{\lambda} \, S(\ii \lambda \nabla_k) \, \Zak 
		\\
		&
		= \Op_{\lambda} \bigl ( S^{-1} \Weyl \pi_{\lambda} \Weyl S \bigr ) + \order_{\norm{\cdot}}(\lambda^{\infty}) \in \mathcal{B}(\Zak \Hil_{\lambda})
	\end{align*}
	and 
	\begin{align*}
		\mathbf{U}_{\lambda} := U_{\lambda} \, S(\ii \lambda \nabla_k) 
		= \Op_{\lambda} \bigl ( u_{\lambda} \Weyl S \bigr ) + \order_{\norm{\cdot}}(\lambda^{\infty}) 
		\in \mathcal{B} \bigl ( \Zak \Hil_{\lambda} , \Zak \Hper \bigr )
	\end{align*}
	are the projection and unitary we are looking for. The properties enumerated in the Theorem follow directly from the corresponding properties of $\Pi_{\lambda}$ and $U_{\lambda}$. Moreover, order-by-order, the effective Maxwell operator defined in \eqref{sapt:eqn:Meff_physical_rep} and in \eqref{sapt:eqn:Meff_lambda_indep_rep} are identical (up to a resummation) since $\Msymb_{\lambda} = S \Weyl \pmb{\Msymb}_{\lambda} \Weyl S^{-1}$ by definition and 
	\begin{align*}
		\piref \, \mathbf{u}_{\lambda} \Weyl \pmb{\Msymb}_{\lambda} \Weyl \mathbf{u}_{\lambda}^* \, \piref &= \piref \, u_{\lambda} \Weyl S \Weyl \pmb{\Msymb}_{\lambda} \Weyl S^{-1} \Weyl u_{\lambda}^* \, \piref 
		\\
		&= \piref \, u_{\lambda} \Weyl \Msymb_{\lambda} \Weyl u_{\lambda}^* \, \piref 
		. 
	\end{align*}
	Lastly, \eqref{intro:eqn:spaces_diagram} also shows how to relate the unitary evolution groups and make use of Proposition~\ref{sapt:prop:effective_Maxwellian}: 
	\begin{align*}
		\e^{- \ii t \Mphys_{\lambda}} \, \pmb{\Pi}_{\lambda} 
		&= \Zak^{-1} \, S(\ii \lambda \nabla_k)^{-1} \, \e^{- \ii t M_{\lambda}} \, \Pi_{\lambda} \, S(\ii \lambda \nabla_k) \, \Zak 
		\\
		&= \Zak^{-1} \, S(\ii \lambda \nabla_k)^{-1} \, U_{\lambda}^{-1} \, \e^{- \ii t \, \Op_{\lambda}(\Msymb_{\eff})} \, U_{\lambda} \, \Pi_{\lambda} \, S(\ii \lambda \nabla_k) \, \Zak + \order_{\norm{\cdot}}(\lambda^{\infty}) 
	\end{align*}
	This finishes the proof. 
\end{proof}
%
% section Effective electrodynamics (end)

%% file: section_04.tex
%!TEX root = /Users/max/Dropbox/research/photonic crystals/sapt for isotropic photonic crystals/arxiv/v2/photonic sapt.tex
\section{Physical solutions localized in a frequency range} % (fold)
\label{physical}
So far the derivation of effective dynamics for states associated to the family of bands $\Omega_{\mathrm{rel}}$ involved only the \emph{dynamical} Maxwell equations \eqref{intro:eqn:traditional_Maxwell_equations_dynamics}. To investigate whether the almost-invariant subspace $\ran \pmb{\Pi}_{\lambda}$ contains \emph{physically relevant} states, we need to establish that up to $\order(\lambda^{\infty})$
\begin{enumerate}[(i)]
	\item $\ran \pmb{\Pi}_{\lambda}$ contains real elements and 
	\item elements of $\ran \pmb{\Pi}_{\lambda}$ satisfy the no-source condition \eqref{intro:eqn:traditional_Maxwell_equations_no_source}. 
\end{enumerate}

\subsection{Real states} % (fold)
\label{physical:real}
The Schrödinger-type equation \eqref{intro:eqn:Maxwell_Schroedinger_type} is naturally defined for complex-valued fields, and it is not at all automatic that initially real-valued electromagnetic fields remain real-valued. In fact, there are physical scenarios where this is incorrect, but more on that below.

\subsubsection{Subspaces supporting real states} % (fold)
\label{physical:real:generic}
For the sake of simplicity, we will consistently denote complex conjugation on $L^2(X,\C^N)$ with $C$. Here $X$ could be $\R^3$, $\T^3$ or any other space.
Clearly, $C$ is anti-linear and involutive. We also need the central notion: 
\begin{definition}\label{physical:defn:support_real}
	A closed subspace $\mathfrak{K} \subseteq L^2(X,\C^N)$ is said to \emph{support real states} if and only if $C \mathfrak{K} = \mathfrak{K}$. 
\end{definition}
Now $C$ is a “particle-hole symmetry” of the free Maxwell operator, 
\begin{align*}
	C \, \Rot \, C &= - \Rot 
	\, , 
\end{align*}
which extends to the Maxwell operator $\Mphys_{\lambda} = W_{\lambda} \, \Rot$ if and only if $(\eps_{\lambda},\mu_{\lambda})$ which enter the weight operator 
\begin{align*}
	W_{\lambda} = \left (
	\begin{matrix}
		\eps_{\lambda}^{-1} & 0 \\
		0 & \mu_{\lambda}^{-1} \\
	\end{matrix}
	\right )
\end{align*}
are real, 
\begin{align}
	C \, W_{\lambda} \, C = W_{\lambda} 
	\; \; \Longleftrightarrow \; \; 
	C \, \Mphys_{\lambda} \, C = - \Mphys_{\lambda} 
	. 
	\label{physical:eqn:C_M_anticommute}
\end{align}
Seeing as the weight operator enters into the definition of the scalar product on $\Hil_{\lambda}$ (\cf equation~\eqref{intro:eqn:weighted_scalar_product}), $C$ is anti-unitary if and only if the weights are real. Note that the periodicity of $(\eps,\mu)$ is not needed for these arguments. 

Consequently, the unitary time evolution \emph{commutes} with $C$ iff the weights are real: 
\begin{align}
	C \, W_{\lambda} \, C = W_{\lambda} 
	\; \; \Longleftrightarrow \; \; 
	C \, \e^{- \ii t \Mphys_{\lambda}} \, C = \e^{+ \ii t \, C \, \Mphys_{\lambda} \, C} 
	= \e^{- \ii t \Mphys_{\lambda}} 
	\label{physical:eqn:C_edyn_commute}
\end{align}
Thus, if the weights and the initial conditions $(\mathbf{E},\mathbf{H})$ are real, then also for all $t \in \R$ the time-evolved fields $\bigl ( \mathbf{E}(t),\mathbf{H}(t) \bigr )$ take values in $\R^6$. Indeed if $\Re , \Im : \Hil_{\lambda} \rightarrow \Hil_{\lambda}$ are \emph{real} and \emph{imaginary part} operators 
\begin{align*}
	\Re := \tfrac{1}{2} \bigl ( 1 + C \bigr ) 
	, 
	\qquad \qquad 
	\Im := \tfrac{1}{\ii 2} \bigl ( 1 - C \bigr ) 
	, 
\end{align*}
then \eqref{physical:eqn:C_edyn_commute} immediately implies 
\begin{align}
	\bigl [ \Re , W_{\lambda} \bigr ] = 0
	\; \; \Longleftrightarrow \; \; 
	\bigl [ \Re , \e^{- \ii t \Mphys_{\lambda}} \bigr ] = 0
	\label{physical:eqn:Re_edyn_commute}
\end{align}
and a similar statement for $\Im$. 
\begin{remark}
	The above equation is not of purely mathematical significance: in gyrotropic photonic crystals, $\eps$ and $\mu$ are positive, hermitian-matrix-valued functions with complex offdiagonal entries \cite{Yeh_Chao_Lin:Faraday_effect:1999,Wu_Levy_Fratello_Merzlikin:gyrotropic_photonic_crystals:2010,Kriegler_Rill_Linden_Wegener:bianisotropic_photonic_metamaterials:2010,Esposito_Gerace:photonic_crystals_broken_TR_symmetry:2013}. Consequently, for those materials $C$ and and $\e^{- \ii t \Mphys_{\lambda}}$ no longer commute. The breaking of this particle-hole symmetry \cite{Altland_Zirnbauer:superconductors_symmetries:1997} is crucial to allow the formation of topologically protected states akin to quantum Hall states in crystalline solids. 
\end{remark}
The next Lemma states that we could equivalently use a criterion involving \emph{orthogonal} projections once we equip $L^2(\R^3,\C^N)$ with a weighted scalar product of the form \eqref{intro:eqn:weighted_scalar_product}: 
\begin{lemma}\label{physical:lem:characterization_real_states}
	A closed subspace $\mathfrak{K}$ of the Hilbert space $\Hil_{\lambda}$ supports real states if and only if the orthogonal projection $P \in \mathcal{B}(\Hil_{\lambda})$ onto $\mathfrak{K} = \ran P$ commutes with $C$, 
	\begin{align*}
		[C,P] = 0 
		. 
	\end{align*}
\end{lemma}
\begin{proof}
	Let $\mathfrak{K} \subseteq \Hil_{\lambda}$ be a closed subspace with $C \mathfrak{K} = \mathfrak{K}$. Then there exists a unique orthogonal projection $P$ onto $\mathfrak{K}$. Since	$P' = C \, P \, C$ projects onto $C \mathfrak{K}$ and $C \mathfrak{K} = \mathfrak{K}$, it follows from the uniqueness of the projection that $P'=P$.
	
	Conversely, let $P$ be an orthogonal projection for which $[C,P] = 0$ holds. Then $C^2 = \id$ implies $C \, P \, C = P$, and thus $C \, \ran P = \ran P$. 
\end{proof}
%
% subsubsection Subspaces supporting real states (end)

\subsubsection{Unperturbed non-gyrotropic photonic crystals} % (fold)
\label{physical:real:periodic}
Now we turn to the case where the material weights $\eps$ and $\mu$ are \emph{real} and $\Gamma$-periodic, \ie $\lambda = 0$; the first assumption excludes gyrotropic photonic crystals. We will be interested in subspaces of the form $\Zak^{-1} \, \ran \hat{\pi}_0 \subset \Hper$ where $\hat{\pi}_0$ is the “local spectral projection” from Lemma~\ref{setup:lem:existence_pi_0} associated to a family of bands $\Omega_{\mathrm{rel}} = \{ \omega_n \}_{n \in \Index}$. 

To exploit the periodicity, let us switch to the BFZ representation where the action of complex conjugation $C^{\Zak} := \Zak \, C \, \Zak^{-1}$ is 
\begin{align}
	\bigl ( C^{\Zak} \Psi \bigr )(k) = \overline{\Psi(-k)}
	 && 
	 \forall \Psi \in \Zak \Hper
	 .
	 \label{physical:eqn:relation_CZak_C}
\end{align}
Hence, $C^{\Zak}$ does not fiber; instead it often induces relationships between the fiber spaces at $k$ and $-k$. In BFZ representation, the condition “$\mathfrak{K}$ supports real states” translates to 
\begin{align}
	\Psi \in \Zak \mathfrak{K} 
	\; \Longleftrightarrow \; 
	C^{\Zak} \Psi \in \Zak \mathfrak{K} 
	, 
	\label{physical:eqn:Zak_reality_condititon}
\end{align}
and implies a symmetry condition for the fiber projections $\pi_0(k)$: 
\begin{proposition}\label{physical:prop:equivalences_reality_condition}
	Suppose $(\eps,\mu)$ are real. Then 
	$\Zak^{-1} \, \ran \hat{\pi}_0$ supports real states if and only if $C \, \pi_0(k) \, C = \pi_0(-k)$ holds for all $k \in \BZ$. 
\end{proposition}
\begin{proof}
	From Lemma~\ref{physical:lem:characterization_real_states}, equation~\eqref{physical:eqn:C_M_anticommute} and the unitarity of $\Zak$ one has that $\Zak^{-1} \, \ran \hat{\pi}_0$ supports real states if and only if $C^{\Zak} \, \hat{\pi}_0 \, C^{\Zak} = \hat{\pi}_0$. From this equality, using the fiber representation \eqref{physical:eqn:pi_0} of the projection $\hat{\pi}_0$, the symmetry of the Bloch functions \eqref{setup:eqn:symmetric_twin_Bloch_function} and the action of $C^{\Zak}$ described by \eqref{physical:eqn:relation_CZak_C}, one easily deduces the fiber-wise relation $C \, \pi_0(k) \, C = \pi_0(-k)$.
\end{proof}
$C \, \pi_0(k) \, C = \pi_0(-k)$ translates to a spectral condition most conveniently written in terms of the \emph{relevant part of the spectrum} 
\begin{align}
	\specrel(k) := \bigcup_{n \in \Index} \bigl \{ \omega_n(k) \bigr \} 
	. 
	\label{physical:eqn:specrel}
\end{align}
\begin{corollary}\label{physical:cor:spectral_condition_reality}
	Suppose $(\eps,\mu)$ are real. Then 
	$\Zak^{-1} \, \ran \hat{\pi}_0$ supports real states if and only if $\specrel(-k) = - \specrel(k)$ holds for all $k \in \BZ$. 
\end{corollary}
\begin{proof}
	Thanks to Proposition~\ref{physical:prop:equivalences_reality_condition} we can equivalently prove that $C \, \pi_0(k) \, C = \pi_0(-k)$ holds iff $\specrel(-k) = - \specrel(k)$. This follows easily from the spectral properties of $\Mper(k)$ for real weights discussed in Section \ref{setup:periodic}.
\end{proof}
\begin{proposition}\label{physical:prop:triviality_bundle}
	Suppose $(\eps,\mu)$ are real, the family of relevant bands $\Omega_{\mathrm{rel}} = \{ \omega_n \}_{n \in \Index}$ satisfies the Gap Condition~\ref{intro:assumption:gap_condition} and $\Zak^{-1} \, \ran \hat{\pi}_0$ supports real states in the sense of Definition~\ref{physical:defn:support_real}. Then the Bloch bundle $\blochb(\pi_0)$ associated to $\hat{\pi}_0$ is trivial. 
\end{proposition}
\begin{proof}
	As explained in \cite[Theorem~4.6]{DeNittis_Lein:exponentially_loc_Wannier:2011} or \cite[Theorem~1]{Panati:triviality_Bloch_bundle:2006} the relation $C \, \pi_0(k) \, C = \pi_0(-k)$ guarantees that the first Chern class of $\mathcal{E}_\BZ(\pi_0)$ vanishes which in turn implies the triviality of the Bloch bundle due to the low dimensionality of the base space. 
\end{proof}
\begin{remark}\label{physical:rem:Bloch_bundle}
	A more careful analysis shows that the trivial bundle $\blochb(\pi_0)$ can be naturally seen as the sum of two potentially non-trivial bundles with opposite Chern numbers: If we decompose $\hat{\pi}_0 = \hat{\pi}_{+ \, 0} + \hat{\pi}_{- \, 0}$ into positive and negative frequency contributions, we obtain the natural splitting of
	\begin{align*}
		\blochb(\pi_0) \cong \blochb(\pi_{+ \, 0}) \oplus \blochb(\pi_{- \, 0})
		.
	\end{align*}
	into the Whitney sum of positive and negative frequency Bloch bundle. The triviality of $ \blochb(\pi_0)$, the composition property of Chern classes and the low dimensionality of the base space $\BZ$ yields
	\begin{align*}
		0 = c_1 \bigl ( \blochb(\pi_0) \bigr ) = c_1 \bigl ( \blochb(\pi_{+ \, 0}) \bigr ) + c_1 \bigl ( \blochb(\pi_{- \, 0}) \bigr )
		% = 0 \in H^2(\BZ,\Z) \cong \Z^3
		,
	\end{align*}
	which  is equivalent to $c_1 \bigl ( \blochb(\pi_{+ \, 0}) \bigr ) = - c_1 \bigl ( \blochb(\pi_{- \, 0}) \bigr )$ (one needs to use also the fact that $H^2(\BZ,\Z)$ is torsion free). However, the particle-hole symmetry $C$ by no means implies $c_1 \bigl ( \blochb(\pi_{\pm \, 0}) \bigr ) = 0$.
\end{remark}
%
% subsubsection Periodic weights (end)

\subsubsection{Slowly modulated non-gyrotropic photonic crystals} % (fold)
\label{physical:real:perturbed}
The periodic case could be handled on the level of fiber decompositions of periodic operators; the superadiabatic projection $\pmb{\Pi}_{\lambda}^{\Zak}$ has a more complicated structure, and we need to transfer \eqref{setup:eqn:complex_symmetry_fiber_Maxwell} to the level of symbols. With this in mind, we introduce the operation 
\begin{align*}
	(\mathfrak{c} f)(r,k) := C \, f(r,-k) \, C
\end{align*}
for suitable operator-valued symbols $f \in \Hoermeq{m}{0} \bigl ( \mathcal{B}(\HperT) \bigr )$, because then, conjugation with $C^{\Zak}$ and $\mathfrak{c}$ are intertwined via $\Op_{\lambda}$: 
\begin{lemma}\label{real_eff:lem:complex_conjugation_PsiDO}
	$C^{\Zak} \, \Op_{\lambda}(f) \, C^{\Zak} = \Op_{\lambda}(\mathfrak{c} f)$ $\quad \forall \Hoermeq{0}{0} \bigl ( \mathcal{B}(\HperT) \bigr )$ 
\end{lemma}
The proof is straightforward and can be found in Section~\ref{technical:physical}. Now we are in a position to prove the first physicality condition. Just like in the case of the perfectly periodic Maxwell operator, $\ran \pmb{\Pi}_{\lambda}$ supports real states if and only if $\specrel(k)$ is chosen antisymetrically. 
\begin{proposition}\label{physical:prop:real_solutions_Pi_lambda}
	Assume $(\eps_{\lambda},\mu_{\lambda})$ are real and $\specrel(k) = - \specrel(-k)$. Then $C \, \pmb{\Pi}_{\lambda} \, C = \pmb{\Pi}_{\lambda} + \order_{\norm{\cdot}}(\lambda^{\infty})$ holds. 
\end{proposition}
\begin{proof}
	First of all, we can reduce the problem to computing $C^{\Zak} \, \Pi_{\lambda} \, C^{\Zak}$ since the relation $\bigl [ C , S(\lambda \hat{x})^{\pm 1} \bigr ] = 0$ implies after the Zak transform that
	\begin{align*}
		C^{\Zak} \, \pmb{\Pi}_{\lambda}^{\Zak} \, C^{\Zak} &= C^{\Zak} \, S(\ii \lambda \nabla_k)^{-1} \, \Pi_{\lambda} \, S(\ii \lambda \nabla_k) \, C^{\Zak} 
		\\
		&= S(\ii \lambda \nabla_k)^{-1} \, C^{\Zak} \, \Pi_{\lambda} \, C^{\Zak} \, S(\ii \lambda \nabla_k) 
		. 
	\end{align*}
	Now we replace $\Pi_{\lambda}$ by $\Op_{\lambda}(\pi_{\lambda})$ and use Lemma~\ref{real_eff:lem:complex_conjugation_PsiDO}, 
	\begin{align}
		\ldots &= S(\ii \lambda \nabla_k)^{-1} \, C^{\Zak} \, \Op_{\lambda}(\pi_{\lambda}) \, C^{\Zak} \, S(\ii \lambda \nabla_k) + \order_{\norm{\cdot}}(\lambda^{\infty})
		\notag \\
		&= S(\ii \lambda \nabla_k)^{-1} \, \Op_{\lambda}(\mathfrak{c} \pi_{\lambda}) \, S(\ii \lambda \nabla_k) + \order_{\norm{\cdot}}(\lambda^{\infty})
		. 
		\label{real_eff:eqn:computation_C_fatPi_C}
	\end{align}
	Given that the weights are real and $\specrel(-k) = - \specrel(k)$, Lemma~\ref{technical:lem:c_pi_pi} applies, and we deduce $\mathfrak{c} \pi_{\lambda} = \pi_{\lambda}$. Once we plug that back into \eqref{real_eff:eqn:computation_C_fatPi_C}, we obtain 
	\begin{align*}
		\ldots &= S(\ii \lambda \nabla_k)^{-1} \, \Op_{\lambda}(\pi_{\lambda}) \, S(\ii \lambda \nabla_k) + \order_{\norm{\cdot}}(\lambda^{\infty})
		\\
		&= \pmb{\Pi}_{\lambda} + \order_{\norm{\cdot}}(\lambda^{\infty})
		, 
	\end{align*}
	and thus finish the proof of the claim. 
\end{proof}
% 
% subsubsection Slowly modulated photonic crystals (end)
% subsection Support of real states (end)

\subsection{Source-free condition} % (fold)
\label{physical:no_sources}
The almost-invariant subspace $\ran \pmb{\Pi}_{\lambda}$ supports real states if and only if the relevant spectrum is chosen symmetrically and the weights $(\eps_{\lambda},\mu_{\lambda})$ are real. Both of these conditions are \emph{not} necessary for the second physicality condition. Instead of checking that elements of the almost-invariant subspace approximately satisfy \eqref{intro:eqn:traditional_Maxwell_equations_no_source}, we use the characterization in terms of the invariant subspaces 
\begin{align*}
	\Hil_{\lambda} = \Jphys_{\lambda} \oplus \Gphys 
	= \ran \Pphys_{\lambda} \oplus \ran \Qphys_{\lambda}
	. 
\end{align*}
The crucial ingredient in the proof is the assumption that $\specrel$ does not contain ground state bands, because then $\ran \pmb{\Pi}_{\lambda}$ is separated “spectrally” from the unphysical space of zero modes $\Gphys$. 
\begin{proposition}\label{physical:prop:divergence_cond_Pi_lambda}
	Under the conditions of Proposition~\ref{sapt:prop:projection},
	$\pmb{\Pi}_{\lambda} \, \mathbf{P}_{\lambda} = \pmb{\Pi}_{\lambda} + \order_{\norm{\cdot}}(\lambda^{\infty})$ holds. 
\end{proposition}
\begin{proof}
	We will equivalently prove $\Pi_{\lambda} \, Q_{\lambda} = \order_{\norm{\cdot}}(\lambda^{\infty})$ in the $\lambda$-independent representation on $\Zak \Hper$ instead. The Gap Condition~\ref{intro:assumption:gap_condition} stipulates that none of the relevant bands are ground state bands in the sense of \cite[Definition~3.6]{DeNittis_Lein:adiabatic_periodic_Maxwell_PsiDO:2013} and thus, by \cite[Proposition~1.4~(iii)]{DeNittis_Lein:adiabatic_periodic_Maxwell_PsiDO:2013}, there exists a neighborhood $\mathcal{U} \subset \R$ of $\omega = 0$ such that $\specrel \cap \mathcal{U} = \emptyset$ where $\specrel := \bigcup_{(r,k) \in \R^6} \tau(r) \, \specrel(k)$. Moreover, let $\chi : \R \longrightarrow [0,1]$ be a smoothened bump function so that $\chi \vert_{\specrel} = 1$ and $\chi \vert_{\mathcal{U}} = 0$. Since we are considering only a finite number of relevant bands, we can choose $\chi$ as to have compact support. Then we can write 
	\begin{align*}
		\chi(M_{\lambda}) &= \frac{1}{\pi} \int_{\C} \dd z \wedge \dd \bar{z} \, \partial_{\bar{z}} \tilde{\chi}(z) \, \bigl ( M_{\lambda} - z \bigr )^{-1}
	\end{align*}
	using the Helffer-Sjöstrand formula \cite[Proposition~7.2]{Helffer_Sjoestrand:mag_Schroedinger_equation:1989} which involves choosing a pseudoanalytic extension $\tilde{\chi}$ of $\chi$ (see \eg \cite[equation~(2)]{Davies:functional_calculus:1995} or \cite[Chapter~8]{Dimassi_Sjoestrand:spectral_asymptotics:1999}). Up to errors of arbitrarily small order in $\lambda$, we can use the local Moyal resolvent from Lemma~\ref{sapt:lem:existence_Moyal_resolvent} to locally define 
	\begin{align}
		\chi_{\Weyl}(r,k) := \frac{1}{\pi} \int_{\C} \dd z \wedge \bar{z} \, \partial_{\bar{z}} \tilde{\chi}(z) \, \bigl ( R_{\lambda}(z) \bigr )(r,k) 
		. 
		\label{sapt:eqn:chi_symbol}
	\end{align}
	Hence, we can find a resummation in $\SemiHoermeq{0}{0} \bigl ( \mathcal{B}(\HperT) \bigr )$ which we will also denote with $\chi_{\Weyl}$, and by definition this resummation satisfies 
	\begin{align*}
		\chi(M_{\lambda}) &= \Op_{\lambda}(\chi_{\Weyl}) + \order_{\norm{\cdot}}(\lambda^{\infty}) 
	\end{align*}
	and consequently 
	\begin{align*}
		\Pi_{\lambda} \, \chi(M_{\lambda}) = \Op_{\lambda} \bigl ( \pi_{\lambda} \Weyl \chi_{\Weyl} \bigr ) + \order_{\norm{\cdot}}(\lambda^{\infty})
		. 
	\end{align*}
	We now use the Helffer-Sjöstrand formalism on the level of symbols to express the projection similar to \eqref{sapt:eqn:chi_symbol}: assume $\mathcal{V}$ is a neighborhood of $(r_0,k_0)$ where we can use a single contour $\Lambda$ to enclose only $\tau(r) \, \specrel(k)$ for all $(r,k) \in \mathcal{V}$. Let $\eta : \R \longrightarrow [0,1]$ be a compactly supported smooth bump function such that $\eta \vert_{\tau(r) \, \specrel(k)} = 1$ holds for all $(r,k) \in \mathcal{V}$ and $\supp \eta$ is contained in the interior of $\Lambda$. Then we can write 
	\begin{align*}
		\pi_{\lambda}(r,k) &= \frac{1}{\pi} \int_{\C} \dd z \wedge \dd \bar{z} \, \partial_{\bar{z}} \tilde{\eta}(z) \, \bigl ( R_{\lambda}(z) \bigr )(r,k) 
		+ \order(\lambda^{\infty})
		\\
		&
		=: \eta_{\Weyl}(r,k) + \order(\lambda^{\infty})
	\end{align*}
	as a complex integral in a neighborhood of $(r_0,k_0)$, and a straightforward adaptation of the arguments in the proof of \cite[Theorem~4]{Davies:functional_calculus:1995} as well as the local nature of the asymptotic expansion of $\Weyl$ yield 
	\begin{align*}
		\bigl ( \pi_{\lambda} \Weyl \chi_{\Weyl} \bigr )(r,k) &= \bigl ( \eta_{\Weyl} \Weyl \chi_{\Weyl} \bigr )(r,k) + \order(\lambda^{\infty})
		\\
		&= \frac{1}{\pi} \int_{\C} \dd z \wedge \dd \bar{z} \, \partial_{\bar{z}} \widetilde{(\eta \, \chi)}(z) \, \bigl ( R_{\lambda}(z) \bigr )(r,k) 
		+ \order(\lambda^{\infty})
		. 
	\end{align*}
	Seeing as $\chi \vert_{\supp \eta} = 1$, we deduce 
	\begin{align*}
		\pi_{\lambda} \Weyl \chi_{\Weyl} = \pi_{\lambda} + \order(\lambda^{\infty}) 
		\in \Hoermeq{0}{0} \bigl ( \mathcal{B}(\HperT) \bigr ) 
	\end{align*}
	and thus also 
	\begin{align*}
		\Pi_{\lambda} \, \chi(M_{\lambda}) &= \Op_{\lambda} \bigl ( \pi_{\lambda} \Weyl \chi_{\Weyl} \bigr ) + \order_{\norm{\cdot}}(\lambda^{\infty})
		\\
		&= \Op_{\lambda}(\pi_{\lambda}) + \order_{\norm{\cdot}}(\lambda^{\infty})
		= \Pi_{\lambda} + \order_{\norm{\cdot}}(\lambda^{\infty})
		. 
	\end{align*}
	Furthermore, from $M_{\lambda} \, Q_{\lambda} = 0$ we obtain 
	\begin{align*}
		Q_{\lambda} &= 1_{\{ 0 \}}(M_{\lambda}) \, Q_{\lambda}
		. 
	\end{align*}
	Inserting $\chi(M_{\lambda})$ and $1_{\{ 0 \}}(M_{\lambda})$ and using that $\bigl ( \chi \, 1_{\{ 0 \}} \bigr ) (\omega) = 0$ in the sense of functions on $\R$, we get 
	\begin{align*}
		\Pi_{\lambda} \, Q_{\lambda} &= \Pi_{\lambda} \, \chi(M_{\lambda}) \, 1_{\{ 0 \}}(M_{\lambda}) \, Q_{\lambda} + \order_{\norm{\cdot}}(\lambda^{\infty}) 
		\\
		&= \Pi_{\lambda} \, \bigl ( \chi \, 1_{\{ 0 \}} \bigr )(M_{\lambda}) \, Q_{\lambda} + \order_{\norm{\cdot}}(\lambda^{\infty}) 
		\\
		&= \order_{\norm{\cdot}}(\lambda^{\infty}) 
		. 
	\end{align*}
	This concludes the proof. 
\end{proof}
%
% subsection Source-free condition (end)
% section Existence of real solutions localized in a frequency regime (end)

%% file: section_05.tex
%!TEX root = /Users/max/Dropbox/research/photonic crystals/sapt for isotropic photonic crystals/arxiv/v2/photonic sapt.tex
\section{Effective electrodynamics} % (fold)
\label{real_eff}
Now we are ready to state and prove our main result: Simply put, the results from Section~\ref{sapt} deal with the dynamical Maxwell equations \eqref{intro:eqn:traditional_Maxwell_equations_dynamics} while those in Section~\ref{physical} concern themselves with the no-source conditions \eqref{intro:eqn:traditional_Maxwell_equations_no_source} as well as with the reality of electromagnetic fields. Their combination yields 
\begin{theorem}[Effective light dynamics]\label{real_eff:thm:eff_edyn}
	Suppose Assumption~\ref{intro:assumption:eps_mu} holds, $(\eps_{\lambda},\mu_{\lambda})$ are real, and 
	that the relevant family of bands $\Omega_{\mathrm{rel}}$ satisfies the \emph{Gap Condition~\ref{intro:assumption:gap_condition}} and $\specrel(k) = - \specrel(-k)$. Then the superadiabatic projection $\pmb{\Pi}_{\lambda}$, the intertwining unitary $\mathbf{U}_{\lambda}$ and the symbol $\Msymb_{\eff}$ constructed in Theorem~\ref{sapt:thm:eff_dyn} exist and we have \emph{effective light dynamics} in the following sense: 
	\begin{enumerate}[(i)] 
		\item States in $\mathfrak{K}_{\lambda} := \ran \pmb{\Pi}_{\lambda}$ are physical states up to $\order_{\norm{\cdot}}(\lambda^{\infty})$, namely
		\begin{enumerate}[(a)]
			\item they satisfy the divergence-free conditions \eqref{intro:eqn:traditional_Maxwell_equations_no_source}, $\pmb{\Pi}_{\lambda} \, \mathbf{P}_{\lambda} = \pmb{\Pi}_{\lambda} + \order_{\norm{\cdot}}(\lambda^{\infty})$, 
			\item $\mathfrak{K}_{\lambda}$ supports real-valued solutions, \ie $\bigl [ C , \pmb{\Pi}_{\lambda} \bigr ] = \order_{\norm{\cdot}}(\lambda^{\infty})$. 
		\end{enumerate}
		\item There exist effective light dynamics generated by 
		\begin{align*}
			\mathbf{M}_{\eff} = \Zak^{-1} \, \mathbf{U}_{\lambda}^{-1} \, \Op_{\lambda}(\Msymb_{\eff}) \, \mathbf{U}_{\lambda} \, \Zak + \order_{\norm{\cdot}}(\lambda^{\infty}) 
		\end{align*}
		which approximate the full dynamics, 
		\begin{align*}
			\e^{- \ii t \mathbf{M}_{\lambda}} \, \pmb{\Pi}_{\lambda} \, \Re &= \Re \, \e^{- \ii t \mathbf{M}_{\eff}} \, \pmb{\Pi}_{\lambda} \, \Re + \order_{\norm{\cdot}}(\lambda^{\infty})
			,
		\end{align*}
		and the subspace $\mathfrak{K}_{\lambda}$ is left invariant up to errors of order $\order(\lambda^{\infty})$ in norm.
	\end{enumerate}
\end{theorem}
\begin{proof}
	In order to be able to apply Theorem~\ref{sapt:thm:eff_dyn}, we need to prove the existence of a symbol $u_0 \in \Hoermeq{0}{0} \bigl ( \mathcal{B}(\HperT) \bigr )$. By our discussion in the proof of Proposition~\ref{sapt:prop:unitary}, this is equivalent to the triviality of the Bloch bundle associated to $\hat{\pi}_0$. However, according to Proposition~\ref{physical:prop:equivalences_reality_condition} the restriction to real weights and the condition $\specrel(-k) = - \specrel(k)$ is in fact equivalent to $\Zak^{-1} \, \ran \hat{\pi}_0$ supporting real states, and hence, the Bloch bundle is trivial by Proposition~\ref{physical:prop:triviality_bundle}. This means the assumptions of Theorem~\ref{sapt:thm:eff_dyn} are satisfied and the superadiabatic projection $\pmb{\Pi}_{\lambda}$, the intertwining unitary $\mathbf{U}_{\lambda}$ and the symbol $\Msymb_{\eff}$ of the effective Maxwell operator exist, and have the enumerated properties. 
	
	Moreover, also the assumptions of Propositions~\ref{physical:prop:real_solutions_Pi_lambda} and \ref{physical:prop:divergence_cond_Pi_lambda} are also satisfied, and states in $\mathfrak{K}_{\lambda}$ are physical states up to $\order(\lambda^{\infty})$. The remaining items are just a restatement of equation~\eqref{sapt:eqn:effective_dynamics_approximate_phys}. 
\end{proof}
\begin{corollary}\label{real_eff:cor:no_eff_single_band_dynamics}
	Suppose the conditions of Theorem~\ref{real_eff:thm:eff_edyn} are satisfied. Then there cannot be effective one-band dynamics describing the evolution of \emph{physical} states. 
\end{corollary}
\begin{proof}
	The one-band case is excluded, because the condition $\specrel(-k) = - \specrel(k)$ implies that $\abs{\Index}$ is even: If $\omega_n(k)$ is contained in $\specrel(k)$, then also its symmetric twin $- \omega_n(-k) \in \specrel(k)$ is in the relevant bands so that positive-negative frequency bands come in pairs. 
\end{proof}
This Corollary is \emph{very} significant with regard to existing work: since single bands cannot support real states, effective single-band dynamics \emph{cannot} describe the evolution of physical states. Instead, one always needs to consider linear combinations of counter-propagating waves (\cf Section~\ref{twin_bands}).

\subsection{A Peierls substitution for the Maxwell operator} % (fold)
\label{twin_bands:Harper}
To get a better idea what explicit expressions for for $\Msymb_{\eff}$ constructed in Theorem~\ref{real_eff:thm:eff_edyn} look like, let us make an attempt to find on in the simplest possible situation. In general, there are two obstacles: (1) we are dealing with a genuine multiband problem and (2) each contiguous family of bands may have a non-trivial topology. Hence, even if $\Omega_{\mathrm{rel}}$ consists of isolated bands, each of these bands may not be geometrically trivial. However, if we impose triviality of the single-band Bloch bundles as an additional assumption, we obtain an explicit expression for $\Msymb_{\eff}$ (\cf Section~\ref{technical:ray_optics} for details of the computation). 
\begin{corollary}\label{real_eff:cor:simplest_effective_Maxwellian}
	In addition to the assumptions in Theorem~\ref{real_eff:thm:eff_edyn}, let us suppose $\Omega_{\mathrm{rel}}$ consists of isolated bands and all associated single-band Bloch bundles are trivial. Then the first two terms of the effective Maxwellian $\Msymb_{\eff} = \Msymb_{\eff,0} + \lambda \, \Msymb_{\eff,1} + \order(\lambda^2)$ are given by equations~\eqref{intro:eqn:Meff_0_simple} and \eqref{intro:eqn:Meff_1_simple}. 
\end{corollary}
Clearly, to require that each isolated band carries a trivial topology is not an innocent fact. However, under this assumption Corollary~\ref{real_eff:cor:simplest_effective_Maxwellian} allows us to explicitly construct a simple effective model in the spirit of the Peierls substitution. As a matter of fact, reduced models of this kind are quite amenable to direct analytic study or numerical simulations. 

For the sake of concreteness, let us analyze the case of topologically trivial twin bands where $\Omega_{\mathrm{rel}} = \{ \omega_{\pm} \}$ for $\omega_{\pm}(k) = \pm \freqband(\pm k)$. Expressing the band function 
\begin{align*}
	\freqband(k) = \sum_{\gamma \in \Gamma} \widehat{\freqband}(\gamma) \; \e^{+ \ii \gamma \cdot k} 
\end{align*}
in terms of its Fourier series leads to a matrix-valued symbol with principal part 
\begin{align*}
	\Msymb_{\eff,0}(r,k) = \tau(r) \; 
	\left (
	\begin{matrix}
		\freqband(k) & 0 \\
		0 & - \freqband(-k) \\
	\end{matrix}
	\right ) 
	. 
\end{align*}
Let us now assume that $\tau$ is periodic on a much larger scale. This large-scale periodicity may be the result of extending a finite sample spanned by $L = \{ L_1 \, e_1 , L_2 \, e_2 , L_3 \, e_3 \}$ to all of $\R^3$ using periodic boundary conditions, for instance. Hence, we can expand $\tau$ in terms of the Fourier decomposition with respect to the dual lattice $\Gamma_L^*$ of the macroscopic lattice $\Gamma_L$, 
\begin{align*}
	\tau(r) &= \sum_{\gamma_L^* \in \Gamma_L^*} \widehat{\tau}(\gamma_L) \; \e^{+ \ii \gamma_L^* \cdot r}
	. 
\end{align*}
To leading order, the quantization of $\Msymb_{\eff,0}$ yields the selfadjoint operator 
\begin{align}
	M_{\mathrm{eff},0} = \frac{\tau(\ii \lambda \nabla_k)}{2} \; \left (
	\begin{matrix}
		\freqband(\hat{k}) & 0 \\
		0 & - \freqband(-\hat{k}) \\
	\end{matrix}
	\right ) + \mathrm{c.c.} 
	\label{twin_band:eqn:Harper_Maxwell_operator}
\end{align}
on $L^2(\T^3) \otimes \C^2$. This operator can be expressed in terms of shifts on real and reciprocal space, \ie 
\begin{align*}
	\hat{U}_j := \e^{+ \ii \hat{k}_j} 
\end{align*}
and, if we use that $\Gamma_L$ is a scaled version of $\Gamma$, 
\begin{align*}
	\hat{T}_j := \e^{- \frac{\lambda}{L_j} \partial_{k_j}} 
	, 
	\qquad \qquad 
	\bigl ( \hat{T}_j \psi \bigr )(k) = \psi \bigl ( k + \tfrac{\lambda}{L_j} \, e_j^{\ast} \bigr ) 
	\, . 
\end{align*}
The six unitary operators are characterized by the commutation relations 
\begin{align*}
	 \hat{T}_l \hat{U}_j= \e^{\ii \frac{\lambda}{L_j} \delta_{lj}}
	 \; \hat{U}_j \hat{T}_i
	 \, ,
	 \qquad 
	 [\hat{T}_l,\hat{T}_j]= 0 = [\hat{U}_l,\hat{U}_j]
	 \,,
	 \qquad 
	 l,j=1,2,3 
	 , 
\end{align*}
and so they generate a representation of a six-dimensional non-commutative torus on $L^2(\T^3)$ \cite[Chapter~12]{Varilly_Figueroa_Gracia_Bondia:noncommutative_geometry:2001}. 

Let us denote with $\mathcal{A}^6(\nicefrac{\lambda}{L})$ the $C^\ast$-algebra generated by $\hat{U}_j$ and $\hat{T}_j$ on $L^2(\T^3)$. We have shown that the effective models for the Maxwell dynamics in the twin bands case can be associated with a diagonal representative of the non-commutative torus  $\mathcal{A}^6(\nicefrac{\lambda}{L}) \otimes \mathrm{Mat}_{\C}(2)$.
This analogy allows us to apply all the well-known results about the theory of Harper operators to the \emph{Harper-Maxwell operator~\eqref{twin_band:eqn:Harper_Maxwell_operator}}. For instance, one can expect to recover the typical Hofstadter's butterfly-like spectrum 
which produces a splitting of the two topologically trivial bands spanned by $\omega_\pm$ in subbands which can carry a non-trivial topology. We stress that in this case the non trivial effect is due only to an incommensurability between the perturbation parameter $\lambda$ and the typical sizes $L_j$ of the lattice without any magnetic effect.
Moreover, if one removes the condition of the triviality of the single band Bloch bundles, one easily deduces that $M_{\eff,0}$ can be associated with a non-diagonal element of $\mathcal{A}^6(\nicefrac{\lambda}{L}) \otimes \mathrm{Mat}_{\C}(2)$.
% subsection Peierl's substitution (end)

\subsection{Effective dynamics for observables} % (fold)
\label{real_eff:observables}
So far, we have only derived effective electrodynamics for \emph{states}; for quantum mechanics, this immediately implies effective dynamics for the time-evolved \emph{observables} as well. However, in electrodynamics, observables are not selfadjoint operators on a Hilbert space, they are continuous functionals 
\begin{align*}
	\mathcal{F} : L^2(\R^3,\R^6) \longrightarrow \R 
\end{align*}
from the space of real-valued fields to the real numbers. Clearly, approximating the time evolution of generic observables is out of reach, but many physically relevant examples are quadratic in the fields, \eg the energy density 
\begin{align*}
	\mathcal{E}(\mathbf{E},\mathbf{H}) := \tfrac{1}{2} \bnorm{(\mathbf{E},\mathbf{H})}_{\Hil_{\lambda}}^2 
	, 
\end{align*}
or components of the space-average of the Poynting vector 
\begin{align*}
	\mathcal{S}(\mathbf{E},\mathbf{H}) := \int_{\R^3} \dd x \, \mathbf{E}(x) \times \mathbf{H}(x) 
	. 
\end{align*}
In these particular cases, we can link quadratic observables $\mathcal{F}$ to a bounded operator $F$ on the Hilbert space $\Hil_{\lambda}$ by writing 
\begin{align}
	\mathcal{F}(\mathbf{E},\mathbf{H}) &= \bscpro{(\mathbf{E},\mathbf{H})}{F (\mathbf{E},\mathbf{H})}_{\Hil_{\lambda}} 
	\label{real_eff:eqn:quadratic_observables}
\end{align}
as expectation value. For such observables, we immediately obtain the following 
\begin{corollary}
	Under the assumptions of Theorem~\ref{real_eff:thm:eff_edyn}, we also obtain effective dynamics for observables of the form \eqref{real_eff:eqn:quadratic_observables}, 
	\begin{align*}
		\mathcal{F} \left ( \e^{- \ii t \Mphys_{\lambda}} (\mathbf{E},\mathbf{H}) \right ) = \mathcal{F} \left ( \Re \e^{- \ii t \mathbf{M}_{\eff}} (\mathbf{E},\mathbf{H}) \right ) + \order(\lambda^{\infty})
		, 
		&&
		(\mathbf{E},\mathbf{H}) \in \ran \pmb{\Pi}_{\lambda}
		. 
	\end{align*}
\end{corollary}
%
% subsection effective_dynamics_for_observables (end)
% section Recovering real solutions for the effective dynamics (end)

%% file: section_06.tex
%!TEX root = /Users/max/Dropbox/research/photonic crystals/sapt for isotropic photonic crystals/arxiv/v2/photonic sapt.tex
\section{The problem of deriving ray optics equations} % (fold)
\label{twin_bands}
At a glance, a derivation of ray optics equations seems just within reach. For instance, one can adapt \cite{Stiepan_Teufel:semiclassics_op_valued_symbols:2012} to the present context and obtain 
\begin{theorem}[Single-band ray optics]\label{twin_bands:thm:single_band}
	Suppose Assumption~\ref{intro:assumption:eps_mu} hold true and that $\freqband$ is an isolated band in the sense of Assumption~\ref{intro:assumption:gap_condition} with Bloch function $\blochf = (\blochf^E,\blochf^H)$. 
	
	Let $\Phi^{\lambda}$ be the flow associated to 
	\begin{align}
		\left [ \left (
		\begin{matrix}
			0 & - \id_{\R^3} \\
			+ \id_{\R^3} & 0 \\
		\end{matrix}
		\right ) - \lambda \; \left (
		\begin{matrix}
			\pmb{\Omega}^{rr} & \pmb{\Omega}^{rk} \\
			\pmb{\Omega}^{kr} & \pmb{\Omega}^{kk} \\
		\end{matrix}
		\right ) \right ] \left (
		\begin{matrix}
			\dot{r} \\
			\dot{k} \\
		\end{matrix}
		\right ) &= \left (
		\begin{matrix}
			\nabla_r \\
			\nabla_k \\
		\end{matrix}
		\right ) \pmb{\Msymb}_{\mathrm{sc}} 
		\label{twin_bands:eqn:semiclassical_eom}
	\end{align}
	that is defined in terms of the components of the \emph{extended Berry curvature} 
	\begin{align}
		\pmb{\Omega}_{jl}^{rr} &= - \ii \trace_{\HperT} \bigl ( \pmb{\pi}_0 \, \bigl [ \partial_{r_j} \pmb{\pi}_0 \, , \, \partial_{r_l} \pmb{\pi}_0 \bigr ] \bigr ) 
		,
		\notag \\
		\pmb{\Omega}_{jl}^{kk} &= - \ii \trace_{\HperT} \bigl ( \pmb{\pi}_0 \, \bigl [ \partial_{k_j} \pmb{\pi}_0 \, , \, \partial_{k_l} \pmb{\pi}_0 \bigr ] \bigr ) 
		,
		\notag \\
		\pmb{\Omega}_{jl}^{rk} &= - \ii \trace_{\HperT} \bigl ( \pmb{\pi}_0 \, \bigl [ \partial_{r_j} \pmb{\pi}_0 \, , \, \partial_{k_l} \pmb{\pi}_0 \bigr ] \bigr ) 
		= - \pmb{\Omega}^{kr}_{lj}
		, 
		\label{twin_bands:eqn:Berry_curvature}
	\end{align}
	and the \emph{semiclassical Maxwellian} 
	\begin{align}
		\pmb{\mathcal{M}}_{\mathrm{sc}} &= \tau \, \omega + \lambda \, \trace_{\HperT} \bigl ( \pmb{\pi}_0 \, \pmb{\Msymb}_1 \bigr ) 
		+ \lambda \, \tfrac{\ii}{2} \, \trace_{\HperT} \Bigl ( \bigl \{ \pmb{\pi}_0 \vert \pmb{\Msymb}_0 \vert \pmb{\pi}_0 \bigr \} \Bigr ) 
		. 
		\label{twin_bands:eqn:semiclassical_Maxwellian}
	\end{align}
	Then for any periodic semiclassical symbol $\mathbf{f} \in \SemiHoermper{0}{0}(\C)$ 
	\begin{align}
		\Bnorm{\pmb{\Pi}_{\lambda}^{\Zak} \, \Bigl ( \e^{+ \ii t \Mphys_{\lambda}^{\Zak}} \, \Op_{\lambda}(\mathbf{f}) \, \e^{- \ii t \Mphys_{\lambda}^{\Zak}} - \Op_{\lambda} \bigl ( \mathbf{f} \circ \Phi^{\lambda} \bigr ) \Bigr ) \, \pmb{\Pi}_{\lambda}^{\Zak}}_{\mathcal{B}(\Zak \Hil_{\lambda})} 
		&= \order(\lambda^2)
	\end{align}
	holds uniformly on bounded time intervals. 
\end{theorem}
The technical modifications to the proofs in \cite{Stiepan_Teufel:semiclassics_op_valued_symbols:2012} are straightforward, but we will postpone a proof to \cite{DeNittis_Lein:ray_optics_photonic_crystals:2013}. 

However, \emph{Theorem~\ref{twin_bands:thm:single_band} does not answer the physical question} as to the correct form of ray optics equations.  In contrast to the Schrödinger case, deriving effective light dynamics is a genuine \emph{multi}band problem since electromagnetic fields are real (Corollary~\ref{real_eff:cor:no_eff_single_band_dynamics}). In the simplest physically relevant case the material weights $(\eps,\mu)$ are real and $\Omega_{\mathrm{rel}}$ consists of a single non-degenerate band $\omega_+(k) > 0$ and its symmetric twin $\omega_-(k) = - \omega_+(-k)$ (\eg the two yellow bands $\omega_{\pm n_2}$ in Figure~\ref{setup:fig:band_picture}). Seeing as $\omega_+$ and $\omega_-$ individually are isolated bands, the corresponding positive and negative frequency superadiabatic proections $\pmb{\Pi}_{\pm \, \lambda}$ exist separately (Proposition~\ref{sapt:prop:projection}) and decompose 
\begin{align}
	\pmb{\Pi}_{\lambda} = \pmb{\Pi}_{+ \, \lambda} + \pmb{\Pi}_{- \, \lambda} + \order_{\norm{\cdot}}(\lambda^{\infty})
	. 
	\label{twin_bands:eqn:projection_splitting}
\end{align}
Moreover, a closer inspection of the proof of Lemma~\ref{technical:lem:c_pi_pi} reveals that the projections associated to the positive and negative frequency bands are related by 
\begin{align}
	C \, \pmb{\Pi}_{+ \, \lambda} \, C = \pmb{\Pi}_{- \, \lambda} + \order_{\norm{\cdot}}(\lambda^{\infty}) 
	. 
	\label{twin_bands:eqn:projection_pm_relation}
\end{align}
The assumption that the weights are real assures $\bigl [ \Re , \pmb{\Pi}_{\lambda} \bigr ] = \order_{\norm{\cdot}}(\lambda^{\infty})$ (Proposition~\ref{physical:prop:real_solutions_Pi_lambda}) and we can split the dynamics into two contributions, 
\begin{align}
	\e^{- \ii t \Mphys_{\lambda}} \, \pmb{\Pi}_{\lambda} \, \Re &= \Re \negthinspace \left ( \e^{- \ii t \Mphys_{\lambda}} \, \pmb{\Pi}_{+ \, \lambda} 
	+ \e^{- \ii t \Mphys_{\lambda}} \, \pmb{\Pi}_{- \, \lambda} \right ) \Re 
	+ \order_{\norm{\cdot}}(\lambda^{\infty}) 
	. 
	\label{twin_bands:eqn:one_two_band_two_one_band_problem}
\end{align}
Assuming the goal is to prove an Egorov-type theorem for a suitable observable $\Op_{\lambda}(\mathbf{f})$, $\mathbf{f} \in \SemiHoermper{0}{0}(\C)$, we 
can split the time-evolved observable 
\begin{align*}
	\Re \, \pmb{\Pi}_{\lambda} \, \e^{+ \ii \frac{t}{\lambda} \Mphys_{\lambda}} \, &\Zak^{-1} \, \Op_{\lambda}(\mathbf{f}) \, \Zak \, \e^{- \ii \frac{t}{\lambda} \Mphys_{\lambda}} \, \pmb{\Pi}_{\lambda} \, \Re 
	= \\
	&= \mathbf{F}_+(t) + \mathbf{F}_-(t) + \mathbf{F}_{\mathrm{int}}(t) + \order_{\norm{\cdot}}(\lambda^{\infty})
\end{align*}
into three parts, the two \emph{single band terms} associated to $\omega_+$ and $\omega_-$, 
\begin{align*}
	\mathbf{F}_\pm(t):= \Re \, \pmb{\Pi}_{\pm \, \lambda} \, \e^{+ \ii \frac{t}{\lambda} \Mphys_{\lambda}} \, \Zak^{-1} \, \Op_{\lambda}(\mathbf{f}) \, \Zak \, \e^{- \ii \frac{t}{\lambda} \Mphys_{\lambda}} \, \pmb{\Pi}_{\pm \, \lambda} \, \Re 
	, 
\end{align*}
and an \emph{intraband interference term} 
\begin{align}
	\mathbf{F}_{\rm int}(t) := &\;
	\phantom{+\ \ }\Re \, \pmb{\Pi}_{+ \, \lambda} \, \e^{+ \ii \frac{t}{\lambda} \Mphys_{\lambda}} \, \Zak^{-1} \, \Op_{\lambda}(\mathbf{f}) \, \Zak \, \e^{- \ii \frac{t}{\lambda} \Mphys_{\lambda}} \, \pmb{\Pi}_{- \, \lambda} \, \Re
	+ \notag \\
	 &\;+\Re \, \pmb{\Pi}_{- \, \lambda} \, \e^{+ \ii \frac{t}{\lambda} \Mphys_{\lambda}} \, \Zak^{-1} \, \Op_{\lambda}(\mathbf{f}) \, \Zak \, \e^{- \ii \frac{t}{\lambda} \Mphys_{\lambda}} \, \pmb{\Pi}_{+ \, \lambda} \, \Re
	. 
	\label{twin_bands:eqn:4_term_semiclassics}
\end{align}
While the two single-band terms $\mathbf{F}_{\pm}$ are covered by Theorem~\ref{twin_bands:thm:single_band}, effective dynamics for the \emph{cross term} $\mathbf{F}_{\rm int}(t)$ have yet to be established. 
\medskip

\noindent
To summarize the discussion: Even in the simplest case, the \emph{twin band case} with topologically trivial bands, the principal symbol of the effective Maxwellian (\cf Corollary~\ref{real_eff:cor:simplest_effective_Maxwellian}) is \emph{not} scalar. Hence, we are in a setting that is not covered by existing semiclassical techniques (\eg \cite[Chapter~3.4.1]{Teufel:adiabatic_perturbation_theory:2003} or \cite{Gat_Lein_Teufel:semiclassical_dynamics_particle_spin:2013}), and
a derivation of physically relevant ray optics equations necessitates the development of new techniques which are beyond the scope of the present work. 
% section The twin band case: ray optics (end)

%% file: section_07.tex
%!TEX root = /Users/max/Dropbox/research/photonic crystals/sapt for isotropic photonic crystals/arxiv/v2/photonic sapt.tex
\section{Discussion, comparison with literature and outlook} % (fold)
\label{discussion}
Physically speaking, there are three frequency ranges one needs to distringuish when studying the propagation of light in photonic crystals: the \emph{low} and \emph{high} frequency regimes and the range of \emph{intermediate} frequencies. 

Our first main result, Theorem~\ref{real_eff:thm:eff_edyn}, proves the existence of effective dynamics for \emph{physical} states from the \emph{intermediate frequency regime} which approximate the full light dynamics up to \emph{arbitrarily small error} $\order(\lambda^N)$ and \emph{any time scale} $\order(\lambda^{-K})$, $N , K \geq 0$. We have systematically exploited the reformulation of the source-free Maxwell equations as a Schrödinger-type equation and adapted techniques initially developed for perturbed periodic Schrödinger operators. Technically, our assumptions only exclude the low frequency regime (\cf part (iii) of the Gap Condition~\ref{intro:assumption:gap_condition}), but our result is less relevant from a physical point of view for light of very \emph{high} frequency. This is because the Maxwell equations in media \eqref{intro:eqn:traditional_Maxwell_equations_dynamics}--\eqref{intro:eqn:traditional_Maxwell_equations_no_source} are an \emph{effective} theory for light whose \emph{in vacuo} wavelength $\lambda_{\mathrm{vac}} = \frac{c}{2 \pi \omega}$ has to be much larger than the average spacing between atoms (see \eg \cite{Jackson:electrodynamics:1998}). Only then can the net effect of the microscopic charges be described by the phenomenological quantities electric permittivity $\eps$ and magnetic permeability $\mu$. 

The second main insight of this paper is that deriving effective dynamics for \emph{real} initial states is a \emph{genuine multiband problem}. We illustrate this for non-gyroscopic photonic crystals (\ie $(\eps,\mu)$ are real) in the absence of perturbations: assume $\omega_+ > 0$ is an isolated, non-degenerate band with Bloch function $\varphi_+$. Then in view of equation~\eqref{physical:eqn:relation_CZak_C} the real and imaginary part of the Bloch wave associated to $\omega_+$
\begin{align}
	\Psi_{\Re}(k) :& \negmedspace= \sqrt{2} \, \bigl ( \Re^{\Zak} \varphi_+ \bigr )(k) 
	% \notag \\
	% &
	% = \frac{1}{\sqrt{2}} \Bigl ( \varphi_+(k) + \overline{\varphi_+(-k)} \Bigr ) 
	= \frac{1}{\sqrt{2}} \Bigl ( \varphi_+(k) + \varphi_-(k) \Bigr ) 
	, 
	\label{discussion:eqn:real_imaginary_Bloch_wave}
	\\
	\Psi_{\Im}(k) :& \negmedspace= \sqrt{2} \, \bigl ( \Im^{\Zak} \varphi_+ \bigr )(k) 
	% = \frac{1}{\ii \sqrt{2}} \Bigl ( \varphi_+(k) - \overline{\varphi_+(-k)} \Bigr ) 
	= \frac{1}{\ii \sqrt{2}} \Bigl ( \varphi_+(k) - \varphi_-(k) \Bigr ) 
	, 
	\notag
\end{align}
are linear combinations of $\varphi_+(k)$ and the Bloch function $\varphi_-(k) = \overline{\varphi_+(-k)}$ of the symmetric twin band $\omega_-(k) = - \omega_+(-k)$ (\cf equation~\eqref{setup:eqn:symmetric_twin_Bloch_function}). 
We can read off the relation $\Psi_{\Re,\Im}(-k) = C \Psi_{\Re,\Im}(k)$ from their definition, and thus the arguments in the proof of Corollary~\ref{physical:cor:spectral_condition_reality} imply the triviality of $\blochb \bigl ( \sopro{\Psi_{\Re,\Im}}{\Psi_{\Re,\Im}} \bigr )$ even if $\blochb \bigl ( \sopro{\varphi_{\pm}}{\varphi_{\pm}} \bigr )$ is non-trivial. 
This rigorously justifies the absence of topological effects in non-gyrotropic photonic crystals. Moreover, the time evolution of real and imaginary part 
\begin{align*}
	\e^{- \ii t \Mper(k)} \Psi_{\Re}(k) &= \frac{1}{\sqrt{2}} \Bigl (  \e^{- \ii t \omega_+(k)} \varphi_+(k) + \e^{+ \ii t \omega_+(-k)} \overline{\varphi_+(-k)} \Bigr ) 
	\\
	\e^{- \ii t \Mper(k)} \Psi_{\Im}(k) &= \frac{1}{\ii \sqrt{2}} \Bigl (  \e^{- \ii t \omega_+(k)} \varphi_+(k) - \e^{+ \ii t \omega_+(-k)} \overline{\varphi_+(-k)} \Bigr ) 
\end{align*}
are the sums of two counterpropagating \emph{complex} waves: not only is Bloch momentum reversed, also the sense of phase rotation for $\varphi_{\pm}$ differs.

\subsection{Existing literature} % (fold)
\label{discussion:literature}
To the best of our knowledge, the only other rigorous result for the \emph{intermediate} frequency regime, sometimes also referred to as \emph{resonant regime}, is \cite{Allaire_Palombaro_Rauch:diffractive_Bloch_wave_packets_Maxwell:2012}. Even if one forgoes mathematical rigor, we only know of two other derivations in the physics literature \cite{Onoda_Murakami_Nagaosa:geometrics_optical_wave-packets:2006,Esposito_Gerace:photonic_crystals_broken_TR_symmetry:2013}. 

All of these previous results cover only the unphysical single-band case. For instance, Raghu and Haldane \cite[equations~(42)--(43)]{Raghu_Haldane:quantum_Hall_effect_photonic_crystals:2008} rely on the analogy to the Bloch electron and propose 
\begin{align}
	\dot{r} &= + \nabla_k \bigl ( \tau \, \omega_n \bigr ) - \lambda \, \dot{k} \times \Omega 
	\label{discussion:eqn:semiclassical_eom}
	\\
	\dot{k} &= - \nabla_r \bigl ( \tau \, \omega_n \bigr )
	\notag 
\end{align}
as semiclassical equations of motion to describe ray optics. Here, $\tau \, \omega_n$ plays the same role as the energy in the analogous equations for the Bloch electron and the \emph{Berry curvature} $\Omega := \nabla_k \wedge \ii \scpro{\varphi_n}{\nabla_k \varphi_n}_{\HperT}$ gives a geometric contribution. 
% CHANGED smoothen sentence
Note that compared to the semiclassical equations in Theorem~\ref{twin_bands:thm:single_band}, equation~\eqref{discussion:eqn:semiclassical_eom} is missing several $\order(\lambda)$ terms. 

Gerace and Esposito have investigated how to derive equation~\eqref{discussion:eqn:semiclassical_eom} from first principles \cite{Esposito_Gerace:photonic_crystals_broken_TR_symmetry:2013}: their rather elegant derivation relies solely on standard perturbation theory, but reproduces only the equation for $\dot{r}$. 

Onodoa et.~al's equations of motion \cite{Onoda_Murakami_Nagaosa:geometrics_optical_wave-packets:2006} are more involved, because they include \emph{degenerate} bands in their discussion. Hence, in addition to extra terms in equation~\eqref{discussion:eqn:semiclassical_eom}, there is an additional equation of motion for the $\C^N$ degree of freedom which describes the degeneracy. 

Lastly, let us compare our setting to that of \cite{Allaire_Palombaro_Rauch:diffractive_Bloch_wave_packets_Maxwell:2012}: while Allaire et.~al include time-dependent material weights and effects of dissipation which are beyond the present scope, their perturbations to the weights are \emph{much} weaker compared to \eqref{intro:eqn:slow_modulation_material_constants}. Indeed, their Hypothesis~1.1 stipulates that the material weights $(\eps_{\lambda},\mu_{\lambda})$ are of the form 
\begin{align}
	f_{\lambda}(x) &= f_0(x) + \lambda^2 \, f_1(\lambda x,x) 
	= \left ( 1 + \lambda^2 \, \frac{f_1(\lambda x,x)}{f_0(x)} \right ) \, f_0(x)
	\label{discussion:eqn:Allaire_weights}
\end{align}
where $f_0$ and $f_1$ are $\Gamma$-periodic in the last argument. Then away from band crossings \cite[Hypothesis~1.4]{Allaire_Palombaro_Rauch:diffractive_Bloch_wave_packets_Maxwell:2012}, they use a multi-scale WKB ansatz to approximate the full light dynamics to leading order for times of $\order(\lambda^{-1})$. Given the semiclassical nature of WKB approaches, the limitation to the \emph{diffractive time scale} (corresponding to the semiclassical time scale) is only natural. However, the perturbation is of the same order of magnitude as the error in Egorov-type theorems, namely $\order(\lambda^2)$ \cite[Proposition~4]{PST:effective_dynamics_Bloch:2003}. Modulated weights of the form \eqref{discussion:eqn:Allaire_weights} are also too small to include the cases studied by physicsts: the proposed semiclassical equations of motion and as well as $\Msymb_{\eff , 0}$ feature $\tau = \tau_{\eps} \, \tau_{\mu}$ in the leading-order term (\cf equation \eqref{discussion:eqn:semiclassical_eom}). 
\medskip

\noindent
All of these works consider only the single-band case, and ray optics equations akin to \eqref{discussion:eqn:semiclassical_eom} suggest the existence of topological effects in non-gyrotropic photonic crystals: unlike in the Schrödinger case, \emph{a priori} the symmetry induced by complex conjugation does \emph{not} ensure the triviality of the associated Bloch bundle, and hence, the Chern numbers which can be computed from the Berry curvature $\Omega$ need not be zero (\cf Remark~\ref{physical:rem:Bloch_bundle}). 
These contradictory conclusions are not the consequence of a lack of mathematical rigor (we \emph{could} have worked out the details of Theorem~\ref{twin_bands:thm:single_band} here with moderate effort), but because states supported in a single band are unphysical. 
% subsection Existing literature (end)

\subsection{Future research} % (fold)
\label{discussion:future}
Both, from the vantage point of theoretical and mathematical physics, studying the dynamical properties of photonic crystals is an open field, and this work is only a first step. We intend to dive deeper in future research, and so we finish this work by outlining some of the most promising and interesting directions.

\subsubsection*{Twin-band ray optics} % (fold)
\label{discussion:future:twin_band_ray_optics}
The next problem we intend to tackle is the rigorous derivation of ray optics equations in photonic crytals \cite{DeNittis_Lein:ray_optics_photonic_crystals:2013}. 
As we have argued in Section~\ref{twin_bands}, a naïve application of well-known semiclassical techniques gives at best an incomplete picture, because they cover only two of the three terms in equation~\eqref{twin_bands:eqn:4_term_semiclassics}. Only if the band transition terms are $\order(\lambda^2)$ can one arrive at a simple semiclassical picture. 
% subsubsection twin_band_ray_optics (end)

\subsubsection*{Study of effective Harper-like twin-band operators} % (fold)
\label{discussion:future:Harper}
We have motivated how to derive simple, \emph{Harper-like model operators} from the Maxwell equations in Section~\ref{twin_bands:Harper}. In case $\Omega_{\mathrm{rel}}$ consists of isolated, topologically trivial bands, we have given an explicit expression for $\Msymb_{\eff}$ (Corollary~\ref{real_eff:cor:simplest_effective_Maxwellian}). It would be interesting to find out whether and how the properties of these model operators depend on the topological triviality of the single-band Bloch bundles. Moreover, we reckon that we can modify these model operators so as to break the particle-hole symmetry, leading to classes of operators which hopefully retain some of the essential features of Maxwell operators for gyrotropic photonic crystals. We expect that our techniques developed for two-band systems \cite{DeNittis_Lein:piezo_graphene:2013} will be useful. 
% subsubsection study_of_effective_twin_band_operators (end)

\subsubsection*{Gyrotropic photonic crystals} % (fold)
\label{discussion:future:gyrotropic_photonic_crystals}
A third very interesting avenue to explore are \emph{gyrotropic} photonic crystals where the non-vanishing imaginary entries in the offdiagonal of $\eps$ and $\mu$ play a role similar to a strong magnetic field in the case of periodic Schrödinger operators -- they break a symmetry, and breaking this “time-reversal-like” symmetry is a necessary condition to allow for non-trivial topological effects \cite{Raghu_Haldane:quantum_Hall_effect_photonic_crystals:2008,Esposito_Gerace:photonic_crystals_broken_TR_symmetry:2013}. 

In that case, the time evolution group no longer commutes with the real and imaginary part operator, and initially real fields are mapped to complex-valued solutions. Moreover, the spectral symmetry described by \eqref{setup:eqn:symmetric_twin_Bloch_function} disappears so that it is not clear whether there exists a “natural” decomposition of Bloch waves into real and imaginary part. Given how central the \emph{real fields} assumption is to our arguments, this raises a slew of interesting physical questions. Moreover, the standard classification theory by Altland and Zirnbauer \cite{Altland_Zirnbauer:superconductors_symmetries:1997} suggests that topological effects should only be observable in two-dimensional photonic crystals; a derivation of this fact for gyrotropic photonic crystals could shed some insight as to why that is. 
% subsubsection gyrotropic_photonic_crystals (end)
% 
% 
% \subsubsection*{A simpler homogenization limit} % (fold)
% \label{discussion:future:homogenization}
% % 
% The dearth of results for the intermediate frequency range is in stark contrast to the low frequency or \emph{homogenization limit} (see \eg \cite{Markowich_Poupaud:Maxwell_homogenization_energy_density:1996,Sjoeberg_Engstroem_Kristensson_Wall_Wellander:Bloch-Floquet_Maxwell_homogenization:2005,Birman_Suslina:homogenization_Maxwell:2007}) where entire books have been dedicated to the subject (\eg \cite{Bakhvalov_Panasenko:homogenization:1989,Zhikov_Kozlov_Oleinik:homogenization:1994}). The exclusion of ground state bands in the Gap Condition~\ref{intro:assumption:gap_condition} is based on technical as much as physical grounds. Even though we believe we could overcome the technical obstacles by means of a suitable regularization procedure (\cf discussion in \cite[Section~3.2]{DeNittis_Lein:adiabatic_periodic_Maxwell_PsiDO:2013}), the physically interesting question in the low frequency regime is different: What are the effective electric permittivity and magnetic permeability for low frequency waves? The answer is intimately connected to the nearly linear dispersion of the ground state bands $\omega_{\pm n_1}$ in Figure~\ref{setup:fig:band_picture} near $k = 0$. 
% % subsubsection a_simpler_homogenization_limit (end)
% subsection Future research (end)
% section Discussion (end)

%% file: section_08.tex
%!TEX root = /Users/max/Dropbox/research/photonic crystals/sapt for isotropic photonic crystals/arxiv/v2/photonic sapt.tex
\section{Technical results and missing proofs} % (fold)
\label{technical}
In this section, we collect the proofs to a number of technical results. 

\subsection{Auxiliary results for Section~\ref{sapt}} % (fold)
\label{technical:sapt}
We start with a series of general results concerning the Moyal resolvent of the symbol $\Msymb_{\lambda} = \tau \, \Mper(\, \cdot \,) + \order(\lambda)$ given by \eqref{setup:eqn:M_lambda_PsiDO} where $\tau = \tau_{\eps} \, \tau_{\mu}$ is the deformation function. 
\begin{lemma}\label{sapt:lem:existence_Moyal_resolvent}
	Suppose Assumption~\ref{intro:assumption:eps_mu} holds and let $(r_0,k_0) \in \R^6$. Moreover, let $\Weyl$ be the formal sum of the Moyal product defined as the right-hand side of \eqref{setup:eqn:Moyal_asymp_expansion}. 
	\begin{enumerate}[(i)]
		\item For any $z \not\in \sigma \bigl ( \Msymb_0(r_0,k_0) \bigr )$, the \emph{local Moyal resolvent} 
		\begin{align*}
			R_{\lambda}(z) \asymp \sum_{n = 0}^{\infty} \lambda^n \, R_n(z) 
		\end{align*}
		which satisfies 
		\begin{align*}
			R_{\lambda}(z) \, \Weyl \, (\Msymb_{\lambda} - z) &= \id_{\HperT} 
			, 
			\\
			(\Msymb_{\lambda} - z) \, \Weyl \, R_{\lambda}(z) &= \id_{\HperT} 
			, 
		\end{align*}
		exists in an open neighborhood $\mathcal{V}$ of $(r_0,k_0)$. 
		\item The local Moyal resolvent satisfies the resolvent identity with respect to $\Weyl$, 
		\begin{align*}
			R_{\lambda}(z) - R_{\lambda}(z') = (z - z') \; R_{\lambda}(z) \, \Weyl \, R_{\lambda}(z') 
			, 
		\end{align*}
		and $R_{\lambda}(z)^* = R_{\lambda}(\bar{z})$ where ${}^*$ is the adjoint induced by $\scpro{\cdot \,}{\cdot}_{\HperT}$. 
		\item For any $a , b \in \N_0^3$ and $n\in\N_0$ there exist constants $C_1(z) , C_2(z) > 0$ so that 
		\begin{align*}
			\bnorm{\bigl ( \partial_r^a \partial_k^b R_n(z) \bigr )(r,k)}_{\mathcal{B}(\HperT)} &\leq C_1(z)
			\\
			\bnorm{\bigl ( \partial_r^a \partial_k^b R_n(z) \bigr )(r,k)}_{\mathcal{B}(\HperT,\domainT)} &\leq C_2(z) \, \bigl ( 1 + \sabs{k} \bigr ) 
		\end{align*}
		hold. 
		% These constants depend polynomially on $\sabs{z}$ and $\sup_{(r,k) \in \mathcal{V}} \mathrm{dist} \bigl ( \sigma \bigl ( \Msymb_0(r,k) \bigr ) , z \bigr )^{-1}$. 
		These constants depend on $\sabs{z}$ and $\sup_{(r,k) \in \mathcal{V}} \mathrm{dist} \bigl ( \sigma \bigl ( \Msymb_0(r,k) \bigr ) , z \bigr )^{-1}$ in a polynomial fashion. 
	\end{enumerate}
\end{lemma}
\begin{proof}
	\begin{enumerate}[(i)]
		\item First of all, Assumption~\ref{intro:assumption:eps_mu} means we can invoke \cite[Theorem~1.1]{DeNittis_Lein:adiabatic_periodic_Maxwell_PsiDO:2013} and conclude $\Msymb_{\lambda} \in \SemiHoermeq{1}{1} \bigl ( \mathcal{B} (\domainT , \HperT) \bigr )$. We may construct $R_{\lambda}(z)$ locally order-by-order in $\lambda$. 
		Pick $(r_0,k_0) \in \R^6$ and $z \not\in \sigma \bigl ( \Msymb_0(r_0,k_0) \bigr ) = \sigma \bigl ( \tau(r_0) \, \Mper(k_0) \bigr )$. The continuity of the spectrum in $(r,k)$ \cite[Theorem~1.4]{DeNittis_Lein:adiabatic_periodic_Maxwell_PsiDO:2013} implies the existence of a neighborhood $\mathcal{V} \subseteq \R^6$ of $(r_0,k_0)$ so that 
		\begin{align*}
			\bigl ( R_0(z) \bigr )(r,k) := \bigl ( \Msymb_0(r,k) - z \bigr )^{-1}
		\end{align*}
		exists for all $(r,k) \in \mathcal{V}$. Now we proceed by induction: assume we have found $R^{(N)}(z) := \sum_{n = 0}^N \lambda^n \, R_n(z)$ for which 
		\begin{align}
			R^{(N)}(z) \, \Weyl \, ( \Msymb_{\lambda} - z ) = \id_{\HperT} + \lambda^{N+1} \, E_{N+1}(z) + \order(\lambda^{N+2}) 
			\label{sapt:eqn:local_Moyal_resolvent_E_N}
		\end{align}
		holds on $\mathcal{V}$. Then one can verify from the associativity of $\Weyl$ that on $\mathcal{V}$ also 
		\begin{align*}
			( \Msymb_{\lambda} - z ) \, \Weyl \, R^{(N)}(z) = \id_{\HperT} + \order(\lambda^{N+1}) 
		\end{align*}
		is true and 
		\begin{align*}
			R^{(N+1)}(z) := R^{(N)}(z) - \lambda^{N+1} \, E_{N+1}(z) \, R_0(z) 
		\end{align*}
		satisfies \eqref{sapt:eqn:local_Moyal_resolvent_E_N} up to $\order(\lambda^{N+2})$. 
		\item The validity of the resolvent identity follows directly from $\Weyl$-multiplying $R_{\lambda}(z) - R_{\lambda}(z')$ from the left and right with $(\Msymb_{\lambda} - z)$ and $(\Msymb_{\lambda} - z')$, respectively, as well as (i). $R_{\lambda}(z)^* = R_{\lambda}(\bar{z})$ is a consequence of $\Msymb_{\lambda}^* = \Msymb_{\lambda}$. 
		\item Let $(r_0,k_0)$ and $z \not\in \sigma \bigl ( \Msymb_0(r_0,k_0) \bigr )$ and $\mathcal{V}$ be an open neighborhood where $R_{\lambda}(z)$ exists. All of the subsequent equalities are meant to hold on $\mathcal{V}$. Clearly, we have 
		\begin{align*}
			\bnorm{\bigl ( R_0(z) \bigr )(r,k)}_{\mathcal{B}(\HperT)} = \frac{1}{\mathrm{dist} \bigl ( \sigma \bigl ( \Msymb_0(r,k) \bigr ) , z \bigr )} 
			. 
		\end{align*}
		To estimate the norm of the first-order derivatives, we remark that 
		\begin{align*}
			\partial R_0(z) = - R_0(z) \, \partial \Msymb_0 \, R_0(z) 
		\end{align*}
		where $\partial$ stands for either $\partial_{r_j}$ or $\partial_{k_j}$. Given that $\Msymb_0(r,k) = \tau(r) \, \Mper(k)$ and $\tau(r)$ is assumed to be scalar and invertible (Assumption~\ref{intro:assumption:eps_mu}), we have 
		\begin{align*}
			\partial_{r_j} R_0(z) &= - \bigl ( \partial_{r_j} \ln \tau \bigr ) \; \bigl ( R_0(z) + z \, R_0(z)^2 \bigr ) 
			. 
		\end{align*}
		Thus, the norm can be estimated from above by 
		\begin{align*}
			\bnorm{\bigl ( \partial_{r_j} R_0(z) \bigr )(r,k)}_{\mathcal{B}(\HperT)} \leq C\; \sup_{(r,k) \in \mathcal{V}} \left ( \frac{\mathrm{dist} \bigl ( \sigma \bigl ( \Msymb_0(r,k) \bigr ) , z \bigr )+\sabs{z}}{\mathrm{dist} \bigl ( \sigma \bigl ( \Msymb_0(r,k) \bigr ) , z \bigr )^2} \right )
			. 
		\end{align*}
		Since $\Msymb_0(r,k)$ is linear in $k$, $\partial_{k_j} \Msymb_0(r,k)$ is in fact independent of $k$ and defines a bounded operator on $\HperT$. Thus, we can also bound first-order derivatives with respect to $k$, 
		\begin{align*}
			&\bnorm{\bigl ( \partial_{k_j} R_0(z) \bigr )(r,k)}_{\mathcal{B}(\HperT)} 
			\leq 
			\\
			&\qquad 
			\leq \Bigl ( \sup_{(r,k) \in \mathcal{V}} \bnorm{\partial_{k_j} \Msymb_0(r,k)}_{\mathcal{B}(\HperT)} \Bigr ) \; \sup_{(r,k) \in \mathcal{V}} \frac{1}{\mathrm{dist} \bigl ( \sigma \bigl ( \Msymb_0(r,k) \bigr ) , z \bigr )^2}
			. 
		\end{align*}
		Higher-order derivatives are estimated in a similar fashion by writing $\partial_r^a \partial_k^b R_0(z)$ in terms of $R_0(z)$ and derivatives of $\Msymb_0$. 
		
		The analogous statements for $R_n(z) = - E_n(z) \; R_0(z)$ follow from the first estimate of (iii) for $R_0(z)$ and 
		\begin{align*}
			E_n(z) = \left ( R^{(n-1)}(z) \, \Weyl \, (\Msymb_{\lambda} - z) - \id_{\HperT} \right )_n 
		\end{align*}
		where $( \cdots )_n$ denotes the $n$th-order term of the asymptotic expansion in the brackets. 
		
		Similarly we can estimate the $\mathcal{B}(\HperT,\domainT)$-norms of derivatives of the $R_n(z)$. The only crucial step is to prove the estimates for $R_0(z)$, then we can proceed inductively as outlined above. The arguments in \cite[Section~4.3]{DeNittis_Lein:adiabatic_periodic_Maxwell_PsiDO:2013} imply that the domains of $\Mper(k)$ and $\Msymb_0(r,k)$ coincide, and the respective graph norms are equivalent since 
		\begin{align*}
			C^{-1} \bigl ( 1 + \sabs{k} \bigr )^{-1} \, \snorm{\Psi}_{\Msymb_0(r,k)} \leq \snorm{\Psi}_{\Mper(0)} 
			\leq C \bigl ( 1 + \sabs{k} \bigr ) \, \snorm{\Psi}_{\Msymb_0(r,k)}
		\end{align*}
		where $C$ is independent of $(r,k)$. To see that these two inequalities hold, one has to modify the graph norm estimates for $\Mper(k)$ in the proof of \cite[Proposition~3.3~(i)]{DeNittis_Lein:adiabatic_periodic_Maxwell_PsiDO:2013}. 
		
		We equip $\domainT$ with the graph norm of $\Mper(0)$ and use the upper bound of the above inequality to deduce 
		\begin{align*}
			&\Bnorm{\Mper(0) \, \bigl ( \Msymb_0(r,k) - z \bigr )^{-1}}_{\mathcal{B}(\HperT)} 
			\leq 
			\\
			&\qquad 
			\leq C \, \bigl ( 1 + \sabs{k} \bigr ) \; \sup_{(r,k) \in \mathcal{V}} \left ( \frac{\sabs{z}}{\mathrm{dist} \bigl ( \sigma \bigl ( \Msymb_0(r,k) \bigr ) , z \bigr )} + 1 \right ) 
			. 
		\end{align*}
		Estimates for higher-order derivatives now follow from writing derivatives of the resolvents as products of resolvents and derivatives of $\Msymb_0$. Hence, we obtain the second estimate. This concludes the proof. 
	\end{enumerate}
\end{proof}
\subsubsection*{Proof of Proposition~\ref{sapt:prop:projection}} % (fold)
\label{technical:sapt:projection}
% 
%\begin{proof}[of Proposition~\ref{sapt:prop:projection}]
	The proof is a straightforward modification of that of \cite[Lemma~5.17]{Teufel:adiabatic_perturbation_theory:2003}. 
	
	Without loss of generality, we may assume that $\Omega_{\mathrm{rel}}$ consists of a single, contiguous family of bands. For otherwise, we may repeat the subsequent arguments for each contiguous family and add up the resulting (almost-)projections afterwards.

	Let $(r_0,k_0) \in \R^6$ be arbitrary but fixed. Then the continuity of $\sigma \bigl ( \Msymb_0(r,k) \bigr ) = \sigma \bigl ( \tau(r) \, \Mper(k) \bigr )$ \cite[Theorem~1.4]{DeNittis_Lein:adiabatic_periodic_Maxwell_PsiDO:2013} and Gap Condition~\ref{intro:assumption:gap_condition} ensure the existence of an open neighborhood $\mathcal{V}$ of $(r_0,k_0)$ and a contour $\Lambda$ with the following properties: $\Lambda$ is a positively oriented circle that is symmetric with respect to reflections about the real axis and encloses only $\tau(r) \, \specrel(k) = \bigl \{ \tau(r) \, \omega_n(k) \; \vert \; \omega_n \in \Omega_{\mathrm{rel}} \bigr \}$ for all $(r,k) \in \mathcal{V}$. Moreover, we have the bounds 	
	\begin{align}
		\mathrm{Radius} (\Lambda) \leq C_r 
		, \qquad\quad \inf_{(r,k) \in \mathcal{V}} \mathrm{dist} \left ( \sigma \bigl ( \Msymb_0(r,k) \bigr ) , \Lambda \right ) \geq \frac{c_{\mathrm{g}}}{4} 
		, 
		\label{sapt:eqn:distance_contour_spectrum}
	\end{align}
	where $c_{\mathrm{g}}$ is the gap constant from the Gap Condition~\ref{intro:assumption:gap_condition}. Given that $\tau(r) \, \specrel(k)$ consists of eigenvalues $\omega_n(k)$ which are scaled by $\tau(r)$, and that $\tau$ is bounded away from $0$ and $+\infty$ (Assumption~\ref{intro:assumption:eps_mu}), the “width” of $\tau(r) \, \specrel(k)$ is bounded from above and thus, the constant $C_r \geq \mathrm{Radius}(\Lambda)$ can be chosen independently of $(r_0,k_0)$.
	
	By equivariance of $\Msymb_0$, 
	\begin{align*}
		\Msymb_0(r,k - \gamma^*) &= \e^{+ \ii \gamma^* \cdot \hat{y}} \, \Msymb_0(r,k) \, \e^{- \ii \gamma^* \cdot \hat{y}} 
		, 
		&&
		\forall (r,k) \in \R^6 
		, \; 
		\gamma^* \in \Gamma^* 
		, 
	\end{align*}
	we may choose $\mathcal{V}$ to be $\Gamma^*$-periodic, \ie if $(r,k) \in \mathcal{V}$, then also $(r,k-\gamma^*) \in \mathcal{V}$.
	
	Condition~\eqref{sapt:eqn:distance_contour_spectrum} on the contour $\Lambda$ also ensures that the local Moyal resolvent exists on $\mathcal{V}$ for all $z \in \Lambda$ (\cf Lemma~\ref{sapt:lem:existence_Moyal_resolvent}), and we define $\pi_{\lambda}(r,k)$ on $\mathcal{V}$ as a formal series $\sum_{n = 0}^{\infty} \lambda^n \, \pi_n$ where
	\begin{align}
		\pi_n(r,k) = \frac{\ii}{2\pi} \int_{\Lambda} \dd z \, \bigl ( R_n(z) \bigr )(r,k) 
		. 
		\label{sapt:eqn:pi_n_Cauchy}
	\end{align}
	$\bigl ( \pi_{\lambda} \Weyl \pi_{\lambda} \bigr )_n = \pi_n$ and $\pi_n^* = \pi_n$ follow from Lemma~\ref{sapt:lem:existence_Moyal_resolvent}~(ii) and our choice of contour. 
	
	Choosing an open covering of $\R^6$ so that on each neighborhood $\mathcal{V}$ we can choose a single, $(r,k)$-\emph{in}dependent contour $\Lambda$ with the aforementioned properties, we can construct the formal symbol $\pi_{\lambda}(r,k)$ for all $(r,k) \in \R^6$. On the overlaps, the locally constructed projections need to agree order-by-order. To show that the $\pi_n$ are in the correct symbol classes, we combine the local estimates of the $\norm{\cdot}_{\mathcal{B}(\HperT)}$- and $\norm{\cdot}_{\mathcal{B}(\HperT,\domainT)}$-norms, 
	\begin{align*}
		\bnorm{\partial_r^a \partial_k^b \pi_n(r,k)} \leq 2 \pi C_r \, \bnorm{\partial_r^a \partial_k^b \bigl ( R_n(z) \bigr )(r,k)} 
		, 
		&&
		a , b \in \N_0^3 
		, 
	\end{align*}
	with the local Moyal resolvent estimates from Lemma~\ref{sapt:lem:existence_Moyal_resolvent}~(iii). Hence, we also get 
	\begin{align*}
		\abs{z} \leq C_r + \sup_{r \in \R^3} \, \abs{\tau(r)} \; \; \sup_{k \in \BZ} \, \max \, \abs{\specrel(k)} < \infty 
		. 
	\end{align*}
	Lastly, all contours involved satisfy \eqref{sapt:eqn:distance_contour_spectrum}, and hence for all $n \in \N_0$ and $a , b \in \N_0^3$ we can find constants $C_{1} , C_{2} > 0$ so that 
	\begin{align*}
		\bnorm{\partial_r^a \partial_k^b \pi_n(r,k)}_{\mathcal{B}(\HperT)} &\leq C_1
		, 
		\qquad\quad
		\bnorm{\partial_r^a \partial_k^b \pi_n(r,k)}_{\mathcal{B}(\HperT,\domainT)} \leq C_2 \, \bigl ( 1 + \sabs{k} \bigr ) 
		, 
	\end{align*}
	hold for \emph{all} $(r,k) \in \R^6$. In other words, we have shown 
	\begin{align*}
		\pi_n \in S^0_{0,\mathrm{eq}} \bigl ( \mathcal{B}(\HperT) \bigr ) \cap S^1_{0,\mathrm{eq}} \bigl ( \mathcal{B}(\HperT , \domainT) \bigr ) 
		, 
	\end{align*}
	and there exists a resummation (using an appropriate choice of $N(\lambda)$)
	\begin{align*}
		\pi_{\lambda} = \sum_{n = 0}^{N(\lambda)} \lambda^n \, \pi_n 
		\in \SemiHoermeq{0}{0} \bigl ( \mathcal{B}(\HperT) \bigr ) \cap \SemiHoermeq{1}{0} \bigl ( \mathcal{B}(\HperT , \domainT) \bigr )
		.
	\end{align*}
 Clearly, by construction $\pi_{\lambda}$ is an orthogonal Moyal projection up to $\order(\lambda^{\infty})$
	\begin{align*}
		\pi_{\lambda}^* &= \pi_{\lambda} 
		, 
		\\
		\pi_{\lambda} \Weyl \pi_{\lambda} &= \pi_{\lambda} + \order(\lambda^{\infty})
		. 
	\end{align*}
	Moreover, as the eigenfunctions $\varphi_n(k)$ of $\Mper(k)$ are also the eigenfunctions of the principal symbol $\Msymb_0(r,k)$, we deduce that in fact, 
	\begin{align*}
		\pi_0(r,k) = \sum_{n \in \Index} \sopro{\varphi_n(k)}{\varphi_n(k)} 
	\end{align*}
	coincides with the projection constructed in Lemma~\ref{setup:lem:existence_pi_0} and is thus independent of $r$. 
	
	$\pi_{\lambda} \in \SemiHoermeq{0}{0} \bigl ( \mathcal{B}(\HperT) \bigr )$ means the Caldéron-Vaillancourt theorem applies and hence, the $\Psi$DO $\Op_{\lambda}(\pi_{\lambda})$ defines a bounded operator on $\Zak \Hper$ \cite[Theorem~B.5]{Teufel:adiabatic_perturbation_theory:2003} which is also selfadjoint if we use the scalar product $\scpro{\cdot \,}{\cdot}_{\Zak \Hper}$. Thus, we can use functional calculus to define the spectral projection 
	\begin{align*}
		\Pi_{\lambda} := \frac{\ii}{2\pi} \int_{\sabs{z - \nicefrac{1}{2}} = 1} \dd z \, \bigl ( \Op_{\lambda}(\pi_{\lambda}) - z \bigr )^{-1} 
		. 
	\end{align*}
	%
%\end{proof}
%
% subsubsection Proof of Proposition~\ref{sapt:prop:projection} (end)

\subsubsection*{Proof of Proposition~\ref{sapt:prop:unitary}} % (fold)
\label{technical:sapt:unitary}
We modify the arguments in the proof of \cite[Proposition~5.18]{Teufel:adiabatic_perturbation_theory:2003}. 

On the level of symbols, if $u_{\lambda}$ exists, it needs to have the following defining properties: 
\begin{align}
	u_{\lambda} \Weyl u_{\lambda}^* &= \id_{\HperT} + \order(\lambda^{\infty}) 
	, 
	\qquad \qquad 
	u_{\lambda} \Weyl u_{\lambda}^* = \id_{\HperT} + \order(\lambda^{\infty}) 
	\notag 
	\\
	u_{\lambda} \Weyl \pi_{\lambda} \Weyl u_{\lambda} &= \piref + \order(\lambda^{\infty})
	\label{sapt:eqn:defining_relations_u}
\end{align}
We will construct $u_{\lambda} \asymp \sum_{n = 0}^{\infty} \lambda^n \, u_n$ order-by-order using recursion relations. To be able to find an appropriate $u_0$, we rely on the assumption that the Bloch bundle $\blochb(\pi_0)$ is trivial. Since in this case topological and analytical triviality are equivalent by the Oka principle, we can find a family of $\scpro{\cdot \,}{\cdot}_{\HperT}$-orthogonal vectors $\{ \psi_j(k) \}_{j = 1}^{\abs{\Index}}$ so that 
\begin{align*}
	\mathrm{span} \{ \psi_j(k) \}_{j = 1}^{\abs{\Index}} = \ran \pi_0(k) 
	, 
\end{align*}
and each of the $k \mapsto \psi_j(k)$ is analytic and can be extended to a left-covariant map on $\R^3$. 
Moreover, we can arbitrarily complete 
\begin{align*}
	u_0(r,k) := \sum_{j = 1}^{\abs{\Index}} \sopro{\chi_j}{\psi_j(k)} + u_0^{\perp}(k)
\end{align*}
by $u_0^{\perp}$ to a right-covariant unitary on $\HperT$ which is also an element of $S^0_0 \bigl ( \mathcal{B}(\HperT) \bigr )$ and depends trivially on $r$: Kuiper's theorem \cite[Corollary~(1)]{Kuiper:homotopy_type_unitary_group:1965} ensures that the principal bundle associated to $\blochb(\pi_0^{\perp})$ is \emph{topologically} trivial. But since topological and analytic triviality are one and the same in this context \cite[Satz I, p.~268]{Grauert:analytische_Faserungen:1958}, and triviality of a principal bundle means the existence of a section, we can find an \emph{analytic}, right-covariant map $k \mapsto u_0^{\perp}(k)$ with the desired properties.
Moreover, by its very definition $u_0$ is a Moyal unitary, 
\begin{align*}
	u_0 \Weyl u_0^* &= u_0 \, u_0^* = \id_{\HperT} 
	,
	\qquad \qquad 
	u_0^* \Weyl u_0 = u_0^* \, u_0 = \id_{\HperT} 
	, 
\end{align*}
which intertwines $\pi_{\lambda}$ with $\piref := u_0(k) \, \pi_0(k) \, u_0^*(k) = \sum_{j = 1}^{\abs{\Index}} \sopro{\chi_j}{\chi_j}$ to first order, 
\begin{align*}
	u_0 \Weyl \pi_{\lambda} \Weyl u_0^* &= u_0 \, \pi_0 \, u_0^* + \order(\lambda) 
	= \piref + \order(\lambda) 
	. 
\end{align*}
Now assume we have found $u^{(N)} := \sum_{n = 0}^N \lambda^n \, u_n$ which satisfies 
\begin{align}
	u^{(N)} \Weyl {u^{(N)}}^* - \id_{\HperT} &=: \lambda^{N+1} \, A_{N+1} + \order(\lambda^{N+2}) 
	, 
	\label{sapt:eqn:amlost_unitarity}
	\\
	{u^{(N)}}^* \Weyl u^{(N)} - \id_{\HperT} &= \order(\lambda^{N+1}) 
	, 
	\notag 
	\\
	u^{(N)} \Weyl \pi_{\lambda} \Weyl {u^{(N)}}^* - \piref &=: \lambda^{N+1} \, B_{N+1} + \order(\lambda^{N+2}) 
	, 
	\label{sapt:eqn:amlost_intertwining}
\end{align}
where the right-hand sides are elements of $S^0_{0,\mathrm{per}} \bigl ( \mathcal{B}(\HperT) \bigr )$. Then a short explicit computation yields that 
\begin{align*}
	u^{(N+1)} := u^{(N)} + \lambda^{N+1} \, \left ( - \tfrac{1}{2} A_{N+1} + [\piref , B_{N+1}] \right ) \, u_0 
\end{align*}
is a right-covariant $S^0_0 \bigl ( \mathcal{B}(\HperT) \bigr )$ symbol which satisfies \eqref{sapt:eqn:amlost_unitarity} and \eqref{sapt:eqn:amlost_intertwining} up to $\order(\lambda^{N+2})$. 

Let us denote a resummation $u_{\lambda} \in \SemiHoerm{0}{0} \bigl ( \mathcal{B}(\HperT) \bigr )$ of $\sum_{n = 0}^{\infty} \lambda^n \, u_n$ with the same letter. We can now repeat the arguments outlined in Step~II of the proof of \cite[Theorem~3.12]{Teufel:adiabatic_perturbation_theory:2003} to construct a true unitary $U_{\lambda} = \Op_{\lambda}(u_{\lambda}) + \order_{\norm{\cdot}}(\lambda^{\infty})$ from the almost-unitary $\Op_{\lambda}(u_{\lambda})$ via the Nagy formula. This concludes the proof. 
% subsubsection Proof of Proposition (end)

\subsubsection*{Proof of Proposition~\ref{sapt:prop:effective_Maxwellian}} % (fold)
\label{technical:sapt:eff_Maxwell}
The strategy of this proof follows that of \cite[Proposition~5.19]{Teufel:adiabatic_perturbation_theory:2003}. The assumptions in the claim assure that $\Pi_{\lambda} = \Op_{\lambda}(\pi_{\lambda}) + \order_{\norm{\cdot}}(\lambda^{\infty})$ and $U_{\lambda} = \Op_{\lambda}(u_{\lambda}) + \order_{\norm{\cdot}}(\lambda^{\infty})$ exist.

Since all we are interested in are resummations of $\Msymb_{\eff}$, we can reshuffle the order of the projections and the unitaries using the intertwining relation~\eqref{sapt:eqn:defining_relations_u}, 
\begin{align*}
	\Msymb_{\eff} = \piref \, u_{\lambda} \Weyl \Msymb_{\lambda} \Weyl u_{\lambda}^* \, \piref
	=\, u_{\lambda} \Weyl \pi_{\lambda} \Weyl \Msymb_{\lambda} \Weyl \pi_{\lambda} \Weyl u_{\lambda}^* 
	. 
\end{align*}
At first glance, we can only deduce $\Msymb_{\eff} \in \Hoermper{2}{0} \bigl ( \mathcal{B}(\HperT) \bigr )$ from the Moyal composition property of Hörmander symbols. However, a more careful look also confirms $\Msymb_{\eff \, n} \in\Hoermper{0}{0} \bigl ( \mathcal{B}(\HperT) \bigr )$: first of all, $\Msymb_1$ can also be seen as an element of $\Hoermeq{0}{0} \bigl ( \mathcal{B}(\HperT) \bigr )$, and so all terms stemming from the asymptotic expansion of $u_{\lambda} \Weyl \pi_{\lambda} \Weyl \Msymb_1 \Weyl \pi_{\lambda} \Weyl u_{\lambda}^*$ are elements of $\Hoermper{0}{0} \bigl ( \mathcal{B}(\HperT) \bigr )$. Moreover, to estimate the terms involving $\Msymb_0 \Weyl \pi_{\lambda}$, we can locally write $\pi_{\lambda}$ as Cauchy integral with respect to the local Moyal resolvent and use Lemma~\ref{sapt:lem:existence_Moyal_resolvent}~(iii) to estimate the norms of the derivatives, 
\begin{align*}
	\Bnorm{\partial_r^a \partial_k^b \bigl ( \Msymb_0 \Weyl R_{\lambda}(z) \bigr )(r,k)}_{\mathcal{B}(\HperT)} \leq C_{a b} 
	, 
	&& 
	a , b \in \N_0^3 
	. 
\end{align*}
The arguments in the proof of Proposition~\ref{sapt:prop:projection} which establish $\pi_{\lambda} \in \SemiHoermeq{1}{0} \bigl ( \mathcal{B}(\HperT,\domainT) \bigr )$ ensure we can choose $C_{a b}$ independently of $z$ and $(r,k)$. Thus, the product $\Msymb_0 \Weyl \pi_{\lambda} \in \SemiHoermeq{0}{0} \bigl ( \mathcal{B}(\HperT) \bigr )$ is in the correct symbol class and any resummation of the formal symbol $\Msymb_{\eff}$ in $\SemiHoermper{0}{0} \bigl ( \mathcal{B}(\HperT) \bigr )$ will define a \emph{bounded} $\Psi$DO \cite[Proposition~B.5]{Teufel:adiabatic_perturbation_theory:2003}. 

Equation~\eqref{sapt:eqn:effective_dynamics_approximate} is established via a standard Duhamel argument where the crucial ingredient is 
\begin{align*}
	\left ( M_{\lambda} - U_{\lambda}^{-1} \, \Op_{\lambda}(\Msymb_{\eff}) \, U_{\lambda} \right ) \Pi_{\lambda} = \order_{\norm{\cdot}}(\lambda^{\infty})
	. 
\end{align*}
%
% subsubsection Proof of Proposition~\ref{sapt:prop:effective_Maxwellian} (end)

\subsection{Auxiliary results for Section~\ref{physical}} % (fold)
\label{technical:physical}
\begin{proof}[Lemma~\ref{real_eff:lem:complex_conjugation_PsiDO}]
	From the (formal) definition of Weyl quantization (\cf Section~\ref{setup:PsiDOs}), 
	we deduce that 
	\begin{align*}
		\bigl ( C^{\Zak} \, \Op(f) \, &C^{\Zak} \varphi \bigr )(k) = \overline{\bigl ( \Op(f) \, C^{\Zak} \varphi \bigr )(-k)}
		\\
		&= \frac{1}{(2\pi)^d} \int_{\R^d} \dd k' \int_{\R^d} \dd r' \, \e^{+ \ii (k' + k) \cdot r'} \, C \, f \bigl ( \eps r' , \tfrac{1}{2} (k' - k) \bigr ) \, C \varphi(-k') 
		\\
		&= \frac{1}{(2\pi)^d} \int_{\R^d} \dd k'' \int_{\R^d} \dd r' \, \e^{- \ii (k'' - k) \cdot r'} \, (\mathfrak{c} f) \bigl ( \eps r , \tfrac{1}{2}(k'' + k) \bigr ) \, \varphi(k'')
		\\
		&= \bigl ( \Op(\mathfrak{c} f) \varphi \bigr )(k) 
		. 
	\end{align*}
\end{proof}
\begin{lemma}\label{technical:lem:c_pi_pi}
	Assume $(\eps_{\lambda},\mu_{\lambda})$ are real. Then $\specrel(-k) = - \specrel(k)$ implies $\mathfrak{c} \pi_{\lambda} = \pi_{\lambda}$. 
\end{lemma}
\begin{proof}
	The symbol of the Maxwell operator has the same symmetry as its quantization, 
	\begin{align}
		\bigl ( \mathfrak{c} \Msymb_{\lambda} \bigr )(r,k) &= - \Msymb_{\lambda}(r,k) 
		. 
		\label{technical:eqn:conjugation_Maxwell_symbol}
	\end{align}
	Consequently, the local Moyal resolvent $R_{\lambda}(z)$ satisfies 
	\begin{align}
		C \, \bigl ( R_{\lambda}(z) \bigr )(r,-k) \, C = - \bigl ( R_{\lambda}(-\bar{z}) \bigr )(r,k)
	\label{real_eff:eqn:complex_conjugate_Moyal_resolvent}
	\end{align}
	Let $\Omega_{\mathrm{rel}} = \bigcup_{\alpha = 1}^N \Omega_{\alpha}$ where the $\Omega_{\alpha} = \{ \omega_n \}_{n\in\Index_{\alpha}}$ are the families of contiguous bands from Gap Condition~\ref{intro:assumption:gap_condition}. Define $\sigma_{\alpha}(k) = \{\omega_n(k)\}_{n \in \Index_{\alpha}}$ to be the pointwise spectrum associated to $\Omega_{\alpha}$. Due to the Gap Condition we can invoke Proposition~\ref{sapt:prop:projection} for each $\alpha$ separately to obtain 
	\begin{align*}
		\pi_{\alpha \, \lambda}(r,k) \asymp \frac{\ii}{2 \pi} \int_{\Lambda_{\alpha}} \bigl ( R_{\lambda}(z) \bigr )(r,k)
	\end{align*}
	locally in an open neighborhood of each point $(r_0,k_0)$ as a Cauchy integral with respect to the local Moyal resolvent from Lemma~\ref{sapt:lem:existence_Moyal_resolvent}.
	There are two cases we need to consider: either $\sigma_{\alpha}(-k) = - \sigma_{\alpha}(k)$ ($\sigma_{\alpha}$ includes the ground state bands) or $\sigma_{\alpha}(-k) = - \sigma_{\beta}(-k)$ for some $\beta \neq \alpha$. 
	Since we have excluded ground state bands (Gap Condition~\ref{intro:assumption:gap_condition}~(iii)), we only need to consider the second case. 
	
	So let $\alpha$ and $\beta \neq \alpha$ be the indices such that $\sigma_{\alpha}(-k) = - \sigma_{\beta}(k)$. 
	Now if $\Lambda_{\alpha}$ is the (locally fixed) contour which encloses only $\tau(r) \, \sigma_{\alpha}(-k)$ for all points $(r,k)$ in an open neighborhood of $(r_0,-k_0)$, then $- \Lambda_{\alpha}$ is the (locally fixed) contour for all points $(r,k)$ in an open neighborhood of $(r_0,k_0)$ which encloses only $\tau(r) \, \sigma_{\beta}(k)$. 
	Without loss of generality, we may assume $\Lambda_{\alpha}$ has the symmetry properties enumerated in the proof of Proposition~\ref{sapt:prop:projection}. 
	Then using \eqref{technical:eqn:conjugation_Maxwell_symbol}, we deduce 
	\begin{align*}
		\bigl ( \mathfrak{c} \pi_{\alpha \, \lambda} \bigr )(r,k) &= - \frac{\ii}{2\pi} \int_{-\Lambda_{\alpha}} \dd \bar{z} \, C \, \bigl ( R_{\lambda}(z) \bigr )(r,-k) \, C 
		\\
		&= + \frac{\ii}{2\pi} \int_{-\Lambda_{\alpha}} \dd \bar{z} \, \bigl ( R_{\lambda}(-\bar{z}) \bigr )(r,k) 
		\\
		&= -\frac{\ii}{2\pi} \int_{ \Lambda_{\alpha}} \dd \bar{z} \, \bigl ( R_{\lambda}(\bar{z}) \bigr )(r,k) \\
		&= \pi_{\beta \, \lambda}(r,k) .
	\end{align*}
	Once we sum over $\alpha$, we finish the proof of the claim. 
\end{proof}
%
% subsection Auxiliary results for Section~\ref{physical} (end)

\subsection{Auxiliary results for Section~\ref{twin_bands}} % (fold)
\label{technical:ray_optics}
\begin{proof}[Corollary~\ref{real_eff:cor:simplest_effective_Maxwellian}]
	The assumptions of the Corollary include those of Theorem~\ref{real_eff:thm:eff_edyn}, and hence we can use equation~\eqref{sapt:eqn:Meff_physical_rep} to compute the symbol of the effective Maxwellian. Moreover, the assumption that all the single band Bloch bundles $\blochb \bigl ( \sopro{\varphi_n}{\varphi_n} \bigr )$ are trivial means the Bloch functions themselves are an analytic frame of the Bloch bundle $\blochb(\pi_0)$. Thus, we can choose 
	\begin{align*}
		u_0(k) &= \sum_{n \in \Index} \sopro{\chi_n}{\varphi_n(k)} + u_0^{\perp}(k) 
	\end{align*}
	as a principal symbol for the Moyal unitary where $u_0^{\perp}$ is the symbol constructed in the proof of Proposition~\ref{sapt:prop:unitary}. 
	
	Then, we obtain the leading-order term of $\Msymb_{\eff}$ by replacing $\Weyl$ with the pointwise product and using the explicit expression for $\Msymb_{\lambda}$ given in \eqref{setup:eqn:M_lambda_PsiDO}, 
	\begin{align*}
		\Msymb_{\eff,0} &= \piref \, u_0 \, \Msymb_0 \, u_0^* \, \piref 
		= \sum_{n \in \Index} \tau \, \omega_n \, \sopro{\chi_n}{\chi_n}
		, 
	\end{align*}
	thereby confirming \eqref{intro:eqn:Meff_0_simple}. To get a handle on the subprincipal symbol, we start with \cite[equation~(3.35)]{Teufel:adiabatic_perturbation_theory:2003}, introduce $\widetilde{\Msymb}_0 := u_0 \, \Msymb_0 \, u_0^*$ and group the terms: 
	\begin{align*}
		\Msymb_{\eff,1} &= \tfrac{1}{\lambda} \, \piref \, \Bigl ( u_{\lambda} \Weyl \Msymb_{\lambda} - \widetilde{\Msymb}_0 \Weyl u_{\lambda} \Bigr ) \, u_0^* \, \piref + \order(\lambda) 
		\\
		&= \piref \, \bigl [ u_1 \, u_0^* , \widetilde{\Msymb}_0 \bigr ] \, \piref + \piref \, u_0 \, \Msymb_1 \, u_0^* \, \piref 
		\, + \\
		&\quad 
		+ \piref \, \tfrac{\ii}{2} \, \Bigl ( \bigl \{ \widetilde{\Msymb}_0 , u_0 \bigr \} - \bigl \{ u_0 , \Msymb_0 \bigr \} \Bigr ) \, u_0^* \, \piref + \order(\lambda) 
	\end{align*}
	By construction (\cf proof of Proposition~\ref{sapt:prop:unitary}), the subprincipal symbol
	\begin{align*}
		u_1 = \bigl ( - \tfrac{1}{2} A_1 + [\piref , B_1] \bigr ) \, u_0
	\end{align*}
	is written as the sum of two terms where 
	\begin{align*}
		\lambda \, A_1 + \order(\lambda^2) = u_0 \Weyl u_0^* - \id_{\HperT} = 0 
	\end{align*}
	is the unitarity defect and $B_1$ the intertwining defect given by equation~\eqref{sapt:eqn:amlost_intertwining}. Hence, the first term vanishes, 
	\begin{align*}
		\piref \, \bigl [ u_1 \, u_0^* \, , \, \widetilde{\Msymb}_0 \bigr ] \, \piref &= \Bigl [ \piref \, [\piref , B_1] \, \piref \, , \, \piref \, \widetilde{\Msymb}_0 \, \piref \Bigr ] 
		= 0 
		. 
	\end{align*}
	The second term involves the scalar product  
	\begin{align*}
		\bscpro{\varphi_n}{\Msymb_1 \varphi_j}_{\HperT}
		&= - \tau \, \frac{\ii}{2} \scpro{\left (
		\begin{matrix}
			\varphi_n^E \\
			\varphi_n^H \\
		\end{matrix}
		\right )}{\left (
		\begin{matrix}
			0 & \bigl ( \nabla_r \ln \frac{\tau_{\eps}}{\tau_{\mu}} \bigr )^{\times} \\
			\bigl ( \nabla_r \ln \frac{\tau_{\eps}}{\tau_{\mu}} \bigr )^{\times} & 0 
		\end{matrix}
		\right ) \left (
		\begin{matrix}
			\varphi_j^E \\
			\varphi_j^H \\
		\end{matrix}
		\right )}_{L^2(\T^3,\C^6)} 
		\\
		% &= - \tau \, \frac{\ii}{2} 
		% \Bigl ( \scpro{\varphi_n^E}{\nabla_r \ln \tfrac{\tau_{\eps}}{\tau_{\mu}} \times \varphi_j^H}_{L^2(\T^3,\C^3)} 
		% + \Bigr . \\
		% &\qquad \qquad \Bigl . 
		% + \scpro{\varphi_n^H}{\nabla_r \ln \tfrac{\tau_{\eps}}{\tau_{\mu}} \times \varphi_j^E}_{L^2(\T^3,\C^3)} \Bigr )
		% \\
		&= + \tau \, \frac{\ii}{2} \sum_{n , j \in \mathcal{I}} 
		\nabla_r \ln \tfrac{\tau_{\eps}}{\tau_{\mu}} \cdot 
		\bigl ( \Poynt_{nj} - \overline{\Poynt_{jn}} \bigr ) 
		, 
	\end{align*}
	and the Poynting tensor $\Poynt_{nj}$ as given by \eqref{intro:eqn:Poynting_vector}, and hence, we obtain an expression for the second term, 
	\begin{align*}
		\piref \, u_0 \, \Msymb_1 \, u_0^* \, \piref = \ii \, \frac{\tau}{2} \sum_{n , j \in \mathcal{I}} 
		\nabla_r \ln \tfrac{\tau_{\eps}}{\tau_{\mu}} \cdot 
		\bigl ( \Poynt_{nj} - \overline{\Poynt_{jn}} \bigr ) 
		\, \sopro{\chi_n}{\chi_j}
		. 
	\end{align*}
	The computation of the last term simplifies because $u_0$ depends only on $k$ and the $r$ dependence of $\Msymb_0$ and $\widetilde{\Msymb}_0$ lies with the scalar factor $\tau$: 
	\begin{align*}
		\frac{\ii}{2} \, \piref \, \Bigl ( \bigl \{ &\widetilde{\Msymb}_0 , u_0 \bigr \} - \bigl \{ u_0 , \Msymb_0 \bigr \} \Bigr ) \, u_0^* \, \piref 
		= 
		\\
		&
		= - \frac{\ii}{2} \, \sum_{l = 1}^3 \partial_{r_l} \tau \; \piref \, \Bigl [ \partial_{k_l} u_0 \, u_0^* \; , \, u_0 \, \Mper(\, \cdot \,) \, u_0^* \Bigr ]_+ \, \piref
		\\
		&= - \sum_{l = 1}^3 \sum_{n , j \in \Index} \partial_{r_l} \tau \; \tfrac{1}{2} ( \omega_n + \omega_j ) \; \ii \bscpro{\partial_{k_l} \varphi_n}{\varphi_j}_{\HperT} \, \sopro{\chi_n}{\chi_j}
		\\
		&= + \sum_{n , j \in \Index} \tfrac{1}{2} ( \omega_n + \omega_j ) \, \nabla_r \tau \cdot \mathcal{A}_{nj} \, \sopro{\chi_n}{\chi_j}
	\end{align*}
	Putting all these terms together yields equation~\eqref{intro:eqn:Meff_1_simple}. 
\end{proof}
%
% subsection proof_of_theorem_ray_optics:thm:single_band (end)
% section Technical proofs (end)